\documentclass[numbib]{imaiai}

\jno{iarxxx}
\received{XX XX XXXX}
\revised{XX XX XXXX}
\accepted{XX XX XXXX}

\usepackage[utf8]{inputenc}
\usepackage{amssymb}
\usepackage{amsmath}
\usepackage{graphicx}
\usepackage{enumitem}
\usepackage{amsfonts}
\usepackage{setspace}
\usepackage[usenames,dvipsnames]{xcolor}
\usepackage[colorlinks=true,linkcolor=blue,citecolor=blue,urlcolor=blue]{hyperref}
\usepackage{tikz}
\usepackage{tikz-3dplot}
\usepackage{pgfplots}

 \usepackage{cleveref}

\usetikzlibrary{calc,intersections,positioning,angles,quotes}
\usetikzlibrary{3d,shapes}
\pgfplotsset{compat=1.11}

\newcommand\bR{\mathbb{R}}

\newcommand\bRp{\bR_{+}}

\newcommand\bC{\mathbb{C}}

\newcommand\sL{\mathcal{L}}
\newcommand\sM{\mathcal{M}}
\newcommand\sH{\mathcal{H}}
\newcommand\sA{\mathcal{A}}
\newcommand\sC{\mathcal{C}}
\newcommand\sE{\mathcal{E}}

\newcommand\sT{\mathcal{T}}
\newcommand\mM{M}

\newcommand\cl[1]{\overline{#1}}

\newcommand{\tconv}{\mathrm{conv}}

\newcommand{\tspan}{\mathrm{span}}

\newcommand{\tker}{\mathrm{ker}}

 \newcommand{\supp}{\mathtt{supp}}
  \newcommand{\eig}{\mathtt{eig}}
    \newcommand{\rank}{\mathtt{rank}}
 \newcommand{\dom}{\mathtt{dom}}
  \newcommand{\diag}{\mathtt{diag}}

 \newcommand{\vol}{\mathrm{vol}}
 
 \newcommand{\ls}{\langle}
 \newcommand{\rs}{\rangle}
   \newcommand{\re}{\mathcal{R}e}

 \newcommand{\llbracket}{[}
 \newcommand{\rrbracket}{]}

\newcommand{\defin}{:=}
\newtheorem{fact}{Fact}[section]

\newcommand{\ie}{{\em i.e.,~}}

\newcommand{\eg}{{\em e.g.,~}}

\newcommand{\specialcell}[2][c]{%
  \begin{tabular}[#1]{@{}l@{}}#2\end{tabular}}

\AtBeginDocument{
  \label{CorrectFirstPageLabel}
  \def\fpage{\pageref{CorrectFirstPageLabel}}
}
  \setcitestyle{numbers, open={[},close={]}}

\DeclareMathAlphabet{\mathcal}{OMS}{cmsy}{m}{n}
\begin{document}

\title{A theory of optimal convex regularization for low-dimensional recovery}
 \author{{\sc Yann  Traonmilin$^{1,*}$, Rémi Gribonval$^2$ and Samuel Vaiter$^3$}\\
$^1$Univ. Bordeaux, Bordeaux INP, CNRS,  IMB, UMR 5251,F-33400 Talence, France.\\
        $^2$Univ Lyon, ENS de Lyon, UCBL, CNRS, Inria, LIP, F-69342 Lyon, France.\\
$^3$CNRS, Université Côte d'Azur, LJAD, Nice, France.\\
$^*$ Contact author: \email{yann.traonmilin@math.u-bordeaux.fr}
}

\shortauthorlist{Y. Traonmilin, R. Gribonval and S. Vaiter}

\maketitle

\begin{abstract}
{We consider the problem of recovering elements of a low-dimensional model from under-determined linear measurements. To perform recovery, we consider the minimization of a convex regularizer subject to a data fit constraint. Given a model, we ask ourselves what is the ``best'' convex regularizer to perform its recovery. To answer this question, we define an optimal regularizer as a function that maximizes a compliance measure with respect to the model. We introduce and study several notions of compliance. We give analytical expressions for compliance measures based on the best-known recovery guarantees with the restricted isometry property. These expressions permit to show the optimality of the $\ell^1$-norm for sparse recovery and of the nuclear norm for low-rank matrix recovery for these compliance measures. We also investigate the construction of an optimal convex regularizer using the examples of sparsity in levels and of  sparse plus low-rank models.}
{inverse problems, convex regularization, low dimensional modeling,  sparse recovery, low rank matrix recovery}
\end{abstract}

\section{Introduction}
In a finite-dimensional  Hilbert space $\sH$ (with associated inner product $\ls \cdot, \cdot \rs$, and norm $\|\cdot\|_\sH$), we consider the observation model:
\begin{equation}
 y = Mx_0
\end{equation}
where $y$ is an $m$-dimensional vector of measurements, $M$ is an under-determined linear operator (from $\sH = \bC^n$, or $\bR^n$, to $\bC^m$), and $x_0 \in \sH$ is the unknown vector we want to recover. The problem of recovering $x_0$ from $y$ is typically ill-posed. It is thus necessary to use prior information on $x_0$ to recover it with a guarantee of success.

In this work, we suppose that $x_0$ belongs to a low-dimensional cone $\Sigma$ (a positively homogeneous set, \ie for every $x \in \Sigma$ and $\lambda \geq 0$, $\lambda x \in \Sigma$) that models known properties of the unknown.  Examples of such models include sparse as well as low-rank models and many of their generalizations. Note that in these examples the models belong to the slightly less general class of models that are (finite or infinite) unions of subspaces (homogeneous sets).

To recover $x_0$, a classical method is to solve the constrained minimization problem
\begin{equation} \label{eq:minimization1}
 x^* \in \arg \min_{Mx=y} R(x)
\end{equation}
where $R$ is a function -- the regularizer -- that aims to enforce some structure on the solution $x^*$.

Many works~\cite{Donoho_2006,Candes_2006b,Recht_2010,Candes_2010}  give practical regularizers ensuring that $x^* = x_0$ for several low-dimensional models (in particular sparse and low-rank models, see \cite{Foucart_2013} for a most complete review of these results).
A practical regularizer  is a function that permits the effective calculation of $x^*$. Without computational constraint, the best possible regularizer  would be $R = \iota_\Sigma$: the characteristic function of $\Sigma$ defined by $\iota_\Sigma(x) = 0$ if $x \in \Sigma$, $\iota_\Sigma(x) = +\infty$ otherwise (see \Cref{sec:comp_general} for a review of this fact). Unfortunately, this function is generally not convex (unless $\Sigma$ itself is a convex set) and can lead to an intractable optimization problem in general, even though recent works show that using $R = \iota_\Sigma$ and a dedicated minimization technique is a possible route for certain particular low-dimensional models that can be smoothly embedded in $\bR^n$ \cite{Chi_2019,Traonmilin_2020inverse,Traonmilin_2020b}.

In this work, we focus on continuous convex regularizers that guarantee the existence of a minimizer $x^*$ and the existence of practical optimization algorithms to perform minimization~\eqref{eq:minimization1} such as the primal-dual method~\cite{Chambolle_2011} (provided their proximity operators can be calculated). Note that convexity in itself is not sufficient to guarantee the practical feasibility of minimization $\eqref{eq:minimization1}$ ($R(x)$ could be $NP$-hard to calculate, \eg the nuclear norm for tensors \cite{Friedland_2018}, and/or the proximal operator of $R$ could be $NP$-hard to compute).

Under conditions on the measurement operator $M$ that typically involve the number of measurements and its structure (\eg random for compressed sensing), the fact that $x_0 \in \Sigma$ permits to give recovery guarantees when the convex regularizer  $R$ is well-chosen. For example, when $\Sigma= \Sigma_k$ is the set of $k$-sparse vectors in $\bR^n$ and $R(\cdot) = \|\cdot\|_1$ ($\ell^1$-norm), or when $\Sigma= \Sigma_r$ is the set of matrices of rank lower than $r$ in $\bR^{p\times p}$ and $R(\cdot) = \|\cdot\|_*$ (nuclear norm), $x_0$ can be recovered as long as the number of measurements is of the order of the dimension of the model (up to some log factors)~: $m \geq O(k \log(n/k))$  for sparse recovery or $m \geq O(rp)$ for low rank recovery.

The conventional approach to provide these results is to exhibit a regularizer $R$ for a given model set $\Sigma$ and to give the best possible recovery guarantees for the pair $(R,\Sigma)$. Recent works aim at giving guidelines to obtain guarantees as tight as possible for general sparse models  and convex regularizers \cite{Chandrasekaran_2012,Amelunxen_2014,vershynin2015estimation,Traonmilin_2016,Amelunxen_2020,Marz_2020}.  With such frameworks, it becomes possible to compare the performance of different regularizers. This leads naturally to the following question which we address in this work: \textbf{what is the ``best'' convex regularizer  to recover a given low-dimensional model $\Sigma$?}

To tackle this problem, it is necessary to define the notion of ``best'' based on recovery guarantees. We propose different possibilities and follow one route that permits us to give optimality results in the sparse and low-rank cases and shows the difficulties that arise when considering more complex generalized sparsity models. This work can be viewed as a way to give meaning to the expression ``convexification'' of a low-dimensional model, that is often used and rarely defined.

\subsection{Related works}

\paragraph{Low-complexity models induced by convex regularization.}
Many regularizers encountered in signal processing and machine learning are known to induce a specific type of model.
Without aiming for exhaustivity, the use of the $\ell^{1}$ norm~\cite{Chen1998AtomicDecompositionBasis} induces a sparse pattern in the solution, while group regularization with mixed $\ell^{1}-\ell^{2}$ norms restricts this sparse pattern to satisfy a specific block structure~\cite{Yuan2006Modelselectionestimation}.
More advanced model sets, such as low-rank matrices are linked to the use of the nuclear norm~\cite{Fazel2001rankminimizationheuristic}.
For a wide class of regularizers, including decomposable norms~\cite{Candes2013Simpleboundsrecovering}, decomposable $M$-estimator~\cite{Negahban2012UnifiedFrameworkHighDimensional}, atomic norms~\cite{Chandrasekaran_2012} and partly smooth functions~\cite{Vaiter2015Modelselectionlow,Vaiter2017ModelConsistencyPartly}, the connection between nonsmooth convexity and model space can be made explicit. Note that all these works take the following stance: given a convex regularizer $R$, what is the model set $\Sigma$ induced by minimizing $R(x)$? Conversely, in this paper, we study the question of finding the best regularizer for a given low-dimensional model $\Sigma$.

\paragraph{Convexification of combinatorial functions.}
Given a real function $f$, it is known that its biconjugate $f^{**}$ is a convex and closed function, whatever the initial properties of $f$.
For instance, if $f$ is the constant function equal to 1 except in 0 -- that is the counting function $\ell^{0}$ in dimension 1 -- restricted to $[-1,1]$, \ie
\begin{equation*}
  f(x) =
  \begin{cases}
    1 & \text{if } x \in [-1,1] \setminus \{0\}, \\
    0 & \text{if } x = 0, \\
    + \infty & \text{otherwise,}
  \end{cases}
\end{equation*}
then its biconjugate is the absolute value $| \cdot |$ restricted to $[-1,1]$.
Unfortunately, this construction is harder to generalize on an unbounded domain or in higher dimension.
For instance, the biconjugate of the $\ell^{0}$ counting function not restricted to a bounded set is the constant 0.
Of interest, we can mention convex closures of submodular functions (functions of $\{0,1\}^p$)  that can be calculated explicitly using the Lovász extension~\cite{bach2013learning} and convex closure of almost convex functions~\cite{jach2008convexenvelope}.

\paragraph{Convexification of objective function}

Many works intent to find a convex proxy to a non-convex objective function. In \cite{bertsekas1979convexification}, adding a Lagrangian term to the regularization of a constrained non-convex minimization  permits to build an equivalent minimization problem that is convex locally. Another possibility is to try to perform a regularization by infimal regularization \cite{bougeard1991towards} for lower semicontinuous objective functionals. In~\cite{Pock_2010},  in a function space setting, Pock \emph{et al.} propose a high dimensional lifting of the Lagrangian formulation of~\eqref{eq:minimization1} where the data-fit functional is non-convex. In the context of non-convex polynomial optimization, Lasserre's hierarchies \cite{lasserre2018moment} are used to recast the original problem in a hierarchy of convex semi-definite positive problems which provide global convergence results. The drawback of this method is the computational cost that makes it impractical for high-dimensional problems. Finally, convex closure of submodular functions also permits to cast sparsity inducing objective functions (where the regularizer is a submodular function of the support) into convex problems~\cite{bach2013learning}.
Note that if one aims to find a non-convex, but continuous, regularization, several works of interest may be cited, such as the use of $\ell^{p}$ minimization~\cite{foucart2009sparsest}, SCAD~\cite{fan2001variable}, or CEL0~\cite{Soubies_2015}.
Nevertheless, in this paper, we focus on convex functions.

\subsection{Contributions}\label{sec:contributions}

In this paper, we define notions of \emph{compliance measures} between a low-dimensional model and a regularizer  in finite dimension. The compliance of a function $R$ for  a model $\Sigma$ is a function
\begin{equation}
R \mapsto A_\Sigma(R)
\end{equation}
that  quantifies the recovery capabilities of $\Sigma$ with $R$ and minimization~\eqref{eq:minimization1}.

An optimal  regularizer for a model $\Sigma$ is defined as a regularizer that maximizes the compliance measure.  In this article, we  focus on the maximization of compliance measures on the set $\sC$ of coercive continuous convex regularizers over $\sH$. Note that this idea was first mentioned in the preliminary  work \cite{Traonmilin_2018} where optimal regularizers for sparse recovery were considered among weigthed $\ell^1$-norms.

\begin{itemize}
 \item We introduce compliance measures in \textbf{\Cref{sec:comp_general}} using tight recovery guarantees: a good regularizer is a regularizer that permits the recovery of $\Sigma$ as often as possible. We discuss the elementary properties of these measures and show that optimal coercive continuous convex regularizers can be found in the smaller class of atomic norms with atoms included in the model set. While such compliance measures can be optimized in basic toy examples, they require to be approximated in order to deal with sparse and low-rank models.

\item We  propose in \textbf{\Cref{sec:RIP_compliance}} compliance measures exploiting best known uniform recovery guarantees based on the restricted isometry property (RIP). We give explicit formulations of such recovery guarantees and show that, indeed, the $\ell^1$-norm and the nuclear norm are optimal for sparse and low-rank recovery (respectively) among coercive continuous convex regularizers.

\item We  study the case of two generalized sparsity models in \textbf{\Cref{sec:sparsity_in_levels}}: sparsity in levels and sparse plus low-rank  models. We show how an optimal regularizer can be explicitly constructed in a small family of convex regularizers ($\ell^1$-norm weighted by levels and mixed weighted $\ell^1$-nuclear norm respectively). While giving optimal weighting schemes for mixed regularizations, these examples also show the difficulty  of calculating optimal regularizers for new low-dimensional models and opens many questions for the study of optimal regularizers.
\end{itemize}
 We give an overview of the different compliance measures and the nature of results considered in this paper in~\Cref{tab:comp_measures}.

\begin{table}[h!]
\centering
\footnotesize

\begin{tabular}{|p{2.5cm}|p{3.5cm}|p{0.6cm}|p{4.5cm}|p{6cm}}
\hline
\specialcell{Compliance\\ (type of recovery)} & Definition & Section & Results \\
\hline
\specialcell{ \textbf{Based on descent cone}\\  (uniform)} & $f(\sT_R(\Sigma))$ & Sec.~\ref{sec:comp_general}& \specialcell{Optimality of atomic norms (Th~\ref{th:simpl_max_adeq}); \\ Equivariance (Lem. \ref{lem:transf_sigma_gen}); \\Invariance  (Cor.~\ref{cor:InvariantCompliance})}  \\
\hline
\specialcell{ $A^{U}_\Sigma(R)$: \textbf{Volume} \\(uniform)} &  $1 - \frac{\vol\left(\sT_R(\Sigma) \cap S(1)\right)}{\vol(S(1))}$ & Sec.~\ref{sec:comp_general}& \specialcell{ Monotonicity (Lem.~\ref{lem:compmono});\\ Invariance  (Cor.~\ref{cor:compinv})} \\
 \hline
\specialcell{ $\delta_\Sigma^{\mathtt{suff}}(R)$:  \textbf{Sufficient RIP} \\ (uniform)}  &  $\frac{1}{ \sqrt{ \underset{z\in \sT_R(\Sigma) \setminus \{0\} }{\sup}  \frac{\|z -P_\Sigma(z)\|_\Sigma^2}{\|P_\Sigma(z)\|_2^2} + 1 }}$ & Sec.~\ref{sec:RIP_compliance}& \specialcell{Characterization (Lem.~\ref{lem:charac_suff_RIP}, Cor.~\ref{cor:suffcond}); \\ 
 Optimal sparse reg. and optimal \\ low-rank reg. (Th.~\ref{th:RIP_suff_atom}, Th.~\ref{th:RIP_suff_atom_nuclear}); \\
  Sharp bound for near-optimal reg. for\\ sparsity in levels (Th.~\ref{th:RIP_suff_in_levels})} \\
 \hline
 \specialcell{ $\delta^{\mathtt{nec}}_\Sigma(R)$: \textbf{Necessary RIP} \\(uniform)}  & $\inf_{z \in \sT_R(\Sigma) \setminus \{0\}} \delta_{\Sigma}(I-\Pi_z )$ & Sec.~\ref{sec:RIP_compliance} & \specialcell{Characterization (Lem.~\ref{lem:expr_RC_sparse}, Cor.~\ref{cor:RCnecUoS});\\
  Optimal sparse reg. and optimal \\ low-rank reg. (Th.~\ref{th:RIP_nec_atom}, Th.~\ref{th:RIP_nec_atom_rank});  \\ Opt. weights for sparsity in levels (Th. \ref{th:opt_in_levels}); \\Opt. weights for sparse+low-rank  (Th.~\ref{th:sparse_LR})}   \\
 \hline
\specialcell{ $\delta^{\mathtt{sharp}}_\Sigma(R)$: \textbf{Sharp RIP} \\ (uniform)}   &  $ \inf_{M : \ker M \cap \sT_R(\Sigma) \neq \{0\} } \delta_{\Sigma}(M)$ & Sec.~\ref{sec:RIP_compliance} & \specialcell{Characterization (Prop.~\ref{prop:explicit}); \\ Invariance (Lem.~\ref{lem:transf_sigma_RC}); \\ Bound by  $\delta^{\mathtt{nec}}$ and $\delta_\Sigma^{\mathtt{suff}}(R)$ (Eq.~\ref{eq:IneqRIPConstants})}\\
 \hline
\specialcell{ $A_\Sigma^{NU}(R)$: Volume \\(non-uniform)}& $1 - \sup_{x \in \Sigma} \frac{\vol\left(\sT_R(x) \cap S(1)\right)}{\vol(S(1))} .$  & Sec.~\ref{sec:comp_general} & - \\
 \hline
 Kinematic formula (non-uniform)&  \specialcell{$\sup_{x \in\Sigma}\mathbb{P}\big(\tker M \cap \sT_R(x) \neq \{0\}\big)$,\\ $M$ Gaussian} & Sec.~\ref{sec:comp_general} &  -\\
 \hline
 Statistical dimension (non-uniform)& $\sup_{x\in\Sigma}\mathtt{statdim} (\sT_R(x)) $ & Sec.~\ref{sec:comp_general} & - \\
 \hline
\end{tabular}
\normalsize
\caption{A summary of different compliances measures and results. Compliances for which some results are given in this article are in bold (we focused on uniform recovery guarantees). Explicit maximization is performed on compliances based on necessary and sufficient conditions with the restricted isometry property (RIP) yielding  bounds on sharp RIP-based compliances.}\label{tab:comp_measures}
\end{table}

\subsection{Notations} \label{sec:notations}

In $\sH$, we denote $S(1) := \{ z \in \sH :  \|z\|_\sH =1 \}$ the unit sphere with respect to $\|\cdot\|_\sH$. Given a linear operator $M : \sH \to \bC^m$, we denote $M^H$ its Hermitian adjoint. 

For $\Sigma \subseteq \sH$ an arbitrary set, we denote $\iota_{\Sigma}$  its characteristic function defined by $\iota_{\Sigma}(x) = 0$ if $x \in \Sigma$, $\iota_{\Sigma}(x) = +\infty$ otherwise. We denote $\sE(\Sigma) \defin \bRp \cdot \cl{\tconv}(\Sigma)$, where $\cl{\tconv}(\Sigma)$ is the closure of the convex hull of $\Sigma$.
We define $\bar\bR \defin \bR \cup \{+\infty\}$.
Given a function $f: \sH \to \bar \bR$, we denote by $\dom(f)$ its domain, \ie the set $\dom(f) := \{ x \in \sH: f(x) < + \infty \}$.

\section{Optimal regularizer for a low dimensional model}\label{sec:comp_general}
In this section, starting from the characterization of exact recovery of a model $\Sigma$, we introduce the notion of compliance measure and associated optimal convex regularizer. 

\subsection{Characterization of exact recovery using descent cones} 

Before defining an optimal regularizer, we must characterize when $\Sigma$ can be recovered by solving~\eqref{eq:minimization1}. The fact that a given $x_0 \in \Sigma$ is recovered by solving~\eqref{eq:minimization1} with regularizer $R$ (\ie that the solution $x^{*}$ of~\eqref{eq:minimization1} is unique and satisfies $x^{*} = x_{0}$ when $y := Mx_{0}$) is equivalent to the fact that $R(x_0+z) > R(x_0)$ for every
 $z \in \ker (M) \setminus \{ 0 \}$ (see \eg \cite{Chandrasekaran_2012}). To summarize this, we use the following definition of symmetrized descent cones.
 \begin{definition}[(Symmetrized) descent cones.]\label{def:DescentCone}
 Consider a regularizer $R: \sH \to \bar\bR$.
 For any $x \in \dom(R)$, the \textbf{descent cone} of $R$ at $x$ is
 \begin{equation}
 \sT_{R}(x) \defin \left\{ \gamma z : \gamma \in \bR, z \in \sH,  R(x+z) \leq R(x) \right\}.
 \end{equation}
 For any set $\Sigma \subset \dom(R)$, we define $\sT_R(\Sigma):=\bigcup_{x\in \Sigma} \sT_R(x)$.
\end{definition}

 Other definitions of descent cones (\eg non-symmetric like in \cite{Chandrasekaran_2012}) could be used. The symmetrization facilitates technical derivations in the following and permits to characterize recovery as well. For ease of reading, in the following, symmetrized descent cones will be referred to as descent cones.
 
\noindent Recovery guarantees with a regularizer  $R$ for a linear operator $M$ generally come in two flavors (recall that $x^*$ is the result of minimization~\eqref{eq:minimization1}): 
\begin{itemize}
 \item {\bf Non-uniform recovery}: If $x_0 \in \Sigma$, then $x^* =x_0$  is equivalent to $\sT_R(x_0) \cap \ker M = \{0\}$.
 \item {\bf Uniform recovery}:   ``For all $x_0 \in \Sigma$, $x^* =x_0$'' is equivalent to 
 \begin{equation}\label{eq:charac_exact_rec}
\sT_R(\Sigma) \cap \ker M =\{0\}.
 \end{equation}
\end{itemize}

In the literature, recovery guarantees are obtained when the measurement operator $M$ fulfills sufficient conditions that imply these characterizations. Distinguishing these two types of recovery guarantees especially makes sense in the context of compressed sensing when $M$ is chosen at random. Typical results are then of the form:
\begin{itemize}
 \item {\bf Non-uniform recovery}: Given $x_0 \in \Sigma$, with high probability on the draw of $M$,  $x^* =x_0$.
 \item {\bf Uniform recovery}: With high probability on the draw of $M$,  $x^* =x_0$ for all $x_0 \in \Sigma$.
\end{itemize}
The main techniques to obtain recovery guarantees using a condition on the number of measurements differ largely between these two cases (see \Cref{sec:RIP_compliance}). In this work, we mostly focus on uniform recovery guarantees to stay in a fully deterministic setting. For such uniform recovery guarantees, we see that the only interactions that matter between the model set $\Sigma$, the regularizer $R$, and the measurement operator $M$ are summarized by equation~\eqref{eq:charac_exact_rec}.

\subsection{Compliance measures and optimal regularization}

To define a notion of optimal regularizer, we simply propose to say that an optimal regularizer is a function that optimizes a (hopefully well-chosen) criterion. We call such a criterion, a \emph{compliance measure} and specifically aim at defining it such that it should be maximized. The objective is to define a compliance measure that \emph{quantifies} the recovery capabilities of a given regularizer $R$ given a model set~$\Sigma$.

Starting from the characterization of exact recovery,  we can make the following remark. If the descent sets of a regularizer $R_1$ are included in the descent sets of another regularizer $R_2$, then the recovery capability of $R_1$ are greater in the following way:  success of recovery with $R_2$ implies success of recovery with $R_1$. Any ``reasonable'' 
compliance measure quantifying recovery capabilities \emph{needs} to fulfill the following axiom:
\begin{center}
 \emph{A compliance measure must be monotonously decreasing with respect to the inclusion of descent sets.}
\end{center}

We also see that the kernel of $M$ heavily influences the recovery capability of $R$. If we had some knowledge that $M \in \sM$ where $\sM$ is a set of linear operators, we would want to define a compliance measure $A_{\Sigma,\sM}(R)$ that tells us how good is a regularizer  in these situations and to maximize it.  Such maximization might yield a function $R^*$ that depends on $\sM$ (\eg in \cite{Soubies_2015}, when looking for tight continuous relaxation of the $\ell^0$ penalty a dependency on $M$ appears). 
In the following, we propose a more \emph{universal} notion of optimal convex regularizer  that does not depend on a particular class of linear operators $\sM$: we propose compliance measures $A_{\Sigma}(R)$ that depend only on the set $\Sigma$ and on the regularizer  $R$, and consider their maximization on some set of convex functions $\sC$
 (that are  coercive and continuous, see \Cref{sec:def_convex_function_set}):
\begin{equation}\label{eq:DefPrincipleCompliance}
\sup_{R \in \sC} A_{\Sigma}(R). 
\end{equation}
Of course, the existence  of a maximizer of  $A_{\Sigma}(R)$ is in itself a general question of interest: we could ask ourselves what conditions on $A_\Sigma(R)$ and $\sC$ are necessary and sufficient for the existence of a maximizer, which is out of the scope of this article -- we notably expect potential difficulties when normalized atoms defining the model set are not a compact set. In the sparse recovery and low-rank matrix recovery examples studied in this article, the existence of a maximizer of the considered  compliance measures will be verified.

To build a compliance measure that does not depend on $M$, we define the optimal regularizer as the regularizer which guarantees recovery of $\Sigma$ in as many situations as possible, \ie for ``as many linear operators $M$ as possible''. Intuitively, a regularizer $R$ is ``good'' if  the set $\sT_R(\Sigma)$ ``leaves a lot of space'' for $\ker M$ to not intersect it (trivially), see \Cref{fig:reconstruction_cond}.
Among non-convex regularizers, the optimal one is the characteristic function of the model set $\Sigma$.

\begin{figure}[!t]
  \centering
\includegraphics[width=0.5\linewidth]{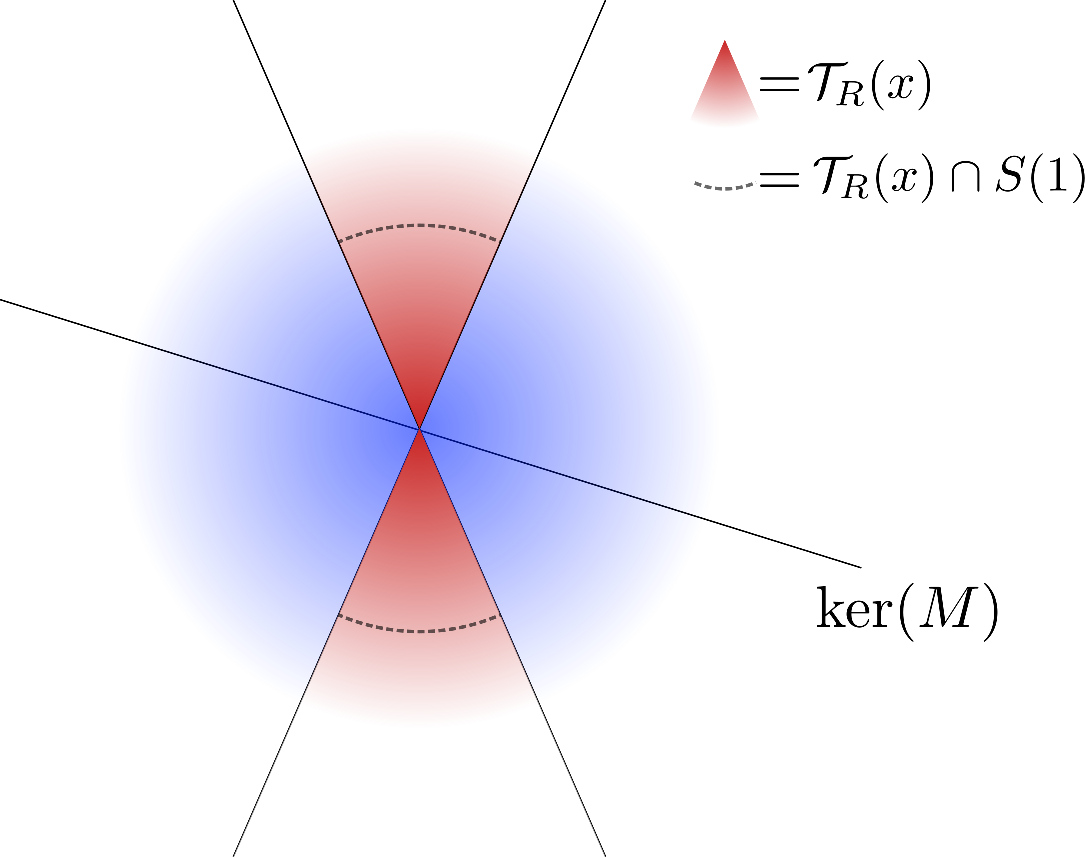}
\caption{A representation of recovery guarantees based on descent cones of a convex function. Recovery of $x \in \Sigma$ fails if $\ker(M)$ intersects $\sT_R(x)$ non trivially. The bigger is the descent cone (red) the more likely recovery will fail. The bigger  the space left by the descent cone (blue), the more likely recovery will succeed} 
  \label{fig:reconstruction_cond}
\end{figure}

\begin{lemma}[Optimality of the characteristic function.]\label{lem:ideal_decoder_descent}
Consider an arbitrary non-empty set $\Sigma \subseteq \sH$ and denote $\iota_\Sigma$ its characteristic function.
The corresponding descent cone is 
\[
\sT_{\iota_{\Sigma}}(\Sigma) = \{\gamma z: \gamma \in \bR, z \in \Sigma-\Sigma\} \supseteq \Sigma-\Sigma
\]
where $\Sigma-\Sigma$ is the so-called secant set of $\Sigma$.
For any regularizer  $R$ such that $\Sigma \subseteq \mathtt{dom}(R)$ we have $\sT_{\iota_{\Sigma}}(\Sigma) \subseteq \sT_{R}(\Sigma)$.
 Finally, if $\Sigma$ is positively homogeneous then $\sT_{\iota_\Sigma}(\Sigma) = \Sigma -\Sigma$.
\end{lemma}
\begin{proof}
 See \Cref{sec:proof_ideal_comp}
\end{proof}

This Lemma shows that $\iota_\Sigma$ is at least as successful as any regularizer $R$ for the exact recovery of $\Sigma$  (without any consideration of compliance measure). Moreover, $\sT_{\iota_{\Sigma}}(\Sigma)$ is the smallest possible descent cone with respect to inclusion.
Hence $\iota_\Sigma$ can be considered as the ideal regularizer \cite{Bourrier_2014} and indeed the optimal one with respect to any compliance measure defined as $A_{\Sigma}(R) = f(\sT_{R}(\Sigma))$ where $f$ is some function on subsets of $\sH$ that is monotonic with respect to set inclusion. This is why the search for optimal regularizers only makes sense under some constraint on $R$. 

\subsection{A first compliance measure}
As a first concrete example, we define here a theoretical compliance measure that reflects the idea that smaller descent cones are better. However, this compliance measure does not lead to analytical expressions for the general study of sparse recovery. Our core results in the next sections rely on compliance measures based on best known uniform recovery guarantees using the restricted isometry property (RIP).

For convex functions, first, observe that, as only the \emph{directions} of the descent cones and the kernel play a role in recovery guarantees, the size of descent cones can be measured by considering only their intersection with the unit sphere $S(1)$.  
Choosing the norm $\|\cdot\|_{\sH}$ to define the unit sphere is natural (although also somewhat arbitrary) as this is the only metric introduced so far in the considered setting. It will also appear to define RIP constants soon.
Second, if we want to consider a measure that is invariant by rotation, the uniform  measure on the  unit sphere $S(1)$ comes somewhat naturally. It is indeed the unique Haar measure. The uniqueness is essentially guaranteed when it is a measure in the sense of measure theory (additive, non-negative function over a $\sigma$-algebra). In our setting, using this measure is a way of considering that we do not have prior information on the orientation of the kernel of $M$, or on the orientation of the model set $\Sigma$.

Using this measure, given a convex function $R$, the ``amount of space left for the kernel of $M$'' can be quantified by the ``volume'' of the intersection $\sT_R(\Sigma) \cap S(1)$ of the descent cone with the unit sphere. Hence, a compliance measure for \emph{uniform recovery} can be defined as 

\begin{equation}\label{eq:DefUniformComplianceMeasure}
A_\Sigma^U(R) := 1 - \frac{\vol\left(\sT_R(\Sigma) \cap S(1)\right)}{\vol(S(1))}.
\end{equation}
Here, the volume $\vol(E)$ of a set $E$ is the measure of $E$ with respect to the uniform measure on the sphere $S(1)$ (\ie the $n-1$-dimensional Hausdorff measure of $\sT_R(\Sigma) \cap S(1)$, when $\sH$ is $n$-dimensional). This measure is well-defined as the descent cones of convex functions are symmetrized convex cones.

When looking at \emph{non-uniform recovery} for random Gaussian measurements, the quantity defined by $\frac{\vol\left(\sT_R(x_0) \cap S(1)\right)}{\vol(S(1))}$ represents the probability that a randomly oriented kernel of  dimension $1$, defined as the span of a random vector uniformly distributed on the sphere $S(1)$, intersects (non trivially) $\sT_R(x_0)$. The highest probability of intersection with respect to $x_0$  quantifies the lack of compliance of $R$, hence we could define:
\begin{equation}\label{eq:anusr}
A_\Sigma^{NU}(R) := 1 - \sup_{x \in \Sigma} \frac{\vol\left(\sT_R(x) \cap S(1)\right)}{\vol(S(1))} .
\end{equation}

This can be linked with the Gaussian width and statistical dimension theory of non-uniform sparse recovery \cite{Chandrasekaran_2012,Amelunxen_2014}. 
Indeed, if $M$ is a random Gaussian matrix of size $(n-1) \times n$, we have 
\begin{equation}
 \mathbb{P}\big(\tker M \cap \sT_R(x_0) \neq \{0\}\big) =  \frac{\vol\left(\sT_R(x_0) \cap S(1)\right)}{\vol(S(1))}.
\end{equation}
As shown in \cite{Amelunxen_2014}, for a random Gaussian matrix $M$ of size $m \times n$ with any number of measurements $m$,
the probability $\mathbb{P}\big(\tker M \cap \sT_R(x_0) \neq \{0\}\big)$
can be guaranteed to be small if $m$ is greater than the statistical dimension of the descent cones. The kinematic  formula (Crofton's formula in this case) gives the exact value
\begin{equation}
 \mathbb{P}\big(\tker M \cap \sT_R(x_0) \neq \{0\}\big) =\sum_{j=m+1,\ j\ \text{even}}^{n} v_j(\sT_R(x_0))
\end{equation}
where $ v_j(K)$ is the $j$-th intrinsic volume of a cone $K$. For a polyhedral cone it is the probability that the orthogonal projection on $K$ of a Gaussian vector lies in a $j$-dimensional face of $K$. The statistical dimension of a descent cone $\sT$ is defined by \cite[Definition 2.2]{Amelunxen_2014}
\begin{equation}\label{eq:DefStatDim}
\mathtt{statdim} (K) =\sum_{j=0}^{n} jv_j(K).
\end{equation}
As it is used to bound the number of measurements in the non uniform case, its supremum over all the descent cones $K = \sT_{R}(x_{0}), x_{0} \in \Sigma$ could be used as a compliance measure. Moreover, it was shown that the statistical dimension is a measure of the ``size'' of the convex cones that is additive, invariant by rotation, and monotonous.

The above compliance measures are completely dependent on the metric defining $S(1)$ (here the Hilbert norm $\|\cdot\|_\sH$), other choices could be considered especially if  measurement operators $M$ showing a particular structure were considered.

In this article, we concentrate on compliance measures based on uniform recovery guarantees.

\begin{remark}
These compliance measures implicitly force $\Sigma \subset \dom(R)$, unless $A_{\Sigma}(R) = 0$. Indeed, suppose there exists $x\in\Sigma$ such that $R(x)=+\infty$, then for all $z \in \sH$, we have $R(x+z)\leq +\infty = R(x)$. This implies $\sT_R(x) = \sH$ and  $A_{\Sigma}^U(R)=A_{\Sigma}^{NU}(R) = 0 $.

\end{remark}

\begin{remark}
When studying convex regularization for low dimensional recovery in infinite dimensional separated Hilbert spaces, the noiseless recovery only depends on the behavior of the regularizer $R$ on $\sE(\Sigma)$ (defined in \Cref{sec:notations}). The behavior of $R$ outside $\sE(\Sigma)$ is only studied to obtain properties of robustness to modeling error \cite{Traonmilin_2016}. In many examples of generalized sparsity and low-dimensional modeling in infinite dimension, the space $\sE(\Sigma)$ has a finite dimension~\cite{Adcock_2013b}.  Our framework still applies in this case.

It is an open question to generalize our framework for low-dimensional recovery in more general settings such as Banach spaces (\eg for off-the-grid super-resolution).

\end{remark}

\begin{remark}
In the uniform recovery case, the compliance measure $A_{\Sigma}^{U}$ defined in~\eqref{eq:DefUniformComplianceMeasure}
is monotonous with respect to the partial ordering of descent cones defined by the inclusion property. However, it does not (at least explicitly) take into account potential effects of the dimension of the kernel of $M$, which may be higher than one. For a given 
dimension $\ell$ of the
kernel of $M$, the uniform measure 
on the corresponding Grassmanian manifold (of all subspaces of dimension $\ell$)
would be more natural as it would directly quantify the probability of 
intersection with
a random kernel of fixed dimension. 
This measure for kernels of dimension $\ell$ and  a  descent cone $K$ is the following:
\begin{equation}
V_\ell(K):=\mu_{O(n)} \left(\{Q \in O(n) : (Q E) \cap K \neq \{ 0\}  \} \right)
\end{equation}
where $\mu_{O(n)}$ is the uniform measure on the orthogonal group and $E$ is 
an arbitrary fixed $\ell$-dimensional subspace. The measure $V_{\ell}$ is invariant by rotation and
for $\ell=1$ it matches the Haar measure used in~\eqref{eq:DefUniformComplianceMeasure}-\eqref{eq:anusr}.

Given a set $\Sigma$, and assuming the existence of a maximizer $R^{*}$ of $A^{U}_{\Sigma}$ (within a prescribed family of possible regularizers), there are only two possibilities: either all maximizers of  $A_{\Sigma}^U(R)$  also minimize $V_\ell(\sT_{R}(\Sigma))$, or not. In this last case, it would mean that there is $R^*$ maximizing $A_{\Sigma}^U$  and not minimizing $V_\ell$. It is an interesting challenge, left to future work, to understand whether this case can indeed happen.

\end{remark}

We remind  the reader that compliances considered in this article are summarized in~\Cref{sec:contributions}.
\subsection{Coercive continuous convex functions} \label{sec:def_convex_function_set}

As mentioned before we look for practical regularizers. 
We define $\sC$ the set of all functions $R: \sH \to \bR$ (\ie with $\dom(R) = \sH$) that are convex, continuous, and coercive.

The coercivity condition is typical in the context of convex regularization in order to avoid constant functions.

With any proper lower semi-continuous regularizer (hence, with any regularizer in $\sC$) the convergence of  the primal dual algorithm is guaranteed \cite{Chambolle_2011}. This guarantees the existence of practical algorithms  (for the problem  $\min_x \tfrac{1}{2}\|Mx-y\|^{2} + \lambda R(x)$ ) when the proximity operator
\begin{equation}
 y \mapsto \mathtt{prox}_{\lambda R}(y) := \arg \min \tfrac{1}{2} \|u-y\|_\sH^2 + \lambda R(u)
\end{equation}
can be approximated efficiently.

\subsection{Elementary properties and reduction to atomic ``norms''}\label{sec:atomic}

As  compliance measures  based on uniform recovery guarantees can be written as functions of descent cones $\sT_{R}(\Sigma)$, we have the following immediate Lemma.

\begin{lemma}[The compliance measure $A_\Sigma^{U}$ is monotonic.]\label{lem:compmono}
 Let $R_1,R_2$ be two functions such that $\sT_{R_1}(\Sigma) \subset \sT_{R_2}(\Sigma)$ then $A_\Sigma^{U}(R_1) \geq A_\Sigma^{U}(R_2)$.
\end{lemma}
In other words, the compliance measure is decreasing with respect to the inclusion of descent cones.
We also verify that multiplying a regularizer  by a scalar does not change the compliance measure which is consistent with recovery guarantees. 
\begin{lemma}[The compliance measures $A_\Sigma^{U}$ and $A_\Sigma^{NU}$ are 0-homogeneous.]
Let $\lambda >0$. Then, 
\begin{equation}
\begin{split}
 A_\Sigma^U( \lambda R)&= A_\Sigma^U(R), \\
 A_\Sigma^{NU}( \lambda R)&= A_\Sigma^{NU}(R). \\
 \end{split}
\end{equation}
\end{lemma}
\begin{proof}
Let $x \in \Sigma$. We remark that, the tangent cone is invariant by scalar multiplication:
\begin{equation}
\begin{split}
 \sT_{\lambda R}(x) &= \{ \gamma z : \gamma \in \bR, \lambda R(x+z) \leq \lambda R(x)\} \\
  &= \{ \gamma z : \gamma \in \bR; R(x+z) \leq  R(x)\} \\
  &=  \sT_{R}(x).
 \end{split}
\end{equation}
This shows directly that $ A_\Sigma^{NU}( \lambda R)= A_\Sigma^{NU}(R) 
$. This also implies that $\sT_{\lambda R}(\Sigma) = \sT_{ R}(\Sigma)  $ and $ A_\Sigma^U( \lambda R)= A_\Sigma^U(R) $.

\end{proof}

More generally, any operation on $R$ that leaves $\sT_R(\Sigma)$ invariant does not change the compliance measure. This is typically the case of the post-composition of $R$ with an increasing function.

We now recall the notion of atomic ``norm'' and show that optimal regularizers can be found in a set of atomic norms.

\begin{definition}[Atomic norm.]
 The \textbf{atomic ``norm''} induced by a set $\sA$ is defined as: 
\begin{equation}
 \|x\|_\sA \defin \inf \left\{ t \in \bRp:  x \in t\cdot\cl{ \tconv}(\sA) \right\}
\end{equation}
where $\cl{ \tconv}(\sA)$ is the closure of the convex hull $\tconv(\sA)$ in $\sH$. This ``norm'' is finite only on 
\begin{equation}
\sE(\sA) \defin \bRp \cdot \cl{\tconv}(\sA) = \{x = t\cdot y, t \in \bRp,y \in \cl{\tconv}(\sA)\} \subset \sH.
\end{equation}
It is extended to $\sH$ by setting $\|x\|_\sA = +\infty$ if $x \notin \sE(\sA)$. 
\end{definition}

Classical convex regularizer for sparse and low rank models are atomic norms:
\begin{itemize}
 \item  The $\ell^1$-norm $\|\cdot\|_1$ is the atomic norm induced by signed canonical basis vectors.
 \item  The nuclear norm $\|\cdot\|_*$ is the atomic norm induced by unitary rank-one matrices.
\end{itemize}

Atomic norms are convex gauges induced by the convex hull of atoms. Their properties can be linked with the properties of the set $\sA$  with classical results on convex gauge functions
 (see \Cref{sec:summary_prev_results}). 

It is possible to define  an atomic norm, denoted $\|\cdot\|_\Sigma$, specifically induced by the model $\Sigma$.

\begin{definition}[Atomic norm induced by the model.] \label{def:modelnorm}
 Given a cone $\Sigma$, we define the \textbf{atomic norm induced} by $\Sigma$ as 
\begin{equation}
 \|\cdot\|_\Sigma \defin \|\cdot\|_{\Sigma \cap S(1)}.
\end{equation}
\end{definition}
This norm is known as the $k$-support norm for sparse model, it is not adapted to perform uniform recovery. In particular, it cannot recover $1$-sparse vectors.

In \cite[Lemma 2.1]{Traonmilin_2016}, it was shown that there is always an atomic norm with smaller descent cones than the descent sets of a proper coercive continuous function  with convex level sets. We give a version of this result for our definition of cones and specify the properties of the obtained atomic norm. 

\begin{lemma}[Optimality of atomic norms for a given model.]\label{lem:atomic_necessary}

Let $\Sigma$ be a cone such that $\sE(\Sigma) = \sH$ and $R$ be a coercive continuous convex function. Then there exists $\sA \subset \Sigma$ such that: 
\begin{equation}
\sT_{\|\cdot\|_\sA}(\Sigma) \subseteq \sT_R(\Sigma).
\end{equation}
and $\|\cdot\|_\sA$ is a  coercive, continuous, positively homogeneous convex function.
\end{lemma}
\begin{proof}
See~\Cref{sec:reduction_proof}.
\end{proof}

With \Cref{lem:atomic_necessary}, for all  coercive continuous  convex functions  $R$ (i.e. elements of $\sC$), it is possible to find an atomic norm $R'$ with atoms included in $\Sigma$ such that  $\sT_{R'}(\Sigma) \subset \sT_{R}(\Sigma)$. We define $\sC_\Sigma$ as the set of coercive continuous positively homogeneous atomic ``norms'' whose atoms $\sA$ are included in $\Sigma$:
\begin{equation}
 \sC_\Sigma := \{ \|\cdot\|_\sA \; : \sA \subset \Sigma , \|\cdot\|_\sA \; \in \sC, \forall x \in \sH, \lambda >0,  \|\lambda x\|_\sA = \lambda \|x\|_\sA \}.
\end{equation}
Note that positive homogeneity is guaranteed if $0$ is in the interior of $\cl{\tconv}(\sA)$ (see~\Cref{sec:summary_prev_results}).
As a consequence of this Lemma, we have the following  immediate result.
\begin{theorem}[Optimization of compliance measures over $\sC_\Sigma$.]\label{th:simpl_max_adeq}
Let $\Sigma$ be a cone such that $\sE(\Sigma) = \sH$.  Suppose $A_\Sigma$ is a compliance measure that is a decreasing function of $\sT_R(\Sigma)$ with respect to set inclusion. Then 
\begin{equation}
  \sup_{R\in\sC} A_\Sigma(R) =  \sup_{R\in\sC_\Sigma} A_\Sigma(R).
\end{equation}

In particular 
\begin{equation}
  \sup_{R\in\sC} A_\Sigma^U(R) =  \sup_{R\in\sC_\Sigma} A_\Sigma^U(R).
\end{equation}
\end{theorem}
\begin{proof}
Let  $R \in \sC$, with \Cref{lem:atomic_necessary}, there exists $\|\cdot\|_\sA \in \sC_\Sigma$ such that $ \sT_{\|\cdot\|_\sA}(\Sigma) \subset\sT_R(\Sigma) $.  This implies $ \sT_{\|\cdot\|_\sA}(\Sigma) \cap S(1) \subset\sT_R(\Sigma) \cap S(1)$ and $A_\Sigma(R) \leq  A_\Sigma(\|\cdot\|_\sA)$.

\end{proof}

\noindent  \Cref{th:simpl_max_adeq} shows that we can limit ourselves to atomic norms if our only objective is to maximize the compliance measure.

With such measures, we can calculate optimal regularizers for elementary linear transformations of  models.

\begin{lemma}[Compliance measures as functions of descent cones are equivariant to linear transformations.]\label{lem:transf_sigma_gen}
Consider a compliance measure defined as: $A_\Sigma(R) := f(\sT_R(\Sigma))$ with $f$ some scalar valued function defined on non-empty subsets of $\sH$. For any invertible linear map $F$ on $\sH$, any model set $\Sigma$ and any regularizer  $R$, we have
\begin{align}
\sT_R(F\Sigma) &= F(\sT_{R\circ F}(\Sigma))\\
 A_{F\Sigma}(R)  &=  f [ F (\sT_{R \circ F}(\Sigma))].
\end{align}
\end{lemma}
\begin{proof}
First $\gamma z \in \sT_R(F\Sigma)$ if, and only if, there exists $x \in \Sigma$ such that 
$R(Fx  + z) \leq R(F x)$,
\ie such that
$(R\circ F)(x  + F^{-1}z) \leq (R \circ F)( x)$.
This is in turn equivalent to $\gamma F^{-1}z \in \sT_{R\circ F}(\Sigma)$, \ie $\gamma z \in F(\sT_{R\circ F}(\Sigma))$.
Second, it follows that $A_{F\Sigma}(R)  = f(\sT_R(F\Sigma)) = f[F(\sT_{R \circ F}(\Sigma))]$.
\end{proof}

Thanks to \Cref{lem:transf_sigma_gen}, we can build optimal regularizers from other optimal regularizers when the model set is obtained from another one by applying a linear isometry.
\begin{corollary}[Compliance measures as functions of descent cones are invariant under invariant maps.]\label{cor:InvariantCompliance}
Consider a compliance measure defined as: $A_\Sigma(R) := f(\sT_R(\Sigma))$ with $f$ some scalar valued function on subsets of $\sH$.  Assume that $f$ is invariant to a family $\mathcal{F}$ of invertible linear maps on $\sH$, \ie for any $F \in \mathcal{F}$ and any non-empty set $\sT \subseteq \sH$, $f(F\sT) = f(\sT)$. Then, for any $F \in \mathcal{F}$, any regularizer  $R$ and any model set $\Sigma$, we have
 \begin{equation}
   A_{F\Sigma}( R\circ F^{-1})  =  A_{\Sigma}( R).
\end{equation}
\end{corollary}
\begin{proof}
By \Cref{lem:transf_sigma_gen},
\(
A_{F\Sigma}(R \circ F^{-1}) = f[F(\sT_{(R \circ F^{-1}) \circ F}(\Sigma))] = f(F\sT_{R}(\Sigma)) = f(\sT_{T}(\Sigma)) = A_{\Sigma}(R).
\)

\end{proof}

\begin{corollary}[Compliance measures $A_{\Sigma}^U$ are invariant by isometries.]\label{cor:compinv}
Consider  $F$ an isometry on $\sH$, $R$ a regularizer  and $\Sigma$ a model set. We have 
 \begin{equation}
   A_{F\Sigma}^U( R\circ F^{-1})  =  A_{\Sigma}^U( R).
\end{equation}
\end{corollary}
\begin{proof}
The volume is invariant to isometries, hence $A^{U}_{\Sigma}(R) = f^{U}(\sT_{R}(\Sigma))$ where $f^{U}(\cdot)$ is invariant to  isometries. 

\end{proof}

\subsection{An exact result in 3D: the most we can do?}  

Compliance measures $A_\Sigma^{U}(R)$ and $A_\Sigma^{NU}(R)$ were effectively optimized \cite{Traonmilin_2018} in the case of $1$-sparse recovery in dimension 3, \ie for $\Sigma = \Sigma_1$ the set of $1$-sparse vectors in $\bR^3$.  In this case, $\sC_\Sigma = \{ \|\cdot\|_\sA : \sA \subset \Sigma_1 \}$. It was shown that for the set $\sC_\Sigma' = \{ \| \cdot\|_\sA : \sA \subset \Sigma_1 , \sA = -\sA \} $ (which is the set of weighted $\ell^1$-norms) that
\begin{equation}
  \begin{split}
   &\arg \max_{R \in C_\Sigma'} A_\Sigma^{U}(R)=\arg \max_{R \in C_\Sigma'}  A_\Sigma^{NU}(R)= \{ \lambda \|\cdot\|_1 : \lambda > 0 \}. \\
  \end{split}
 \end{equation}

To show this, the solid angles of all descent cones of arbitrary weighted $\ell^1$-norms were calculated exactly, and their size minimized with respect to the weights. 

\begin{figure}[!h]
  \centering
\includegraphics[width=0.3\linewidth]{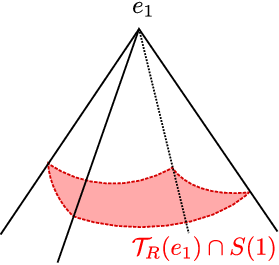}
\caption{Solid angle of  a half descent cone of a  weighted $\ell^1$-norm }
  \label{fig:3d-1-sparse}
\end{figure}

Unfortunately, calculating exactly these solid angles in dimension $d$ seems out of reach for any sparsity and atomic norm in $\sC_\Sigma$ even if some progress in bounds of these quantities \cite{Marz_2020} in some particular cases (non-uniform recovery with $\ell^1$-norm in probability with random matrices). To the best of our knowledge, no general bound is available for the volume of descent cones of arbitrary atomic norms in $\sC_\Sigma$ for sparse recovery. To build a compliance measure that we could optimize, we would need to first to establish such general bounds with some tightness.

In the next section,  we propose  to build compliance measures based on best-known uniform recovery guarantees that have some ``tightness'' properties. This will enable us to manipulate analytical expressions and give results for sparse recovery and low-rank recovery.

\section{Compliance measures based on the restricted isometry property}
\label{sec:RIP_compliance}

The most used tool for the study of uniform recovery is  the restricted isometry property (RIP). This property is adequate for multiple models~\cite{Traonmilin_2016}, to be tight in some sense~\cite{Davies_2009} for sparse and low-rank recovery, to be necessary in some sense~\cite{Bourrier_2014}, and to be well adapted to the study of random operators \cite{Puy_2015}.
In \cite{Traonmilin_2016}, given a regularizer  $R$, an explicit constant $\delta_\Sigma(R)$ is given,  such that $\delta_{\Sigma}(M)<\delta_\Sigma(R)$ guarantees the exact recovery of elements of $\Sigma$ by minimization~\eqref{eq:minimization1}. Hence, using $\delta_\Sigma(R)$ as a compliance measure, the higher
the value of $\delta_\Sigma(R)$, the less stringent the RIP condition on $M$ to ensure recovery of all
elements of $\Sigma$ using $R$ as a regularizer.

To formalize further this idea, we start by recalling definitions and results about RIP recovery guarantees then apply our methodology. We also give results that emphasize the relevant quantity (depending on the geometry of the regularizer  and the model) to optimize. 
\begin{definition}[RIP constant.] \label{def:RIP}
Consider an arbitrary non-empty set $\Sigma \subset \sH$ and $M$ a linear map from $\sH$ to $\bC^{m}$. The \textbf{RIP constant} of $M$ is defined as
 \begin{equation}
  \delta_{\Sigma}(M) = \sup_{x \in \Sigma-\Sigma} \left| \frac{\|Mx\|_2^2}{\|x\|_\sH^2}-1 \right|,
 \end{equation}
 where $\Sigma-\Sigma$ (differences of elements of $\Sigma$) is called the secant set. 
For brevity, we will simply denote $\delta(M)$ when the model set $\Sigma$ is clear from context.
 \end{definition}
 \noindent
 This coincides with the usual notion of RIP for sparse recovery when $\Sigma = \Sigma_{k}$ is the set of vectors with at most $k$ nonzero entries (and $\Sigma-\Sigma = \Sigma_{2k}$); and with the RIP for low-rank recovery when $\Sigma = \Sigma_{r}$ is the set of matrices of rank at most $r$ (then, $\Sigma-\Sigma =\Sigma_{2r}$).

A natural and sharp RIP-based compliance measure is $A^{RIP,\mathtt{sharp}}_{\Sigma}(R) = \delta^{\mathtt{sharp}}_\Sigma(R)$ defined as: 
\begin{equation}\label{eq:DefSharpRIPCst}
 \delta^{\mathtt{sharp}}_\Sigma(R) := \inf_{M : \ker M \cap \sT_R(\Sigma) \neq \{0\} } \delta_{\Sigma}(M).
\end{equation}
This is the smallest RIP constant of measurement operators where uniform recovery fails  \cite{Davies_2009}, hence the following immediate theorem.
\begin{theorem}[The compliance measure $\delta^{\mathtt{sharp}}_\Sigma(R)$ is sharp.]
Consider an arbitrary non-empty set $\Sigma \subseteq \sH$. Suppose $M$ has RIP with constant $\delta_{\Sigma}(M)<  \delta^{\mathtt{sharp}}_\Sigma(R)$, then  for all $x_0 \in \Sigma$ and $x^*$ the result of minimization~\eqref{eq:minimization1} satisfies
 \begin{equation}
  x^* = x_0.
 \end{equation}
Conversely, there exists $M$ with $\delta_{\Sigma}(M) \geq  \delta^{\mathtt{sharp}}_\Sigma(R)$ and $x_0 \in \Sigma$ such that $x^* \neq x_0$.
\end{theorem}
Despite this sharp property with respect to recovery, $\delta_{\Sigma}^{\mathtt{sharp}}(R)$ is not endowed with any known analytic expression more explicit than its definition, and it is an open question to derive closed-form formulations of this constant for a general $R$, even for the particular case of sparse or low-rank models. This limits the possibility to conduct an exact optimization with respect to $R$, and motivates the investigation of more explicit RIP-based compliance measures, with two flavors:  
\begin{itemize}
 \item  Compliance measures $\delta^{\mathtt{nec}}_\Sigma(R)$ based on necessary RIP conditions \cite{Davies_2009} which yield sharp recovery constants for particular set of operators $M$, \eg
\begin{equation}\label{eq:DefRIPCstNec}
 \delta^{\mathtt{nec}}_\Sigma(R) := \inf_{z \in \sT_R(\Sigma) \setminus \{0\}} \delta_{\Sigma}(I-\Pi_z ).
\end{equation}
 where $\Pi_{z}$ is the orthogonal projection onto the one-dimensional subspace $\tspan(z)$  (other intermediate necessary RIP constants can be defined). Such measures upper bound $\delta_\Sigma^{\mathtt{sharp}}(R)$ ($\delta^{\mathtt{nec}}_\Sigma(R)$ characterizes RIP recovery guarantees of measurement operators having the shape $I-\Pi_z$).
 \item Compliance measures $\delta^{\mathtt{suff}}_\Sigma(R)$ based on sufficient RIP constants for recovery (\eg the explicit sufficient RIP constant
 $\delta_\Sigma(R)$ from \cite{Traonmilin_2016}, which is explained in \Cref{sec:suff_RIP}), which are lower bounds to $\delta_\Sigma^{\mathtt{sharp}}(R)$.
\end{itemize}
Note that we have the relation 
 \begin{equation}\label{eq:IneqRIPConstants}
\delta_\Sigma^{\mathtt{suff}}(R) \leq  \delta^{\mathtt{sharp}}_\Sigma(R)  \leq \delta_\Sigma^{\mathtt{nec}}(R). 
\end{equation}

The next sections aim at showing that the $\ell^1$-norm (resp. the nuclear norm) maximizes the (best known) upper and lower bounds of $\delta_\Sigma^{\mathtt{sharp}}(R)$ for $k$-sparse model (resp. low rank models).

First, in \Cref{sec:RC}, we recall   that when $\Sigma$ is a non-trivial cone, the compliance measures associated to RIP constants can be cast to equivalent compliance measures associated to a restricted conditioning (RC), which can be written more explicitly. 

Second, in \Cref{sec:nec_RIP}, we use the expression of the RC-based compliance measure associated to $\delta_{\Sigma}^{\mathtt{nec}}(\cdot)$ (from Equation~\eqref{eq:DefRIPCstNec}) to show that the $\ell^{1}$ norm (resp. the trace-norm) optimizes $\delta_{\Sigma}^{\mathtt{nec}}(\cdot)$ for $k$-sparse vectors (resp. for matrices of rank at most $r$), with $\delta_{\Sigma}^{\mathtt{nec}}(R^{\star})\approx 1/\sqrt{2}$ when $n$ is large enough. 

Finally, in \Cref{sec:suff_RIP}, we give a characterization of $\delta_\Sigma^{\mathtt{suff}}(R)$ and show the optimality of the $\ell^1$-norm (resp. the nuclear norm) with $\delta_{\Sigma}^{\mathtt{suff}}(R^{\star})= 1/\sqrt{2}$.

\subsection{Restricted conditioning as a compliance measure} \label{sec:RC}
When working with a model set $\Sigma$ that is a cone, instead of working with actual RIP constants, it is easier to use (equivalently) the restricted conditioning.

\begin{definition}[Restricted conditioning.]
Consider a cone $\Sigma \subset \sH$ and a linear operator $M$ from $\mathbb{R}^{n}$ to $\mathbb{C}^{m}$.
We define the \textbf{restricted conditioning} of $M$ as
\begin{equation}
 \gamma_{\Sigma}(M) := \frac{\sup_{x\in (\Sigma-\Sigma)\cap S(1)} \|Mx\|_2^2}{\inf_{x\in (\Sigma-\Sigma)\cap S(1)} \|Mx\|_2^2} \in [1,\infty]
 \end{equation}
where by convention here $a/0=+\infty$ for any $a \geq 0$.
For brevity we will simply denote $\gamma(M)$ when $\Sigma$ is clear from context.
\end{definition}
As shown below, the RIP constant $\delta_{\Sigma}(M)$ is a monotonically increasing function of $\gamma_{\Sigma}(M)$. The advantage of the latter is that it is invariant by rescaling $M$ (rescaling leaves unchanged the recovery properties when measuring $x_0$ with  $M$). 

\begin{lemma}[The RIP constant $\delta_{\Sigma}(M)$ is monotone in $\gamma_{\Sigma}(M)$.]\label{lem:link_RIP_RC}
Consider a cone $\Sigma \subseteq \sH$.
For any $M$ such that $\gamma_{\Sigma}(M) < \infty$, there is a unique $\lambda>0$ such that
\begin{equation}
\gamma_{\Sigma}(M)= \frac{1+\delta_{\Sigma}(\lambda M)}{1-\delta_{\Sigma}(\lambda M)}
\end{equation}
or equivalently
\begin{equation}
\delta_{\Sigma}(\lambda M) = \frac{\gamma_{\Sigma}(M) -1}{\gamma_{\Sigma}(M)+1}.
\end{equation}
\end{lemma}
\begin{proof}
  See \Cref{sec:link_RIP_RC}.
\end{proof}

 Thus, for cones, RIP-based compliance measures have equivalent RC-based compliance measures such that
\begin{equation}\label{eq:TranslateRIPtoRC}
\gamma_{\Sigma}(R) = \frac{1+\delta_{\Sigma}(R)}{1-\delta_{\Sigma}(R)}\qquad\text{and}\qquad
\delta_{\Sigma}(R) = \frac{\gamma_{\Sigma}(R)-1}{\gamma_{\Sigma}(R)+1}.
\end{equation}
The sharp RIP constant~\eqref{eq:DefSharpRIPCst}, the necessary RIP constant~\eqref{eq:DefRIPCstNec} and the sufficient RIP constant $\delta^{\mathtt{suff}}_\Sigma(R)$ of \cite{Traonmilin_2016} are associated to

\begin{align}\label{eq:ComplianceRIPisRCsharp}
  \gamma^{\mathtt{sharp}}_\Sigma(R)
  &:= \inf_{M : \ker M \cap \sT_R(\Sigma) \neq \{0\} } \gamma_{\Sigma}(M)
 = \frac{1+\delta^{\mathtt{sharp}}_\Sigma(R)}{1-\delta^{\mathtt{sharp}}_\Sigma(R)},\\
\label{eq:ComplianceRIPisRCnec}
  \gamma^{\mathtt{nec}}_\Sigma(R)
  &:=\inf_{z \in  \sT_R(\Sigma) \setminus \{0\} } \gamma_{\Sigma}(I-\Pi_z)
 = \frac{1+\delta_\Sigma^{\mathtt{nec}}(R)}{1-\delta^{\mathtt{nec}}_\Sigma(R)},\\
\label{eq:ComplianceRIPisRCsuff}
  \gamma^{\mathtt{suff}}_\Sigma(R)
  &:= \frac{1+\delta_\Sigma^{\mathtt{suff}}(R)}{1-\delta^{\mathtt{suff}}_\Sigma(R)}.
\end{align}
We deduce from~\eqref{eq:IneqRIPConstants} the inequalities 
\begin{equation}\label{eq:IneqRCConstants}
\gamma_\Sigma^{\mathtt{suff}}(R) \leq  \gamma^{\mathtt{sharp}}_\Sigma(R)  \leq \gamma_\Sigma^{\mathtt{nec}}(R). 
\end{equation}

The following result shows that  $\gamma^{\mathtt{sharp}}_\Sigma(R)$ can be simplified.

\begin{proposition}[Explicit expression of $\gamma^{\mathtt{sharp}}_\Sigma(R)$.]\label{prop:explicit}
Consider a cone $\Sigma \subseteq \sH$.
Let $\mathcal{P}$ be the set of symmetric positive semi-definite (PSD) linear operators on $\sH$: $N \in \mathcal{P}$ if and only if $N^{H} = N$ and $N \succeq 0$. For $z \in \sH \setminus \{0\}$ define
\begin{equation}
f^{RC}_{\Sigma}(z) := \inf_{N \in \mathcal{P}: \ker N = \tspan(z) } \gamma_{\Sigma}(N)
\end{equation}
and for any non-empty set $\mathcal{T} \subseteq \sH$ such that $\mathcal{T} \neq \{0\}$ define
\begin{equation}
f^{RC}_{\Sigma}(\mathcal{T}) := \inf_{z \in \mathcal{T}\backslash\{0\}} f^{RC}_{\Sigma}(z).
\end{equation}
We have
\begin{equation}
\inf_{M: \ker M \cap \mathcal{T} \neq \{0\}} \gamma_{\Sigma}(M) = 
f^{RC}_{\Sigma}(\mathcal{T}).
\end{equation}
\end{proposition}
\begin{proof}
This is an immediate consequence of \Cref{lem:best_op_nec_RIP_dim1} in \Cref{sec:link_RIP_RC}.

\end{proof}

Using $\sT = \sT_{\Sigma}(R)$, the sharp RC (or RIP) constant is the smallest RC constant of positive symmetric definite PSD operators with kernels of dimension 1 for which recovery of $\Sigma$ fails:
\begin{equation}
 \gamma^{\mathtt{sharp}}_\Sigma(R)  = f^{RC}_{\Sigma}( \sT_R(\Sigma)).
\end{equation}
Since $I-\Pi_{z} \in \mathcal{P}$ for any nonzero $z$, we have $f^{RC}_{\Sigma}(z) \leq \gamma_{\Sigma}(I-\Pi_{z})$ hence we recover the inequality
\[
 \gamma^{\mathtt{sharp}}_\Sigma(R) \leq  \inf_{z \in \sT_R(\Sigma) \setminus \{0\}} \gamma_{\Sigma}(I-\Pi_z ) = \gamma^{\mathtt{nec}}_{\Sigma}(R),
 \]
however it is an open question to determine whether this is an equality in particular cases or in general. The constant $\gamma^{\mathtt{nec}}_{\Sigma}$ is already  sharp in the following sense: for sparse recovery with the $\ell^1$-norms, as well as for low-rank recovery with the nuclear norm, it is the biggest possible RIP constant ($\delta_\Sigma^{\mathtt{suff}}(R) =\frac{1}{\sqrt{2}}$) that guarantees uniform recovery with  $\|\cdot\|_1$ (respectively with the nuclear norm) for {\em all} sparsities  $k$  (for all rank levels $r$ respectively)~\cite{Davies_2009}.

Similarly, to the compliance measures from \Cref{sec:comp_general}, we can deduce an optimal regularizer after an isometric linear  transformation of the model.

\begin{lemma}[Invariance of $\gamma^{\mathtt{sharp}}_\Sigma(R)$ under linear isometries.]\label{lem:transf_sigma_RC}
Consider a cone $\Sigma \subseteq \sH$, an arbitrary regularizer  $R$ such that $\Sigma \subseteq \dom(R)$, and a (linear) isometry $F$.
We have
\begin{equation}
\gamma_{F\Sigma}^{\mathtt{sharp}}(R \circ F^{-1}) =   \gamma_{\Sigma}^{\mathtt{sharp}}(R).
\end{equation}
Hence, for any class $\sC'$ of regularizers,
 \begin{equation}
  R^* \in \arg \max_{R \in \sC'} \gamma_{\Sigma}^{\mathtt{sharp}}(R)  
  \Leftrightarrow
  R^*\circ F^{-1}  \in \arg \max_{R' \in \sC'} \gamma_{F\Sigma}^{\mathtt{sharp}}(R') .
\end{equation}
\end{lemma}
\begin{proof}
 See \Cref{sec:link_RIP_RC}.
\end{proof}

\subsection{Compliance measures based on necessary RC conditions}\label{sec:nec_RIP}
In this section, we characterize the compliance measure
\begin{equation}
\gamma^{\mathtt{nec}}_\Sigma(R) =  \inf_{z \in \sT_R(\Sigma) \setminus \{0\}} \gamma_{\Sigma}(I-\Pi_z ).
\end{equation}

 To show optimality of the $\ell^1$-norm for sparse recovery and of the nuclear norm for low-rank recovery, we will use the following characterization of 
$ \gamma^{\mathtt{nec}}_\Sigma(R) $
when $\Sigma$ is a cone.

\begin{lemma}[Characterization of $\gamma^{\mathtt{nec}}_\Sigma(R)$ for a cone.]\label{lem:expr_RC_sparse}
Consider a cone $\Sigma \subseteq \sH$ such that $ \Sigma \neq \{0\} $
 and $R$ an arbitrary regularizer  such that $\Sigma \subseteq \dom(R)$. 

\begin{enumerate}
\item If there is $x \in  \sH$ such that $\Sigma \subseteq \tspan(x)$, then 
\begin{equation}
\gamma^{\mathtt{nec}}_\Sigma(R)=
\begin{cases}
+\infty \; &\text{if}\; \sT_R(\Sigma) \subseteq \Sigma, \\
1 \; & \text{otherwise.} \\
\end{cases}
\end{equation}
\item If $\Sigma \subsetneq \tspan(x)$ for every $x \in \sH$, then
\begin{equation}\label{eq:TranslateRCtoB}
\gamma^{\mathtt{nec}}_\Sigma(R)=
\frac{ 1 }{ 1- \inf_{ z \in \sT_R(\Sigma) \setminus \{0\}} \sup_{x\in (\Sigma-\Sigma)\cap S(1)}\frac{\ls x,z\rs^2}{\|z\|_\sH^2}}.
\end{equation}
\end{enumerate}
\end{lemma}

\begin{proof}
See \Cref{sec:proofuos}.
\end{proof}

To go further, we exploit two assumptions related to orthogonal projections on certain sets.
\begin{definition}[Orthogonal projection.]
For any set $E$ we define, for all $z \in \sH$
 \begin{equation}
  P_E(z) =  \arg\min_{y \in E} \|z-y\|_\sH .
 \end{equation}
 This is a set-valued operator is called the \textbf{orthogonal projection}, and $P_{E}(z)$ may be empty if the minimum is not achieved. 
\end{definition}
Some assumptions on $E$ ensure that $P_{E}(z)$ is not empty for any $z$.
\begin{lemma}[Existence of the projection.]\label{lem:ProjectionExists}
Consider a union of subspaces $E \subseteq \sH$, and assume that $E \cap S(1)$ is compact. Then for every $z \in \sH$, $P_{E}(z) \neq \emptyset$. Moreover, for every $x,x' \in P_E(z)$ we have $\|z-x\|_{\sH}^{2} = \|z-x'\|_{\sH}^{2}$ and $\langle z,x\rangle = \|x\|_{\sH}^{2} = \|x'\|_{\sH}^{2}  = \langle z,x'\rangle$, hence the notations $\|z-P_{E}(z)\|_{\sH}^{2}$, $\langle z,P_{E}(z)\rangle$ and $\|P_{E}(z)\|_{\sH}^{2}$ are unambiguous.  We  also have $\|z\|_{\sH}^{2} = \|z-P_{E}(z)\|_{\sH}^{2}+\|P_{E}(z)\|_{\sH}^{2}$ and 
\[
\langle z,P_{E}(z)\rangle = \|P_{E}(z)\|_{\sH}^{2} = \sup_{x \in E \cap S(1)} |\langle x,z\rangle|^{2}.
\]
\end{lemma}
\begin{proof}
 See \Cref{sec:proofuos}.
\end{proof}

Even if $E$ is a union of subspaces and $E \cap S(1)$ is compact, $P_{E}(z)$ may not always be a singleton. For example, consider $E$ the set of $k$-sparse vectors and $z$ the vector with all entries equal to one.

Thanks to \Cref{lem:ProjectionExists}, we have the following characterization of the maximizers of $\delta_\Sigma^{\mathtt{nec}}$.

\begin{corollary}[Characterization of $\delta_\Sigma^{\mathtt{nec}}$.]\label{cor:RCnecUoS}
Consider a cone $\Sigma \subset \sH$ and assume that $\Sigma-\Sigma$ is a union of subspaces, $(\Sigma-\Sigma) \cap S(1)$ is compact, and  $\Sigma \neq \tspan(x)$ for each $x \in \Sigma$.
For any class $\sC'$
 of regularizers such that $\Sigma \subseteq \mathtt{dom}(R)$ for every $R \in \sC'$, the set of maximizers of $\delta_\Sigma^{\mathtt{nec}}(\cdot)$ satisfies (whether this set of maximizers is empty)
 \begin{equation}\label{eq:DefBSigmaRGeneric}
\arg \max_{R\in\sC'} \delta_\Sigma^{\mathtt{nec}}(R)=   \arg\min_{R \in \sC'} B_{\Sigma}(R)\qquad \text{with}\qquad 
B_\Sigma(R) := \underset{z\in \sT_R(\Sigma) \setminus \{0\} }{\sup}  \frac{\|z-P_{\Sigma-\Sigma}(z)\|_\sH^2}{\|P_{\Sigma-\Sigma}(z)\|_\sH^2}.
\end{equation}
For any  regularizer $R$ we have
\begin{equation}
\delta_{\Sigma}^{\mathtt{nec}}(R) = (1+2B_{\Sigma}(R))^{-1}.
\end{equation}
\end{corollary}
\begin{proof}
See \Cref{sec:proofuos}.
\end{proof}

We now have the tools to study optimality for sparse and low rank models. 

\paragraph{Optimal regularization for sparse recovery and for low-rank recovery}
Consider now $\Sigma = \Sigma_k$ the set of $k$-sparse vectors in $\sH = \bR^n$ (associated with the canonical scalar product $\ls\cdot,\cdot\rs$ and the $\ell^2$-norm $\|\cdot\|_\sH = \|\cdot\|_2$), where $1 \leq k \leq n/2$, $n \geq 3$. We have $\Sigma - \Sigma = \Sigma_{2k}$ (for $n \leq 2k$, in particular for $n \leq 2$ and any $k \geq 1$, uniform recovery is not possible for non-invertible $M$: as $\Sigma -\Sigma = \bR^n$, by \Cref{lem:ideal_decoder_descent} we have $\sT_{R}(\Sigma) = \bR^{n}$ for any regularizer, thus $\sT_{R}(\Sigma) \cap \ker M = \{0\}$ implies $\ker M = \{0\}$). It is well-known that $\Sigma$ and $\Sigma-\Sigma$ are both unions of subspaces (hence $\Sigma$ is a cone), and it is straightforward that $(\Sigma-\Sigma) \cap S(1)$ is compact and $\Sigma$ is not reduced to a single line.
Moreover, for any nonzero $z \in \mathbb{R}^{n}$, $P_{\Sigma-\Sigma}(z)$ contains the restriction $z_{T_{2}}$ of $z$ to any set $T_2 = T_{2}(z) \subseteq \{1,\ldots,n\}$ of size $2k$ such that $\min_{i \in T_{2}}|z_i|\geq \max_{j \in T_{2}^{c}}|z_j|$. 
Hence, we can invoke \Cref{cor:RCnecUoS} to replace the maximization of $\delta_\Sigma^{\mathtt{nec}}(R)$ by the minimization of 
\begin{equation}\label{eq:PropBSigmaRKSparse}
B_\Sigma(R) = \underset{z\in \sT_R(\Sigma) \setminus \{0\} }{\sup}  \frac{\|z_{T_2^c}\|_2^2}{\|z_{T_2}\|_2^2}.
\end{equation}

Similarly, We consider $\Sigma = \Sigma_r$ the set of matrices of rank at most $r$ in the Hilbert space $\sH$ of  $n \times n$ symmetric matrices (we study the symmetric case for simplicity, but conjecture that our result can be extended to the non-symmetric case) with $\|\cdot\|_\sH = \|\cdot\|_F$ (the Frobenius norm).  We have again $\Sigma - \Sigma = \Sigma_{2r}$ and all conditions are satisfied such that \Cref{cor:RCnecUoS}  can be invoked. Denoting $\Delta = \eig(z)$ the vector of eigenvalues of matrix $z \in \sH$ sorted by decreasing absolute value, so that $z = U^T \Delta U$ for some unitary matrix $U$, and defining $z_{T} := z = U^T \Delta_{T} U$ for every index set $T$, we have $P_{\Sigma-\Sigma}(z) = z_{T_2}$ and $z-P_{\Sigma-\Sigma}(z) = z_{T_2^{c}}$ where 
$T_2 = T_{2}(z) \subseteq \llbracket 1,n\rrbracket$ is any index set containing the $2k$ largest 
eigenvalues (in magnitude) of $z$, \ie such that $\min_{i \in T_{2}} |\Delta_i| \geq \max_{j \in T_{2}^{c}} |\Delta_j|$. With these observations and notations, we are again left to optimize~\eqref{eq:PropBSigmaRKSparse}.

Specializing to the class $\sC$ of convex, coercive, continuous regularizers, we obtain the following theorems. 
\begin{theorem}[Optimality of $\ell^1$-norm for $k$-sparse vectors for $\delta_\Sigma^{\mathtt{nec}}$.]\label{th:RIP_nec_atom}
With $k$-sparse vectors, $\Sigma = \Sigma_k \subseteq  \sH = \bR^n$, $k < \frac{n}{2}$, 
and $R^{\star}(\cdot) = \|\cdot\|_{1}$, we have

  \begin{equation}
\delta_\Sigma^{\mathtt{nec}}(R^{\star}) = \sup_{R\in \sC} \delta_\Sigma^{\mathtt{nec}}(R) 
= 
(2 B^{\star}_{k,n}+1)^{-1}\qquad\text{with}\qquad B^{\star}_{k,n} := \max_{1 \leq L \leq n-2k} \frac{L/k}{(L/k+1)^{2}+1} .
 \end{equation}
  Moreover,  for $k=1$,  the unique functions $R \in \sC_{\Sigma}$ maximizing  $\delta_\Sigma^{\mathtt{nec}}$  are of the form $R(\cdot) =  \lambda \|\cdot\|_1 $ with $\lambda > 0$.
\end{theorem}

\begin{theorem}[Optimality of the nuclear norm for rank-$r$ matrices for $\delta_\Sigma^{\mathtt{nec}}$.]\label{th:RIP_nec_atom_rank}
 With the set of rank-$r$ matrices, $\Sigma = \Sigma_r$, in the space $\sH$ of symmetric $n \times n$ matrices, $r < \frac{n}{2}$,  and with $R^{\star}(\cdot) = \|\cdot\|_{*}$ (the nuclear norm), we have

  \begin{equation}
\delta_\Sigma^{\mathtt{nec}}(R^{\star}) = \sup_{R\in \sC} \delta_\Sigma^{\mathtt{nec}}(R) 
= 
(2 B^{\star}_{r,n}+1)^{-1}\qquad\text{with}\qquad B^{\star}_{r,n} := \max_{1 \leq L \leq n-2r} \frac{L/r}{(L/r+1)^{2}+1} .
 \end{equation}
\end{theorem}
The proofs of these two theorems exploits  technical lemmas that we detail in \Cref{sec:proofsparsity} and \Cref{sec:proofLRnex}.

\begin{proof} 
We give the proof for the case of sparse recovery. To express it for low-rank recovery simply replace the 
notation $k$ by $r$. For $1 \leq s \leq n$ and any regularizer  $R$ we define 
 \begin{equation}\label{eq:DefBsSigma}
 B_{\Sigma}^{s}(R) := \underset{z\in \sT_R(\Sigma) \setminus \{0\}, 
 z \in \Sigma_{s}
}{\sup}  \frac{\|z_{T_2^c}\|_2^2}{\|z_{T_2}\|_2^2}.
 \end{equation}
For $s \leq 2k$ and any $z \in \Sigma_{s}$ we have $z_{T_{2}^{c}}=0$ hence $B_{\Sigma}^{s}(R) = 0$, thus
 $B_{\Sigma}(R) = \max_{1 \leq L \leq n-2k} B_{\Sigma}^{2k+L}(R)$.
First consider $R\in \sC_\Sigma$. Since $R$ is positively homogeneous and subadditive, by \Cref{lem:charact_supBL2_atom} for $\Sigma_{k}$ / \Cref{lem:charact_supBL2_atom_rank} for $\Sigma_{r}$,
\[
B_{\Sigma}^{2k+L}(R) \geq \frac{ \frac{L}{k}}{ \left(\frac{L}{k}+1\right)^2 + 1 },
\qquad\text{for each}\ 1 \leq L \leq n-2k.
\]
For $R^{\star}$ and $1 \leq L \leq n-2k$ we also have 
(\Cref{lem:charact_supBL2_l1} / \Cref{lem:charact_supBL2_nuc}, inspired by \cite{Davies_2009}) that
  \[
  B_{\Sigma}(R^{\star}) = \max_{1\leq L \leq n-2k}\frac{ \frac{L}{k}}{ \left(\frac{L}{k}+1\right)^2 + 1 }.
  \]
As a result,
\[
B_{\Sigma}(R) \geq  B_\Sigma(R^{\star}) = \max_{1 \leq L \leq n-2k} \frac{ \frac{L}{k}}{ \left(\frac{L}{k}+1\right)^2 + 1 } =: B^{\star}_{k,n}
\]

Finally,  remark that $B_\Sigma(R)$ is an increasing function of $\sT_R(\Sigma)$. Using \Cref{lem:atomic_necessary}, for any $R\in\sC$  there is $R '\in \sC_\Sigma$ such that 
\[
B_{\Sigma}(R) \geq B_{\Sigma}(R')  \geq B^{\star}_{k,n}.
\]

For $k=1$, uniqueness comes from the fact that on a given orthant for $ R \in \sC_\Sigma$,  $R$ is a weighted $\ell^1$ norm: $R((x_1, \ldots, x_n)) = \sum_i w_i |x_i|$ and the equality  case  in \Cref{lem:charact_supBL2_atom} forces   $ w_i  = \max_i w_{i}$.

\end{proof}

Because of the uniqueness result for $k=1$,
the $\ell^1$-norm is the unique convex regularizer  in $\cap \sC_{\Sigma_k}$ that jointly maximizes  $\delta_{\Sigma_k}^{\mathtt{nec}}$ for all $k < \frac{n}{2}$ (the proof of \Cref{th:RIP_nec_atom} is valid for  $\sC_{\Sigma_{k'}}$, with $k \leq k' <\frac{n}{2}$). It is an open question to determine if  we have uniqueness model by model. As the result might change for tighter compliance measures, we leave this question for future work.

In the next section,  we will see that the optimization of the sufficient RIP constant leads to very similar expressions.

\subsection{Compliance measures based on sufficient RC conditions }\label{sec:suff_RIP}
When $\Sigma$ is a union of subspaces and $R$ is an arbitrary regularizer, an ``explicit'' RIP constant $\delta_\Sigma^{\mathtt{suff}}(R)$ is sufficient to guarantee reconstruction \cite{Traonmilin_2016}.
The expression of this constant \cite{Traonmilin_2016}[Eq. (5)] is recalled in the Appendix (Equation~\eqref{eq:ExplicitExpressionSuffRIP}) and can be used as a compliance measure.
It gives rise to the following lemma, which is proved in \Cref{sec:ProofSuffRIP}.

\begin{lemma}[Equality case of the sufficient conditions.]\label{lem:charac_suff_RIP}
Assume that $\Sigma = \cup_{V \in \mathcal{V}} V$ is a union of subspaces and that $\Sigma \cap S(1)$ is compact. Consider $R$ any  regularizer  such that $\Sigma \subseteq \dom(R)$. We have
\begin{equation}\label{eq:RIPsuffVsRIPsuff2}
\begin{split}
 \delta_\Sigma^{\mathtt{suff}}(R) & \geq  \frac{1}{ \sqrt{ \underset{z\in \sT_R(\Sigma) \setminus \{0\} }{\sup}  \frac{\|z -P_\Sigma(z)\|_\Sigma^2}{\|P_\Sigma(z)\|_2^2} + 1 }} =: \delta_\Sigma^{\mathtt{suff2}}(R).
 \end{split}
\end{equation}

Further, assume that $P_\Sigma(z) \subseteq \arg \min_{x \in \Sigma} \|x-z\|_\Sigma$  for every $z \in \sH$ and that, for every $V \in \mathcal{V}$ and every $u\in \Sigma$, $P_{V^\perp}(u) \in \Sigma$. 
Then, there is equality in~\eqref{eq:RIPsuffVsRIPsuff2}. 
\end{lemma}

\begin{proof}
 See \Cref{sec:ProofSuffRIP}. Note that the assumption $P_\Sigma(z) \subseteq \arg \min_{x \in \Sigma} \|x-z\|_\Sigma$ could be replaced by the slightly weaker $P_\Sigma(z) \cap \arg \min_{x \in \Sigma} \|x-z\|_\Sigma/\|x\|_{2} \neq \emptyset$. 

\end{proof}

We get an immediate corollary of the first claim in the above lemma.
\begin{corollary}[Expression of a sufficient condition.]\label{cor:suffcond}
Assume that $\Sigma = \cup_{V \in \mathcal{V}} V$ is a union of subspaces and that $\Sigma \cap S(1)$ is compact. For any class $\sC'$ of regularizers such that $\Sigma \subseteq \mathtt{dom}(R)$ for every $R \in \sC'$, the set of maximizers of $\delta_\Sigma^{\mathtt{suff2}}(\cdot)$ satisfies (whether this set of maximizers is empty)
 \begin{equation}\label{eq:DefDSigmaRGeneric}
\arg \max_{R\in\sC'} \delta_\Sigma^{\mathtt{suff2}}(R)=   \arg\min_{R \in \sC'} D_{\Sigma}(R)\qquad \text{with}\qquad 
D_\Sigma(R) := \underset{z\in \sT_R(\Sigma) \setminus \{0\} }{\sup}  \frac{\|z-P_{\Sigma}(z)\|_\Sigma^2}{\|P_{\Sigma}(z)\|_\sH^2}.
\end{equation}
For any optimal regularizer $R^{\star}$ we have 
\begin{equation}
\delta_\Sigma^{\mathtt{suff2}}(R^{\star}) = (1+D_{\Sigma}(R^{\star}))^{-1/2}.
\end{equation}
\end{corollary}
Note the subtle difference in the norm at the numerator in $B_{\Sigma}(R)$ and $D_{\Sigma}(R)$.

\paragraph{Optimal regularization for sparse recovery and low-rank recovery}
When considering sparse recovery or low-rank recovery, there is indeed equality $\delta_\Sigma^{\mathtt{suff}}(R) =  \delta_\Sigma^{\mathtt{suff2}}(R)$ thanks to the following Lemma.
\begin{lemma} \label{lem:assumption_suff} 
The assumptions for the equality case of  \Cref{lem:charac_suff_RIP} hold for $\Sigma = \Sigma_k$ the set of $k$-sparse vectors in $\sH = \bR^{n}$, as well as for the set $\Sigma = \Sigma_r$ of symmetric matrices of rank at most $r$ in $\sH$ the set of symmetric $n \times n$ matrices.
\end{lemma}
\begin{proof}
 See \Cref{sec:ProofSuffRIP}.
\end{proof}

Consider $\Sigma := \Sigma_k$, and regularizers in $\sC_\Sigma$.
Similarly to the necessary case, from \Cref{lem:charac_suff_RIP},  we have (when $\Sigma$ is a union of subspace and $\Sigma \cap S(1)$ is closed)
\begin{equation}
 D_\Sigma(R) 
= \underset{z\in \sT_R(\Sigma) \setminus \{0\} }{\sup}  \frac{\|z_{T^c}\|_\Sigma^2}{\|z_T\|_2^2}
\end{equation}
where $T$ denotes  the  support of the $k$ largest coordinates of $z$.

We obtain similar results as in the necessary RIP constant case. 

\begin{theorem}[Optimality of $\ell^1$-norm for $k$-sparse vectors for $ \delta_\Sigma^{\mathtt{suff}}$.]\label{th:RIP_suff_atom}
With $k$-sparse vectors, $\Sigma = \Sigma_k \subseteq  \sH = \bR^n$, $k < \frac{n}{2}$, 
and $R^{\star}(\cdot) = \|\cdot\|_{1}$, we have

  \begin{equation}
\delta_\Sigma^{\mathtt{suff}}(R^{\star}) = \sup_{R\in \sC} \delta_\Sigma^{\mathtt{suff}}(R) 
= 
\frac{1}{\sqrt{2}} .
 \end{equation}
  Moreover,  for $k=1$,  
   the unique functions $R \in \sC_{\Sigma}$ maximizing  $\delta_\Sigma^{\mathtt{suff}}$  are of the form $R(\cdot) =  \lambda \|\cdot\|_1 $ with $\lambda > 0$.
\end{theorem}
\begin{theorem}[Optimality of the nuclear norm for rank-$r$ matrices for $ \delta_\Sigma^{\mathtt{suff}}$.]\label{th:RIP_suff_atom_nuclear}
With the set of rank-$r$ matrices, $\Sigma = \Sigma_r$, in the space $\sH$ of symmetric $n \times n$ matrices, $r < \frac{n}{2}$, 
and with $R^{\star}(\cdot) = \|\cdot\|_{*}$ (the nuclear norm), we have

  \begin{equation}
\delta_\Sigma^{\mathtt{suff}}(R^{\star}) = \sup_{R\in \sC} \delta_\Sigma^{\mathtt{suff}}(R) 
= 
\frac{1}{\sqrt{2}} .
 \end{equation}
\end{theorem}
\begin{proof} 
We give the proof for the case of sparse recovery. To express it for low-rank recovery simply replace the 
notation $k$ by $r$. For $1 \leq s \leq n$ and any regularizer  $R$ we define
 \[
 D_{\Sigma}^{s}(R) := \underset{z\in \sT_R(\Sigma) \setminus \{0\}, 
 z \in \Sigma_{s}
}{\sup}  \frac{\|z_{T^c}\|_\Sigma^2}{\|z_{T}\|_2^2}.
 \]
For $s \leq k$ and any $z \in \Sigma_{s}$ we have $z_{T^{c}}=0$ hence $D_{\Sigma}^{s}(R) = 0$, thus
 $D_{\Sigma}(R) = \max_{1 \leq L \leq n-k} D_{\Sigma}^{k+L}(R)$.

 First consider $R\in \sC_\Sigma$. Since $R$ is positively homogeneous and subadditive, by \Cref{lem:charact_supDL2_atom} for $\Sigma_{k}$ / \Cref{lem:charact_supDL2_atom_nuclear} for $\Sigma_{r}$,
\[
D_{\Sigma}^{k+L}(R) \geq  \min(1, \frac{L}{k}),
\qquad\text{for each}\ 1 \leq L \leq n-k.
\]
For $R^{\star}$ and $1 \leq L \leq n-k$ we also have 
(with \Cref{lem:sup_DL_l1} / \Cref{lem:sup_DL_nuclear}) that
  \[
  D_{\Sigma}^{k+L}(R^{\star}) = \min(1, \frac{L}{k}).
  \]
As a result,
\[
D_{\Sigma}(R) \geq  D_\Sigma(R^{\star}) = \max_{1 \leq L \leq n-k}  \min(1, \frac{L}{k}) =1 .
\]

Finally,  remark that $D_\Sigma(R)$ is an increasing function of $\sT_R(\Sigma)$. Using \Cref{lem:atomic_necessary}, for any $R\in\sC$  there is $R '\in \sC_\Sigma$ such that 
\[
D_{\Sigma}(R) \geq D_{\Sigma}(R')  \geq 1.
\]

\end{proof}

\subsection{Discussion}

Even without an analytic expression of the sharp RIP constant, it would have been possible to show  that $R^\star$  optimizes $\delta_{\Sigma}^{\mathtt{sharp}}$ if it were  simultaneously optimizing its lower and upper bound, \ie if we had
\begin{equation}
\sup_{R \in \sC} \delta_{\Sigma}^{\mathtt{suff}}(R)
  = 
 \delta_{\Sigma}^{\mathtt{suff}}(R^{\star})
 = 
 \delta_{\Sigma}^{\mathtt{nec}}(R^{\star})
 =
 \sup_{R \in \sC} \delta_{\Sigma}^{\mathtt{nec}}(R).
\end{equation}

Unfortunately, this is not the case in the sparse and low rank examples. We observe that for fixed $k,n$ we have in both  cases
\begin{equation}
\frac{1}{\sqrt{2}} = \delta_\Sigma^{\mathtt{suff}}(R^{\star}) < \delta_\Sigma^{\mathtt{nec}}(R^{\star}).
\end{equation}

Because of this fact, we cannot conclude on the optimality of $R^\star$ for $\delta_\Sigma^{\mathtt{sharp}}$. 
However,  indexing all objects of the problem by  $n$ the dimension of $\sH$ (respectively the dimension of the diagonals): the set of regularizers $\sC^{(n)}$, the models $\Sigma_k^{(n)}$ and the corresponding $R^{\star, (n)}$ (independent of $k$ for $k<n/2$ as we saw previously). We  have (see \Cref{rk:opt_B_kL})
\begin{equation}
\inf_{n \geq 3}\inf_{k \in \{1,\ldots, \lfloor n/2 \rfloor \}} \sup_{R\in\sC^{(n)}}\delta_{\Sigma_k^{(n)}}^{\mathtt{nec}}(R)  = \frac{1}{\sqrt{2}} = \delta_{\Sigma_k^{(n)}}^{\mathtt{suff}}(R^{\star, (n)}).
\end{equation}

We deduce
\begin{equation}
\inf_{n \geq 3} \inf_{k \in \{1,\ldots, \lfloor n/2 \rfloor \}} \sup_{R\in\sC^{(n)}}\delta_{\Sigma_k^{(n)}}^{\mathtt{sharp}}(R) = \frac{1}{\sqrt{2}} .
\end{equation}
and  
\begin{equation}
 \inf_{n \geq 3} \inf_{k \in \{1,\ldots, \lfloor n/2 \rfloor \}}  \left| \delta_{\Sigma_k^{(n)}}^{\mathtt{sharp}}(R^{\star,(n)})- \left[ \sup_{R\in\sC^{(n)}}\delta_{\Sigma_k^{(n)}}^{\mathtt{sharp}}(R)\right]  \right| = 0.
\end{equation}

This shows that the functions $R^{\star,(n)}$  are optimal as a family with respect to a family of models $\Sigma_k^{(n)}$  and the worst case of their associated compliance measures $\delta_{\Sigma_k^{(n)}}^{\mathtt{sharp}}(R) $.

These results can be interpreted in terms of number of measurements needed to recover uniformly sparse or low rank objects with convex regularization. Under the \emph{best known} (RIP-based) uniform recovery conditions, it is guaranteed that using the optimal regularization with respect to RIP-based  compliance measures will enable the use of fewer measurements. In particular in the case of an operator $M$ built from $m$ random Gaussian measurements, it has been proven (see e.g.~\cite{Foucart_2013}) that there is a universal constant $C$ such that if $m \geq C\frac{k \log(k/n)}{\delta^{2}}$ then 
 $M$ satisfies a prescribed RIP constant $\delta$ with high probability. Hence, the larger the required RIP constant is, the lower the number of measurement needs to be. Such results on the required number of measurement to verify the RIP have been extended to more general low dimensional models (see e.g. \cite{Puy_2015}), making RIP-based optimal regularizers tools of choice to optimize the number of random measurements of elements of a given low dimensional model.

\section{Towards the construction of optimal convex regularizers? The examples of sparsity in levels and beyond.}\label{sec:sparsity_in_levels}

In the previous Section, optimality was shown by exhibiting the optimal regularizer ($\ell^1$-norm and nuclear norm). The only constructive part in these  results is given by \Cref{th:simpl_max_adeq} that shows that we can look for optimal regularizers in the set of atomic norms $\sC_\Sigma$ constructed using the model set $\Sigma$. Ideally,  given a compliance measure, we would like to be able to construct for any model $\Sigma$, an optimal regularizer $R^\star\in \sC_\Sigma$. As such an objective seems out of reach with the tools we have developed so far, we study on an example (the case of {\em sparsity in levels}) the simpler problem of looking for the optimal regularizer  in a  smaller set of regularizers. We consider a set of weighted $\ell^1$-norms and explore the explicit construction of an optimal regularizer  for the compliance measure  $\delta_\Sigma^{\mathtt{nec}}$. We then extend this result to the similar setting of Cartesian product of sparse and low-rank models.

\subsection{Sparsity in levels}
Sparsity \emph{in levels} is a possible extension of plain sparsity, which is useful for the precise modeling of signals such as medical images \cite{Adcock_2013b,Bastounis_2015}. It is also useful for simultaneous modeling of sparse signal and sparse noise \cite{Studer_2013,Traonmilin_2015}.

\begin{definition}[Sparsity in levels.]
In $\sH = \bR^{n_1} \times  \bR^{n_2}  \times \ldots \times  \bR^{n_L}  $, given sparsity levels $k_1,\ldots,k_L$, we define the \textbf{sparsity in levels} model with
\begin{equation}
\Sigma_{k_1,\ldots,k_L} := \{ x = (x_1, \ldots, x_L) :  x_{i} \in \Sigma_{k_i} \}  
\end{equation}
where $\Sigma_{k_i}$ is the $k_i$-sparse model in $\bR^{n_i}$.
\end{definition}

While the $\ell^1$-norm was shown to be is a candidate to perform regularization for models that are sparse in levels \cite{Adcock_2013b}, it was also shown that it is possible to obtain better sufficient RIP recovery guarantees when weighting the $\ell^1$ norm by $\sqrt{k_i}$ in each level \cite{Traonmilin_2016}.  The following theorem permits to show that this weighting scheme is close to optimal for the compliance measure $\delta_\Sigma^{\mathtt{nec}}$  by giving explicit expressions for the calculation of optimal weights.

Given weights $w = (w_1,\ldots,w_L) \in \bRp^L$, we define the $\ell^1$-norm weighted by levels.

\begin{equation}
 \|(x_1,\ldots,x_L)\|_w = \sum_{i=1}^L w_i\|x_i\|_1.
\end{equation}

We have the following result.

\begin{theorem}[Optimal  weighted $\ell^1$  norms for $\delta_\Sigma^{\mathtt{nec}}$ for sparsity in levels.]\label{th:opt_in_levels}
  Consider two integers $k_1,k_2 \geq 2$ and the model set $\Sigma = \Sigma_{k_1,k_2}$ in $\sH = \bR^{n_1}\times \bR^{n_2}$ where we assume that $n_1 \geq 4k_1$, $n_2 \geq 4k_2$. Let $\tilde{a}= 2\sqrt{3}-3$. We define
  \begin{equation}\label{eq:MinDefiningNu12}
   B_\Sigma^{\star} := \min_{\stackrel{\nu_1 \in  [\tilde{a},1-\tilde{a}]}{\nu_2 =1-\nu_1}} \max_{i\in \{1,2\}}\max_{x_i \in \{ \lfloor k_i\sqrt{1+1/\nu_i}\rfloor; \lceil k_i\sqrt{1+1/\nu_i}\rceil\}}\frac{x_i/k_i} {\nu_i ( x_i/k_i+1)^2 +1}
  \end{equation}
 where $\lfloor\cdot \rfloor$ and $\lceil \cdot\rceil$ denote the lower and upper integer part and $(\nu_1^*,\nu_2^*)$ minimizing this expression.

   Then
\begin{equation}
w^{*}   
 \in \arg \max_{
 w
 } \delta_\Sigma^{\mathtt{nec}} (\|\cdot\|_{w
 } )
 \end{equation}
 if and only if $w^{*} =  (w_1^*,w_2^*)$ where $w_{1}^{*},w_{2}^{*}>0$ satisfy
\begin{equation}\label{eq:th_s_in_levels}
\begin{split}
\frac{w_2^*}{w_1^*} &= \sqrt{\frac{k_{1}}{k_{2}} (1/\nu_1^*-1)}.
\end{split}
\end{equation}
Moreover, denoting $w_{0} = w_{0}(k_{1},k_{2}) := (1/\sqrt{k_{1}},1/\sqrt{k_{2}})$ we have
\begin{equation}
 \label{eq:MainThmLevelsBound1}
  \begin{split}
B_{\Sigma 
}( \|\cdot\|_{w^{*}
} ) 
= B_\Sigma^\star 
&\leq B_{\Sigma 
}( \|\cdot\|_{w_{0}} ) \leq (\sqrt{3}-1)/2\\
\delta_\Sigma^{\mathtt{nec}} (\|\cdot\|_{w^{*} 
} ) = (1+2B_\Sigma^\star)^{-1}
& \geq \delta_\Sigma^{\mathtt{nec}} (\|\cdot\|_{w_0} ) \geq 1/\sqrt{3}.
\end{split}
\end{equation}
Finally, we have
\begin{equation} \label{eq:MainThmLevelsBound2}
\inf_{k_{1},k_{2} \geq 1} \inf_{n_{1} \geq 4k_{1},n_{2} \geq 4 k_{2}}
\delta_\Sigma^{\mathtt{nec}} (\|\cdot\|_{w_{0}(k_{1},k_{2})}) = 1/\sqrt{3}.
\end{equation}
\end{theorem}
\begin{proof}
 See \Cref{sec:proofslevels}.

\end{proof}

This theorem comes from the fact that (see proof) the quantity defined in \eqref{eq:PropBSigmaRKSparse} satisfies

\begin{equation}\label{eq:BSigmaL1L2}
 \begin{split}
B_{\Sigma_{k_1,k_2}}( \|\cdot\|_{(w_1,w_2)} )  &= \max_{L_1,L_2} B^{L_1,L_2}_{\Sigma_{k_1,k_2}}((w_1,w_2))\\
\end{split}
\end{equation}
where $B^{L_1,L_2}_{\Sigma_{k_1,k_2}}(\|\cdot\|_{(w_1,w_2)})$ can be computed explicitly (similarly to $B_\Sigma^{2k+L}$ from \eqref{eq:DefBsSigma} for sparsity).

Thanks to the expression of $B_\Sigma(\|\cdot\|_{w^*})$ from \Cref{th:opt_in_levels}, it becomes tractable to evaluate numerically optimal weights. We simply perform the minimization over $\nu_1 \in [\tilde{a}, 1-\tilde{a}]$ over a regular grid (of $10^6$ points in our experiment) and take the minimum. The value of $w_{1}^{*}/w_{2}^{*}$ is returned according to~\eqref{eq:th_s_in_levels}.  Let $w_0 =w_{0}(k_{1},k_{2}) = (1/\sqrt{k_1},1/\sqrt{k_2})$. In \Cref{fig:s_in_levels}, we show a representation of the two criteria $C_1(k_1,k_2)= |1 - \frac{\ls w^*,w_0\rs}{\|w^*\|_2 \|w_0\|_2}|$ and $C_2(k_1,k_2)= |\delta_\Sigma^{\mathtt{nec}}(\|\cdot\|_{w^*})-\delta_\Sigma^{\mathtt{nec}}(\|\cdot\|_{w_0}) |$ for different pairs $(k_1,k_2)$. The case $C_1(k_1,k_2)=C_2(k_1,k_2)=0$ occurs if and only if $w_0$ is optimal).

\begin{figure}[!h]
  \centering
\includegraphics[width=0.45\linewidth]{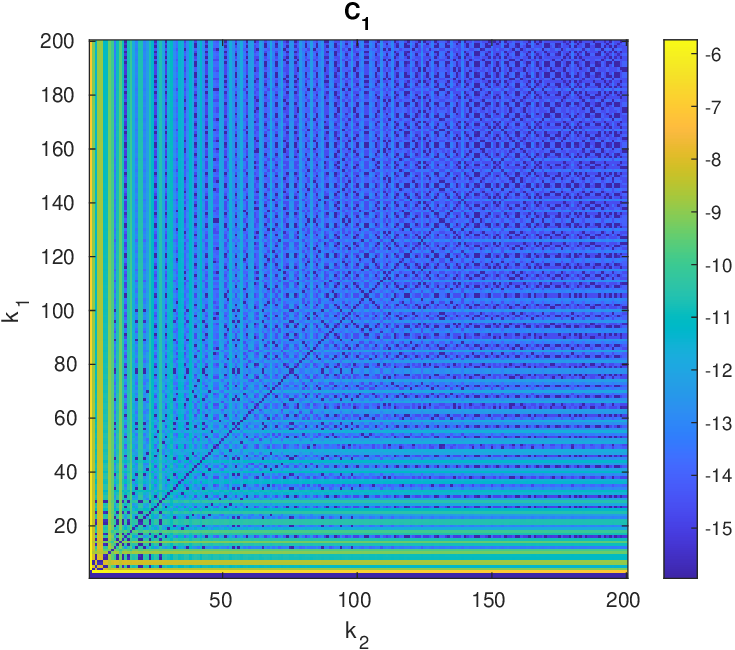}
\includegraphics[width=0.45\linewidth]{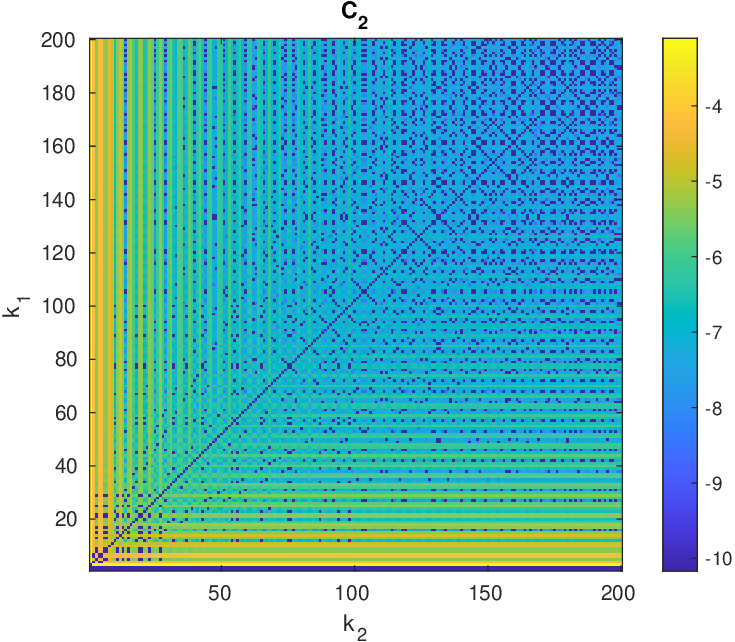}

\caption{Then quantities $\log_{10} (C_1(k_1,k_2)) := \log_{10}\left(|1 - \frac{\ls w^*,w_0\rs}{\|w^*\|_2 \|w_0\|_2}|\right)$ (left) and $\log_{10}(C_2(k_1,k_2)):= \log_{10}(|\delta_\Sigma^{\mathtt{nec}}(\|\cdot\|_{w^*})-\delta_\Sigma^{\mathtt{nec}}(\|\cdot\|_{w_0}) |)$ (right) where $w^{*}=(w_1^*,w_2^*)$ is obtained from \Cref{th:opt_in_levels} and $w_0 = (1/\sqrt{k_1},1/\sqrt{k_2})$ for different $k_1, k_2\geq 2$ . }
  \label{fig:s_in_levels}
\end{figure}

We observe  numerically that for $ 2\leq k_1, k_2\leq 200 $,   $C_1(k_1,k_2)\leq 10^{-5}$ and $C_2(k_1,k_2)\leq 5\cdot 10^{-3}$ and that the error tends to decrease for greater $k_1,k_2$. This comes from the fact that the result of the discrete optimization over the integers $L_i$ in \eqref{eq:BSigmaL1L2} 
 gets closer to the result of a continuous optimization that yields $w_{2}^{*}/w_{1}^{*}= \sqrt{k_1}/ \sqrt{k_2}$ (obtained by dropping the integer parts in \Cref{th:opt_in_levels}).

For the ``asymptotically optimal'' weighting scheme $w_0=w_{0}(k_{1},k_{2}) =\left(\frac{1}{\sqrt{k_1}}, \frac{1}{\sqrt{k_2}}\right)$, we find
\begin{equation}\label{eq:SandwichIneqDeltaSuffLevels}
\inf_{k'_1,k'_2\geq 1, n'_1\geq 4 k_1',n'_2\geq 4 k_2'} \delta_{\Sigma_{k'_1,k'_2}}^{\mathtt{nec}} ( \|\cdot\|_{w_0})
\stackrel{\eqref{eq:MainThmLevelsBound2}}{=} \frac{1}{\sqrt{3}}
\stackrel{(*)}{\leq}
\delta_{\Sigma_{k_1,k_2}}^{\mathtt{suff}}( \|\cdot\|_{w_0}) \leq \delta_{\Sigma_{k_1,k_2}}^{\mathtt{sharp}}( \|\cdot\|_{w_0}) \leq \delta_{\Sigma_{k_1,k_2}}^{\mathtt{nec}}( \|\cdot\|_{w_0})                                                                                                                                                                                                                                                                                             . \end{equation}
The inequality (*) is shown in Theorem~\ref{th:RIP_suff_in_levels} below (improving for $L=2$ the lower bound $\frac{1}{\sqrt{2+L}} =1/\sqrt{4}=1/2$ for sparsity in $L$ levels previously given in \cite[Theorem 4.2]{Traonmilin_2016}),  and the last inequalities are generic, cf \eqref{eq:IneqRIPConstants}.

The double-sided bound~\eqref{eq:SandwichIneqDeltaSuffLevels} confirms that the weighting scheme $\left(\frac{1}{\sqrt{k_1}}, \frac{1}{\sqrt{k_2}}\right)$ is close to an optimal choice (w.r.t the maximization of $\delta_{\Sigma_{k_1,k_2}}^{\mathtt{sharp}}$) when the sparsities are known.

\begin{theorem}[Sufficient RIP condition for near-optimal $\ell^1$ norms for sparsity in levels.]\label{th:RIP_suff_in_levels}
  Consider two integers $k_1,k_2 \geq 2$ and the model set $\Sigma = \Sigma_{k_1,k_2}$ in $\sH = \bR^{n_1}\times \bR^{n_2}$ 
  with $n_{i} \geq k_{i}$, $i =1,2$,
 and the norm $ \|(x_1,x_2)\|_{w} = \sum_{i=1}^2 \frac{1}{\sqrt{k_i}}\|x_i\|_1.$ Then
 \begin{equation}\label{eq:LowerBoundRIPsufflevels}
  \delta_{\Sigma_{k_1,k_2}}^{\mathtt{suff}} ( \|\cdot\|_w) 
  \geq
     \frac{1}{\sqrt{3}}.
 \end{equation}
\end{theorem}
\begin{proof}
 See \Cref{sec:proofslevels}.
\end{proof}

\subsection{Beyond sparsity in levels}
Beyond sparsity in levels, we
 obtain exactly the same result for the Cartesian product of a sparse model and a low-rank model. Consider $\Sigma_{k,r} = \Sigma_k \times \Sigma_r \subset \bR^n \times \sH_p$ where $\sH_p$ is the set of symmetric matrices of size $p \times p$. This model with $n=p^{2}$ can be used to model sums of sparse and low rank matrices. To address associated matrix reconstruction problems it was suggested in \cite{tanner2020compressed} to use a weighted sum of the $\ell^1$-norm and the nuclear norm with weights ratio $\frac{\sqrt{k}}{\sqrt{r}}$, ie $\|(z_{1},z_{2})
 \|_{w} = \frac{1}{\sqrt{k}} \|z_1\|_1 +\frac{1}{\sqrt{r}} \|z_{2}
 \|_{*}$. The following Theorem guarantees that the previous
 numerical experiments hold with this model (by replacing $k_1$ by $k$  and $k_2$ by $r$). It thus confirms that this is a near optimal choice of weights.

\begin{theorem}[Optimal mixed norms for $\delta_\Sigma^{\mathtt{nec}}$ for sparse plus low-rank models.]\label{th:sparse_LR}
  Consider two integers $k,r \geq 2$ and the model set $\Sigma = \Sigma_{k}\times \Sigma_{r}$ in $\sH = \bR^{n}\times \sH_p$ where we assume that $n \geq 4k$, $p \geq 4r$.
   Consider $\tilde{a}= 2\sqrt{3}-3$, $B_\Sigma^\star$ and $(\nu_1^*,\nu_2^*)$ from \Cref{th:opt_in_levels} with $k_1 =k$ and $k_2=r$. Then, with $\|(z_{1},z_{2})
 \|_{w} := w_{1} \|z_1\|_1 + w_{2} \|z_{2}
 \|_{*}$, we have:
\begin{equation}
w^{*}   
 \in \arg \max_{
 w
 } \delta_\Sigma^{\mathtt{nec}} (\|\cdot\|_{w
 } )
 \end{equation}
 if and only if $w^{*} =  (w_1^*,w_2^*)$ where $w_{1}^{*},w_{2}^{*}>0$ satisfy
\begin{equation}
\begin{split}
\frac{w_2^*}{w_1^*} &= \sqrt{\frac{k}{r} (1/\nu_1^*-1)}.
\end{split}
\end{equation}
Moreover, denoting $w_{0} = w_{0}(k,r) := (1/\sqrt{k},1/\sqrt{r})$ we have
\begin{equation}
  \begin{split}
B_{\Sigma 
}( \|\cdot\|_{w^{*}
} )
= B_\Sigma^\star
&\leq B_{\Sigma 
}( \|\cdot\|_{w_{0}} ) \leq (\sqrt{3}-1)/2\\
\delta_\Sigma^{\mathtt{nec}} (\|\cdot\|_{w^{*} 
} ) = (1+2B_\Sigma^\star)^{-1}
& \geq \delta_\Sigma^{\mathtt{nec}} (\|\cdot\|_{w_0} ) \geq 1/\sqrt{3}.
\end{split}
\end{equation}
Finally, we have
\begin{equation}
\inf_{k,r \geq 1} \inf_{n \geq 4k,p \geq 4 r}
\delta_\Sigma^{\mathtt{nec}} (\|\cdot\|_{w_{0}(k,r)}) = 1/\sqrt{3}.
\end{equation}
\end{theorem}
\begin{proof}
 See \Cref{sec:proofslevels}.
\end{proof}

 The resulting weighting scheme for the sparse + low-rank model has been used practically in \cite{guennec2023adaptive}. In the context of structure-texture decomposition of images, the structure is modeled by images with sparse gradients and the texture is modeled by images having patches with low-rank structure. In this work, a heuristic based on the weigthing scheme $\frac{w_2^*}{w_1^*} = \sqrt{\frac{k}{r}}$ is proposed and permits to automatically set regularization parameters for local sparse + low rank models.  These results for sparsity in levels and beyond  show that even with a simple model and parametrized family of functions, optimization might lead to complicated expressions. We also remark that we could perform the optimization because restricting to weighted atomic norms leads to an analytical description of atoms generating the regularizers. This in turn leads to an analytical description of descent cones. The question of optimality within more general sets of atomic norms remains. Unfortunately the lack of analytical description of descent cones in the general case makes the direct extension of our proof technique difficult.

\section{Discussion and future work} \label{sec:conclusion}

We gave a general way of defining compliance measures between a  regularizer $R$ and a low dimensional model set $\Sigma$, and described some elementary properties of such measures. This opens questions on  conditions on compliance measures that guarantee the existence of an optimal convex regularizer. Also, the question of manipulating compliance measures for transformations and combinations of models  (intersections, unions, sums, ...) is a wide and challenging potential area of research. 

We considered  noiseless observations instead of the classical noisy model $y = M x_0 + e$ where $e$ is a measurement noise with finite energy $\|e\|_2$ because of the following remark. Suppose we define an optimal regularizer for a given  noise level $\|e\|_2$. There are two possible cases, either the regularizer is also optimal for $\|e\|_2=0$ or it is not. In the second case, it means that we would need to trade exact recovery capability for improved stability to noise. This is a possible route to follow especially if there is some latitude on the design of the measurement operator $M$, \ie it is possible to increase measurements to improve stability to noise. The analysis of such trade-offs is out of the scope of this article and left for future work.

We have shown that the $\ell^1$-norm is optimal among coercive continuous convex functions for sparse recovery for compliance measures based on necessary and sufficient RIP conditions. This result had to be expected due to the symmetries of the problem. The important point is that we could explicitly quantify the notion of good regularizer. We obtained the same expected results with the nuclear norm for low-rank matrix recovery.

It must be noted that  we did not use constructive proofs (we \emph{exhibited} the candidate maximum of the compliance measure) for the sparsity and low-rank cases. A full constructive proof, \ie an exact calculation and optimization of the quantities $B_\Sigma(R)$ and $D_{\Sigma}(R)$ would be intellectually more satisfying as it would not require the prior knowledge of the candidate optimum, which is our ultimate objective. We saw in the case of sparsity in levels and beyond that we can construct the regularizer that achieved optimality among a simple parametrized family of convex functions (weighted $\ell^1$-norms in levels). It is an open question to obtain more general constructions.

We used compliance measures based on (uniform) RIP recovery guarantees to give results for the  sparse recovery case, it would be interesting to do such analysis using  (non-uniform) recovery guarantees based on the statistical dimension or on the Gaussian width of the descent cones \cite{Chandrasekaran_2012,Amelunxen_2014}. One would need to precisely lower and upper bound these quantities, similarly to our approach with the RIP, to get satisfying results.

Finally, while these compliance measures are designed to make sense with respect to known results in the area of sparse recovery, one might design other compliance measures tailored for particular needs, in this search for optimal regularizers.

\section*{Acknowledgements}
This work was partly supported by the ANR, project EFFIREG ANR-20-CE40-0001, project AllegroAssai ANR-19-CHIA-0009 and project GraVa ANR-18-CE40-0005.

\bibliographystyle{imaiai}
 \bibliography{opti_convex_reg_journal}

\ifx\undefined\BySame
\newcommand{\BySame}{\leavevmode\rule[.5ex]{3em}{.5pt}\ }
\fi
\ifx\undefined\textsc
\newcommand{\textsc}[1]{{\sc #1}}
\newcommand{\emph}[1]{{\em #1\/}}
\let\tmpsmall\small
\renewcommand{\small}{\tmpsmall\sc}
\fi
\begin{thebibliography}{99}

\bibitem{Adcock_2013b}
\textsc{Adcock, B., Hansen, A., Roman, B.  {\small \&} Teschke, G.}  (2014)
  Generalized Sampling: Stable Reconstructions, Inverse Problems and Compressed
  Sensing over the Continuum. \emph{Advances in Imaging and Electron Physics},
  \textbf{182}(1), 187--279.

\bibitem{Amelunxen_2014}
\textsc{Amelunxen, D., Lotz, M., McCoy, M.~B.  {\small \&} Tropp, J.~A.}
  (2014) Living on the Edge: Phase Transitions in Convex Programs with Random
  Data. \emph{Information and Inference, A Journal of the IMA}, \textbf{3}(3),
  224--294.

\bibitem{Amelunxen_2020}
\textsc{Amelunxen, D., Lotz, M.  {\small \&} Walvin, J.}  (2020) Effective
  Condition Number Bounds for Convex Regularization. \emph{Information Theory,
  IEEE Transactions on}, \textbf{66}(4), 2501--2516.

\bibitem{Argyriou_2012}
\textsc{Argyriou, A., Foygel, R.  {\small \&} Srebro, N.}  (2012) {Sparse
  Prediction with the k-Support Norm}. \emph{Advances in Neural Information
  Processing Systems}, \textbf{25}, 1457--1465.

\bibitem{bach2013learning}
\textsc{Bach, F.}  (2013) \emph{{Learning with Submodular Functions: A Convex
  Optimization Perspective}}, Foundations and Trends in Machine Learning. {Now
  Publishers}.

\bibitem{Bastounis_2015}
\textsc{Bastounis, A.  {\small \&} Hansen, A.~C.}  (2015) On Random and
  Deterministic Compressed Sensing and the Restricted Isometry Property in
  Levels. in \emph{Sampling Theory and Applications (SampTA), 2015
  International Conference on}, pp. 297--301.

\bibitem{bertsekas1979convexification}
\textsc{Bertsekas, D.~P.}  (1979) Convexification Procedures and Decomposition
  Methods for Nonconvex Optimization Problems. \emph{Journal of Optimization
  Theory and Applications}, \textbf{29}(2), 169--197.

\bibitem{bougeard1991towards}
\textsc{Bougeard, M., Penot, J.-P.  {\small \&} Pommellet, A.}  (1991) Towards
  Minimal Assumptions for the Infimal Convolution Regularization. \emph{Journal
  of Approximation Theory}, \textbf{64}(3), 245--270.

\bibitem{Bourrier_2014}
\textsc{Bourrier, A., Davies, M., Peleg, T., Perez, P.  {\small \&} Gribonval,
  R.}  (2014) Fundamental Performance Limits for Ideal Decoders in
  High-Dimensional Linear Inverse Problems. \emph{Information Theory, IEEE
  Transactions on}, \textbf{60}(12), 7928--7946.

\bibitem{Candes2013Simpleboundsrecovering}
\textsc{Cand{\`e}s, E.  {\small \&} Recht, B.}  (2013) Simple Bounds for
  Recovering Low-Complexity Models. \emph{Mathematical Programming},
  \textbf{141}(1-2), 577--589.

\bibitem{Candes_2010}
\textsc{Candes, E.~J.  {\small \&} Plan, Y.}  (2010) Matrix Completion with
  Noise. \emph{Proceedings of the IEEE}, \textbf{98}(6), 925--936.

\bibitem{Candes_2006b}
\textsc{Cand\`{e}s, E.~J., Romberg, J.  {\small \&} Tao, T.}  (2006) {Robust
  Uncertainty Principles: Exact Signal Reconstruction from Highly Incomplete
  Frequency Information}. \emph{Information Theory, IEEE Transactions on},
  \textbf{52}(2), 489--509.

\bibitem{Chambolle_2011}
\textsc{Chambolle, A.  {\small \&} Pock, T.}  (2011) A First-order Primal-dual
  Algorithm for Convex Problems with Applications to Imaging. \emph{Journal of
  Mathematical Imaging and Vision}, \textbf{40}(1), 120--145.

\bibitem{Chandrasekaran_2012}
\textsc{Chandrasekaran, V., Recht, B., Parrilo, P.  {\small \&} Willsky, A.}
  (2012) The Convex Geometry of Linear Inverse Problems. \emph{Foundations of
  Computational Mathematics}, \textbf{12}(6), 805--849.

\bibitem{Chen1998AtomicDecompositionBasis}
\textsc{Chen, S., Donoho, D.  {\small \&} Saunders, M.}  (1998) Atomic
  {{Decomposition}} by {{Basis Pursuit}}. \emph{SIAM Journal on Scientific
  Computing}, \textbf{20}(1), 33--61.

\bibitem{Chi_2019}
\textsc{Chi, Y., Lu, Y.~M.  {\small \&} Chen, Y.}  (2019) Nonconvex
  Optimization Meets Low-rank Matrix Factorization: An Overview. \emph{Signal
  Processing, IEEE Transactions on}, \textbf{67}(20), 5239--5269.

\bibitem{Davies_2009}
\textsc{Davies, M.~E.  {\small \&} Gribonval, R.}  (2009) Restricted Isometry
  Constants where {$\ell^p$} Sparse Recovery Can Fail for {$0 < p \leq 1$ }.
  \emph{Information Theory, IEEE Transactions on}, \textbf{55}(5), 2203--2214.

\bibitem{Donoho_2006}
\textsc{Donoho, D.~L.}  (2006) {For Most Large Underdetermined Systems of
  Linear Equations the Minimal $\ell_1$-norm Solution is also the Sparsest
  Solution}. \emph{Communications on Pure and Applied Mathematics},
  \textbf{59}(6), 797--829.

\bibitem{fan2001variable}
\textsc{Fan, J.  {\small \&} Li, R.}  (2001) Variable Selection via Nonconcave
  Penalized Likelihood and its Oracle Properties. \emph{Journal of the American
  Statistical Association}, \textbf{96}(456), 1348--1360.

\bibitem{Fazel2001rankminimizationheuristic}
\textsc{Fazel, M., Hindi, H.  {\small \&} Boyd, S.~P.}  (2001) A Rank
  Minimization Heuristic with Application to Minimum Order System
  Approximation. in \emph{Proceedings of the 2001 {{American Control
  Conference}}}, vol.~6, pp. 4734--4739 vol.6.

\bibitem{foucart2009sparsest}
\textsc{Foucart, S.  {\small \&} Lai, M.-J.}  (2009) Sparsest Solutions of
  Underdetermined Linear Systems via $\ell^q$-minimization for $0<q \leq1$.
  \emph{Applied and Computational Harmonic Analysis}, \textbf{26}(3), 395--407.

\bibitem{Foucart_2013}
\textsc{Foucart, S.  {\small \&} Rauhut, H.}  (2013) \emph{A Mathematical
  Introduction to Compressive Sensing}. Springer.

\bibitem{Friedland_2018}
\textsc{Friedland, S.  {\small \&} Lim, L.-H.}  (2018) Nuclear Norm of
  Higher-order Tensors. \emph{Mathematics of Computation}, \textbf{87}(311),
  1255--1281.

\bibitem{guennec2023adaptive}
\textsc{Guennec, A., Aujol, J.-F.  {\small \&} Traonmilin, Y.}  (2023) Adaptive
  Parameter Selection for Gradient-sparse + Low Patch-rank Recovery:
  Application to Image Decomposition. \emph{HAL preprint hal-04207313}.

\bibitem{jach2008convexenvelope}
\textsc{Jach, M., Michaels, D.  {\small \&} Weismantel, R.}  (2008) The
  {{Convex Envelope}} of (n\textendash 1)-{{Convex Functions}}. \emph{SIAM
  Journal on Optimization}, \textbf{19}(3), 1451--1466.

\bibitem{lasserre2018moment}
\textsc{Lasserre, J.~B.}  (2018) The Moment-SOS Hierarchy. \emph{Proceedings of
  the International Congress of Mathematics}, pp. 3761--3784.

\bibitem{Marz_2020}
\textsc{M{\"a}rz, M., Boyer, C., Kahn, J.  {\small \&} Weiss, P.}  (2023)
  Sampling Rates for {$\ell^1$}-Synthesis. \emph{Foundations of Computational
  Mathematics}, \textbf{23}, 2089--2150.

\bibitem{Negahban2012UnifiedFrameworkHighDimensional}
\textsc{Negahban, S.~N., Ravikumar, P., Wainwright, M.~J.  {\small \&} Yu, B.}
  (2012) A {{Unified Framework}} for {{High}}-{{Dimensional Analysis}} of
  ${M}$-{{Estimators}} with {{Decomposable Regularizers}}. \emph{Statistical
  Science}, \textbf{27}(4), 538--557.

\bibitem{Pock_2010}
\textsc{Pock, T., Cremers, D., Bischof, H.  {\small \&} Chambolle, A.}  (2010)
  Global Solutions of Variational Models with Convex Regularization. \emph{SIAM
  Journal on Imaging Sciences}, \textbf{3}(4), 1122--1145.

\bibitem{Puy_2015}
\textsc{Puy, G., Davies, M.~E.  {\small \&} Gribonval, R.}  (2017) Recipes for
  Stable Linear Embeddings from Hilbert Spaces to $\mathbb{R}^m$.
  \emph{Information Theory, IEEE Transactions on}, \textbf{63}(4), 2171--2187.

\bibitem{Recht_2010}
\textsc{Recht, B., Fazel, M.  {\small \&} Parrilo, P.}  (2010) {Guaranteed
  Minimum-Rank Solutions of Linear Matrix Equations via Nuclear Norm
  Minimization}. \emph{SIAM Review}, \textbf{52}(3), 471--501.

\bibitem{Rockafellar_1970}
\textsc{Rockafellar, R.~T.}  (1970) \emph{Convex Analysis}, no.~28. Princeton
  university press.

\bibitem{Soubies_2015}
\textsc{Soubies, E., Blanc-F{\'e}raud, L.  {\small \&} Aubert, G.}  (2015) A
  Continuous Exact $\ell_0$ Penalty (CEL0) for Least Squares Regularized
  Problem. \emph{SIAM Journal on Imaging Sciences}, \textbf{8}(3), 1607--1639.

\bibitem{Studer_2013}
\textsc{Studer, C.  {\small \&} Baraniuk, R.~G.}  (2014) {Stable Restoration
  and Separation of Approximately Sparse Signals}. \emph{Applied and
  Computational Harmonic Analysis}, \textbf{37}(1), 12--35.

\bibitem{tanner2020compressed}
\textsc{Tanner, J.  {\small \&} Vary, S.}  (2023) Compressed Sensing of
  Low-rank plus Sparse Matrices. \emph{Applied and Computational Harmonic
  Analysis}, \textbf{64}, 254--293.

\bibitem{Traonmilin_2020inverse}
\textsc{Traonmilin, Y.  {\small \&} Aujol, J.-F.}  (2020) {The Basins of
  Attraction of the Global Minimizers of the Non-convex Sparse Spike Estimation
  Problem}. \emph{{Inverse Problems}}, \textbf{36}(4), 045003.

\bibitem{Traonmilin_2020b}
\textsc{Traonmilin, Y., Aujol, J.-F.  {\small \&} Leclaire, A.}  (2023) The
  Basins of Attraction of the Global Minimizers of Non-convex Inverse Problems
  with Low-dimensional Models in Infinite Dimension. \emph{Information and
  Inference: A Journal of the IMA}, \textbf{12}(1), 113--156.

\bibitem{Traonmilin_2016}
\textsc{Traonmilin, Y.  {\small \&} Gribonval, R.}  (2018) Stable Recovery of
  Low-dimensional Cones in Hilbert Spaces: One RIP to Rule them All.
  \emph{Applied and Computational Harmonic Analysis}, \textbf{45}(1), 170--205.

\bibitem{Traonmilin_2015}
\textsc{Traonmilin, Y., Ladjal, S.  {\small \&} Almansa, A.}  (2015) Robust
  Multi-image Processing with Optimal Sparse Regularization. \emph{Journal of
  Mathematical Imaging and Vision}, \textbf{51}(3), 413--429.

\bibitem{Traonmilin_2018}
\textsc{Traonmilin, Y.  {\small \&} Vaiter, S.}  (2018) Optimality of 1-norm
  Regularization Among Weighted 1-norms for Sparse Recovery: A Case Study on
  How to Find Optimal Regularizations. \emph{Journal of Physics: Conference
  Series}, \textbf{1131}, 012009.

\bibitem{Vaiter2015Modelselectionlow}
\textsc{Vaiter, S., Golbabaee, M., Fadili, J.  {\small \&} Peyré, G.}  (2015)
  Model Selection with Low Complexity Priors. \emph{Information and Inference:
  A Journal of the IMA}, \textbf{4}(3), 230--287.

\bibitem{Vaiter2017ModelConsistencyPartly}
\textsc{Vaiter, S., Peyre, G.  {\small \&} Fadili, J.}  (2017) Model
  {{Consistency}} of {{Partly Smooth Regularizers}}. \emph{Information Theory,
  IEEE Transactions on}, \textbf{PP}(99), 1--1.

\bibitem{vershynin2015estimation}
\textsc{Vershynin, R.}  (2015) Estimation in High Dimensions: A Geometric
  Perspective. in \emph{Sampling theory, A Renaissance}, pp. 3--66. Springer.

\bibitem{Yuan2006Modelselectionestimation}
\textsc{Yuan, M.  {\small \&} Lin, Y.}  (2006) Model Selection and Estimation
  in Regression with Grouped Variables. \emph{Journal of the Royal Statistical
  Society: Series B (Statistical Methodology)}, \textbf{68}(1), 49--67.

\end{thebibliography}

\appendix

\section{Appendices}

This section describes the tools and proofs used to obtain our results.
\subsection{Summary of properties used in proofs}\label{sec:summary_prev_results}

 From  \cite[Table 1]{Traonmilin_2016} (which summarizes results from \cite{Rockafellar_1970}  ),  the function $x \in \sE(\sA) \mapsto \|x\|_{\sA}$ is always non-negative, lower semi-continuous and subadditive  (\ie it satisfies the triangle inequality). It is furthermore positively homogeneous as soon as $0 \in \cl{\tconv}(\sA)$, continuous as soon as $0$ is in the interior of $\cl{\tconv}(\sA)$, and coercive as soon as $\cl{\tconv}(\sA)$ is bounded. Finally, it is indeed a norm if $\cl{\tconv}(\sA) = -\cl{\tconv}(\sA)$.

  We refer the reader to  \cite{Traonmilin_2016}[Section 2.2] and \cite{Argyriou_2012} for properties of the atomic norm $\|\cdot\|_\Sigma$ (cf \Cref{def:modelnorm}).
  We will use the following two properties of $\|\cdot\|_\Sigma$ (defined in \Cref{sec:atomic}). 
\begin{fact}[{From e.g. \cite{Traonmilin_2016}}] \label{fact:atom1}

 For all $x \in \Sigma$, $\|x\|_\Sigma = \|x\|_\sH$.
\end{fact}

\begin{fact}[{From \cite{Traonmilin_2016}[Fact 2.1] applied to $\|\cdot\|_\Sigma$}] \label{fact:atom2}
 For all $z \in \sH$
 \begin{equation}
  \|z\|_\Sigma = \inf \left\{ \sqrt{\sum \lambda_i \|u_i\|_\sH^2} :  \lambda_i \in \bRp, \sum \lambda_i =1, 
 u_i \in \Sigma, z=\sum  \lambda_i u_i\right\}.
 \end{equation}

\end{fact}

\subsection{Proofs for  \texorpdfstring{\Cref{sec:comp_general}}{Section \ref{sec:comp_general}}}\label{sec:proof_ideal_comp}

\subsubsection{Proof of  \texorpdfstring{\Cref{lem:ideal_decoder_descent}}{Lemma \ref{lem:ideal_decoder_descent}}}

Consider $x\in \Sigma$, and $z \in \sH$. We have   $\iota_\Sigma(x+z)\leq \iota_\Sigma(x) =0$ if and only if $x+z \in \Sigma$, \ie if there is $x' \in \Sigma$ such that $z = x'-x$. Hence, $\sT_{\iota_\Sigma}(x) = \{\gamma (x'-x): \gamma \in \bR, x' \in \Sigma\}$. It follows that $\sT_{\iota_\Sigma}(\Sigma) = \{\gamma z: \gamma \in \bR, z \in \Sigma-\Sigma\} \supseteq \Sigma-\Sigma$. When $\Sigma$ is positively homogeneous, for any $z = x'-x \in \Sigma-\Sigma$ and $\gamma \in \bR$ we have: if $\gamma >0$ then $\gamma z = \gamma x'-\gamma x \in \Sigma-\Sigma$; if $\gamma <0$ then $\gamma z = (-\gamma)x-(-\gamma)x' \in \Sigma-\Sigma$; if $\gamma = 0$ then $\gamma z = 0 = x-x \in \Sigma-\Sigma$, hence indeed $\sT_{\iota_\Sigma}(\Sigma)  \subseteq \Sigma-\Sigma$.

Now consider $y \in \sT_{\iota_{\Sigma}}(\Sigma)$ and write it as $y = \gamma(x_{1}-x_{2})$ where $x_{1},x_{2} \in \Sigma$ and $\gamma \in \bR$. Since $\Sigma \subseteq \mathtt{dom}(R)$ we have $\max(R(x_{1}),R(x_{2})) < \infty$. We will prove that $y \in \sT_{R}(\Sigma)$. We distinguish two cases: if $R(x_1)\leq R(x_2)$ then $R(x_{2}+(x_{1}-x_{2})) = R(x_{1})  \leq R(x_{2})$ hence $y = \gamma (x_{1}-x_{2}) \in \sT_{R}(x_{2})$, and as $x_{2} \in \Sigma$ it follows that $y  \in \sT_{R}(\Sigma)$; otherwise $R(x_{2})<R(x_{1})$ hence $R(x_{1}+(x_{2}-x_{1})) = R(x_{2}) < R(x_{1})$ hence $y = (-\gamma) (x_{2}-x_{1}) \in \sT_{R}(x_{1})$ and therefore $y \in \sT_{R}(\Sigma)$. $\qed$

\subsubsection{Proof of \texorpdfstring{\Cref{lem:atomic_necessary}}{Lemma \ref{lem:atomic_necessary}} }\label{sec:reduction_proof}

Given  $t > R(0)$, the level set $\sL(R,t)= \{y\in \sH : R(y)\leq t \}$ is nonempty, convex and closed (by convexity and lower semi-continuity of $R$), and bounded (by coercivity of $R$). 
We define $\sA \defin \sL(R,t) \cap \Sigma = \{x\in \Sigma : R(x) \leq t \}$. 

Consider $z \in \sT_{\|\cdot\|_\sA}(\Sigma)$. If $z = 0$ then clearly $z \in \sT_R(\Sigma)$. Let us prove that the same holds when $z \neq 0$.
By definition, there exists $\gamma \in \bR \setminus \{0\}$ and $x\in \Sigma$ such that 
\[
\|x+z/\gamma\|_\sA \leq \|x\|_\sA.
\]
On the one hand we have $R(0 \cdot x) = R(0) < t$. On the other hand, since $R$ is coercive, we have $R(\lambda x)\underset{\lambda \to +\infty}{\to} +\infty$. Since $R$ is continuous, by the mean value theorem, there is $\lambda_0 >0$ such that
\[
 R(\lambda_0 x) = t.
 \]

Since $\Sigma$ is a cone, the vector $x' = \lambda_0 x$ belongs to $\Sigma$ and, since  $R(x')=t$, by definition of $\sA$ we have indeed $x' \in \sA$, hence $\|x'\|_{\sA} \leq 1$. Furthermore, since  $\|\cdot\|_\sA$ is positively homogeneous (because $0 \in \overline{\tconv}(\sA)$), we have  
\[
\|x'+ \lambda_0 z/\gamma\|_\sA = \lambda_{0} \|x+z/\gamma\|_{\sA} \leq \lambda_{0} \|x\|_{\sA} = \|x'\|_\sA. 
\]
We now observe that, on the one hand, the level set $\sL( \|\cdot\|_\sA,1) = \overline{\tconv}(\sA)$  is the smallest closed convex set containing $\sA$; on the other hand $\sA \subset \sL( R,t)$ and $\sL( R,t)$ is convex and closed. Thus $\sL( \|\cdot\|_\sA,1) \subset \sL(R,t)$ and the fact that $\|x'+\lambda_{0}z/\gamma\|_{\sA} \leq \|x'\|_{\sA} \leq 1$ therefore implies
\begin{equation}
R(x'+\lambda_{0} z/\gamma) \leq t = R(x').
\end{equation}
This shows that $z \in \sT_R(\Sigma)$ and establishes that $\sT_{\|\cdot\|_{\sA}}(\Sigma) \subseteq \sT_R(\Sigma)$.

 Let us now prove that $\|\cdot\|_{\sA}$ is continuous, convex, coercive and positively homogeneous. First, from the property of gauges (see \Cref{sec:summary_prev_results}), $\|\cdot\|_\sA$ is always convex and lower semi-continuous. 
Second, since $R$ is coercive, its level sets are bounded, hence  $\cl{ \tconv}(\sA)$ is bounded and $\|\cdot\|_\sA$ is coercive.  Finally, as $R(0)<t$ and $R$ is continuous, $0$ is in the interior of $\sL(R,t)$.  There exists $\epsilon >0$ such that an open ball $O$ of radius $\epsilon$ centered on $0$ is included in $\sL(R,t)$.  This implies $O \cap \Sigma \subset  \sL(R,t) \cap \Sigma = \sA$ which in turns imply $\tconv( O \cap \Sigma) \subset  \tconv(\sA)  \subset \overline{\tconv(\sA) }$.  Remark that $ \sE(O \cap \Sigma) = \sE(\Sigma) = \sH$. Now we need to find $ O'$ an open ball of radius $\epsilon'$  such that $ O'\subset \tconv( O \cap \Sigma)$. In each orthant $\Omega_r$, we can find a normalized  basis $E = (e_i) \in \Sigma$ such that $\Omega_r \subset \sE(E)$. We define the norm $\|\sum_i \mu_i e_i\|_E=  \sum \mu_i$. This norm is equivalent to $\| \cdot\|_\sH$. This implies  there is a constant $c_r$ depending on the orthant $\Omega_r$, such that  for $x = \sum_i \mu_i e_i \in O' \cap \Omega_r$, $\max_i \mu_i < c_r\epsilon'$.  This implies 
\begin{equation}
x = t  \sum_i \frac{\mu_i}{\sum_j \mu_j} \epsilon e_i
 \end{equation}
with $t =\frac{\sum_j \mu_j}{\epsilon} \leq n c_r \frac{\epsilon'}{\epsilon}$. Taking $\epsilon' < {\epsilon}/(n c_r)$ implies $t < 1$ and  $x \in \tconv( O \cap \Sigma)$. As there is a finite number of orthants we can chose $\epsilon'$ such that we always have $x \in O'$ implies $x \in \tconv( O \cap \Sigma)$. $\qed$

\subsection{Proofs for \texorpdfstring{\Cref{sec:RC}}{Section \ref{sec:RC}}}\label{sec:link_RIP_RC}

\begin{proof}[Proof of \Cref{lem:link_RIP_RC}]

Denote $\alpha =\inf_{x\in (\Sigma-\Sigma) \cap S(1)}   \|M x\|_2^2$ and $\beta = \sup_{x\in (\Sigma-\Sigma) \cap S(1)}   \|M x\|_2^2$, so that $\gamma(M) = \beta/\alpha$. Since $\Sigma$ is a cone, we have for every
 $x\in \Sigma- \Sigma$,
\begin{equation}\label{eq:link_RIP_RC_star}
\alpha\|x\|_\sH^2 \leq \|\mM x\|_2^2 \leq \beta \|x\|_{\sH}^{2} = \gamma(M)\alpha\|x\|_\sH^2,
\end{equation}
Multiplying $x$ in~\eqref{eq:link_RIP_RC_star} by any $\lambda>0$, we have
\begin{equation*}    
\lambda^2\alpha\|x\|_\sH^2 \leq \|\lambda \mM x\|_2^2 \leq \lambda^2\gamma(M)\alpha\|x\|_\sH^2.
\end{equation*}
We look for $\lambda > 0$, $\delta \neq 1$ such that $\lambda M$ satisfies a symmetric RIP with constant $\delta$, \ie
\begin{equation*}
  \lambda^2\alpha = 1-\delta
  \quad \text{ and } \quad
  \lambda^2 \gamma(M)\alpha = 1+\delta.
\end{equation*}
Adding these two equalities yields $\lambda^2 \alpha(1 +  \gamma(M) )= 1$, hence $\lambda^2 = \frac{1}{\alpha(1 +  \gamma(M))}$. Dividing them yields
\begin{equation*}
  \frac{1-\delta}{1+\delta} = \gamma(M)
  \iff
  \delta =\frac{\gamma(M)-1}{\gamma(M)+1}.
\end{equation*}
We have shown that for any $M$, there exists $\lambda >0$ such that
 \begin{equation*}
  \delta(\lambda M)  \leq  \frac{\gamma(M)-1}{\gamma(M)+1}.
 \end{equation*}
Remark that the value of $\lambda$ that makes the RIP bounds symmetrical is unique, and that no better symmetrical RIP bound can be obtained, otherwise we could construct a better restricted conditioning (which is impossible by definition of $\gamma(M)$). We deduce
\begin{equation*}
  \delta(\lambda M)  =  \frac{\gamma(M)-1}{\gamma(M)+1}.  
 \end{equation*}

\end{proof}
\begin{lemma}\label{lem:best_op_nec_RIP_dim1}
Consider a cone $\Sigma \subseteq \sH$ and $\mathcal{T} \subseteq \sH$ a non-empty set, and denote $\mathcal{P}$ the set of symmetric positive semi-definite linear operators on $\sH$, \ie $N \in \mathcal{P}$ if and only if $N^{H} = N$ and $N \succeq 0$. Then
\begin{equation}\label{eq:best_op_nec_RIP_dim1}
  \inf_{M: \ker M \cap \mathcal{T} \neq \{0\}} \gamma_{\Sigma}(M) =
  \inf_{N \in \mathcal{P}:   \dim \ker N = 1,
    \ker N \cap \mathcal{T} \neq \{0\}} \gamma_{\Sigma}(N).
\end{equation}
\end{lemma}
\begin{proof}
The infimum on the r.h.s. of~\eqref{eq:best_op_nec_RIP_dim1} is over a more constrained set than on the l.h.s., hence
\[
\inf_{M: \ker M \cap \mathcal{T} \neq \{0\}} \gamma_{\Sigma}(M) \leq \inf_{N \in \mathcal{P}: \dim \ker N = 1, \ker N \cap \mathcal{T} \neq \{0\}} \gamma_{\Sigma}(N).
\]
If the l.h.s. is infinite, then the right-hand side must also be infinite, and we are done.

Assume that the l.h.s. is finite.
We now prove the reverse inequality.
For this, consider $M$ a linear operator with $\ker M \cap \mathcal{T} \neq \{0\}$ and $\gamma_{\Sigma}(M) < \infty$.
There exists a nonzero vector $t \in \ker M \cap \mathcal{T}$.
We build an operator $N \in \mathcal{P}$ such that $\ker N = \tspan(t)$ and with $\gamma_{\Sigma}(N)$ arbitrarily close to $\gamma_{\Sigma}(M)$. 

Since $\gamma_{\Sigma}(M) < \infty$, $M$ is nonzero hence
a singular value decomposition allows writing $M = \sum_{i=1}^{r} \sigma_{i} u_{i}v_{i}^{H}$ where $(u_{i})_{i=1}^{r}$ and  $(v_{i})_{i=1}^{r}$ are orthonormal families and $\min_{1 \leq i \leq r} \sigma_{i}>0$.
First we deal with the case where $\dim \ker M = 1$. We set $N = \sum_{i=1}^{r} \sigma_{i} v_{i} v_{i}^{H}$ so that $N \in \mathcal{P}$ and $\dim \ker N = 1$ too. Since $\|N x\|_2^{2}  = \sum_{i=1}^{r} \sigma_{i}^{2} \langle v_{i},x\rangle^{2} = \|Mx\|_2^{2}$  for any vector $x$ we have $\gamma(N) = \gamma(M)$, and we are done. Assume now that $k := \dim \ker M \geq 2$. Observe that $\tspan(t) \subset \ker M$ and let $(e_{1},\ldots,e_{k-1})$ be an orthonormal basis of the orthogonal complement of $\tspan(t)$ in $\ker M$, so that $(v_{1},\ldots,v_{r},e_{1},\ldots,e_{k-1})$ is an orthonormal family. For each $\epsilon>0$, define $N_{\epsilon} = \sum_{i=1}^{r} \sigma_{i} v_{i}v_{i}^{H} + \epsilon \sum_{j=1}^{k-1}e_{j} e_{j}^{H}$. Again, $N_{\epsilon} \in \mathcal{P}$ and $\tspan(t) = \ker N_{\epsilon}$ so that $\dim \ker N_{\epsilon} = 1$, and for each $x \in \mathcal{H}$ we have
\begin{equation*}
\|N_{\epsilon}x\|_2^{2} = \sum_{i=1}^{r} \sigma_{i}^{2} \langle v_{i},x\rangle^{2} + \epsilon^{2} \sum_{j=1}^{k-1} \langle e_{j},x\rangle^{2} = \|Mx\|_2^{2}+\epsilon^{2} \sum_{j=1}^{k-1} \langle e_{j},x\rangle^{2},
\end{equation*}
hence $\|Mx\|_2^{2} \leq  \|N_{\epsilon}x\|_2^{2} \leq \|Mx\|_2^{2}+\epsilon^{2}\|x\|_2^{2}$. Since $\gamma_{\Sigma}(M)<\infty$, we get
\[
0< \inf_{x \in (\Sigma-\Sigma) \cap S(1)} \|Mx\|_2^{2}
\leq 
\inf_{x \in (\Sigma-\Sigma) \cap S(1)} \|N_{\epsilon}x\|_2^{2}
\leq 
\sup_{x \in (\Sigma-\Sigma) \cap S(1)} \|N_{\epsilon}x\|_2^{2}
\leq 
\sup_{x \in (\Sigma-\Sigma) \cap S(1)} \|Mx\|_2^{2}+\epsilon^{2}
\]
which implies
\[
\gamma_{\Sigma}(N_{\epsilon})
 \leq 
 \frac{\sup_{x \in (\Sigma-\Sigma) \cap S(1)} \|Mx\|_2^{2}+\epsilon^{2}}{\inf_{x \in (\Sigma-\Sigma) \cap S(1)} \|Mx\|_2^{2}} 
 = \gamma_{\Sigma}(M) + \frac{\epsilon^{2} }{\inf_{x \in (\Sigma-\Sigma) \cap S(1)} \|Mx\|_2^{2}}.
\]
This implies that $\inf_{\epsilon>0} \gamma_{\Sigma}(N_{\epsilon}) \leq \gamma_{\Sigma}(M)$ as claimed.

\end{proof}

\begin{proof}[Proof of \Cref{lem:transf_sigma_RC}]
We define  
\begin{equation}
G(\Sigma, E, M) :=  \frac{\sup_{y\in (\Sigma-\Sigma)\cap E} \|My\|_2^2}{\inf_{y\in (\Sigma-\Sigma)\cap E} \|My\|_2^2}.
\end{equation}
For any nonzero $M$, we have 
\begin{equation}
\begin{split}
 \gamma_{F\Sigma} (M) &=  \frac{\sup_{x\in (F\Sigma-F\Sigma)\cap S(1)} \|Mx\|_2^2}{\inf_{x\in (F\Sigma-F\Sigma)\cap S(1)} \|Mx\|_2^2}
=  \frac{\sup_{y\in (\Sigma-\Sigma)\cap F^{-1}S(1)} \|MFy\|_2^2}{\inf_{y\in (\Sigma-\Sigma)\cap F^{-1}S(1)} \|MFy\|_2^2}.\\
 \end{split}
\end{equation}
Hence, 
\begin{align*}
A_{F\Sigma}^{RC}(R \circ F^{-1}) 
&=  \inf_{M : \ker M \cap \sT_{R \circ F^{-1}}(F\Sigma) \neq \{ 0\}} \gamma_{F\Sigma} (M) \\
 &=  \inf_{M : \ker M \cap \sT_{R \circ F^{-1}}(F\Sigma) \neq \{ 0\}}  G(\Sigma, F^{-1}S(1), MF).
\end{align*}
By \Cref{lem:transf_sigma_gen} with $R' = R \circ F^{-1}$, $\sT_{R \circ F^{-1}}(F\Sigma) = \sT_{R'}(F\Sigma) = F(\sT_{R' \circ F}(\Sigma)) = F(\sT_{R}(\Sigma))$. Also, $ \ker M \cap \sT_{R \circ F^{-1}}(F\Sigma) \neq \{ 0\} $ is equivalent to the existence of  $  z \in \ker M$ such that
$ z' :=   F^{-1}z \in \sT_{R}(\Sigma)$, \ie of $ z' \in \sT_{R}(\Sigma) $ such that $ z :=  Fz' \in \ker M$. As a result,

\begin{equation}
\begin{split}
 \inf_{M : \ker M \cap \sT_{R \circ F^{-1}}(F\Sigma) \neq \{ 0\}} \gamma_{F\Sigma} (M) 
 &=  \inf_{M : F^{-1}\ker M \cap \sT_{R }(\Sigma) \neq \{ 0\}}  G(\Sigma, F^{-1}S(1), MF) .\\
 \end{split}
\end{equation}
Rewriting  $M' = MF$, we have  $\ker M' = F^{-1} \ker M $ and 
\begin{equation}
\begin{split}
 \inf_{M : \ker M \cap \sT_{R \circ F^{-1}}(F\Sigma) \neq \{ 0\}} \gamma_{F\Sigma} (M) 
 &=  \inf_{M' : \ker M' \cap \sT_{R }(\Sigma) \neq \{ 0\}}  G(\Sigma, F^{-1}S(1), M')\\
 \end{split}
\end{equation}
which gives the desired result using the fact that $F^{-1}S(1) = S(1)$ since $F$ is a linear isometry. 

\end{proof}
\subsection{Proofs for \texorpdfstring{\Cref{sec:nec_RIP}}{Section \ref{sec:nec_RIP}} }\label{sec:proofuos}

\begin{proof}[Proof of \Cref{lem:expr_RC_sparse}] 

Consider $ z \in \sH \setminus \{0\}$  and $M = I-\Pi_{z}$. For every $x \in S(1)$, we have
\begin{equation}
 \|Mx\|_2^2 = 1 -\frac{\ls x,z\rs^2}{\|z\|_\sH^2}
\end{equation}
hence 
\begin{align*}
\gamma_{\Sigma}(M) &= \frac{\sup_{x\in (\Sigma-\Sigma)\cap S(1)} \|Mx\|_2^2}{\inf_{x\in (\Sigma-\Sigma)\cap S(1)} \|Mx\|_2^2} 
= \frac{ 1  -\inf_{x\in (\Sigma-\Sigma)\cap S(1)} \frac{\ls x,z\rs^2}{\|z\|_\sH^2} }{ 1-\sup_{x\in (\Sigma-\Sigma)\cap S(1)}\frac{\ls x,z\rs^2}{\|z\|_\sH^2}} 
\end{align*}
\textbf{Case 1:} By assumption there is $x_{0}$ such that $\|x_{0}\|_{\sH}=1$ and $\Sigma \subseteq \tspan(x_0)$. Since $\Sigma \neq \{0\}$ is a cone, it follows that  $(\Sigma-\Sigma) \cap S(1) =
\tspan(x_{0}) \cap S(1) = \{-x_{0},+x_{0}\}$ and 
\begin{equation}
\begin{split}
\inf_{x\in (\Sigma-\Sigma)\cap S(1)} \frac{\ls x,z\rs^2}{\|z\|_\sH^2}&= \sup_{x\in (\Sigma-\Sigma)\cap S(1)}\frac{\ls x,z\rs^2}{\|z\|_\sH^2}  =  \frac{\ls x_0,z\rs^2}{\|z\|_\sH^2}.
\end{split}
\end{equation}
Hence, if $z \in \Sigma = \tspan(x_{0})$ we have $\gamma_{\Sigma}(M) = +\infty$, otherwise 
$\frac{\ls x_0,z\rs^2}{\|z\|_\sH^2} <1$ and $\gamma_{\Sigma}(M)=1$. Thus, if $\sT_{R}(\Sigma) \subseteq \Sigma$ we have   $A_\Sigma^{RIP,\mathtt{nec}}(R) = +\infty$, otherwise there is $z \in \sT_R(\Sigma) \setminus \Sigma$, and $A_\Sigma^{RIP,\mathtt{nec}}(R) = 1$.\\
\textbf{Case 2:} Let us show that for any $z \neq 0$ there is some $x \in (\Sigma - \Sigma) \setminus \{0\}$ such that  $\ls x,z\rs =0$. This implies $\inf_{x\in (\Sigma-\Sigma)\cap S(1)} \frac{\ls x,z\rs^2}{\|z\|_\sH^2}=0$ and yields the result. Indeed, by assumption, given any $x_1 \in \Sigma \setminus \{ 0\}$ there is $x_2 \in \Sigma$ such that $x_2 \notin \tspan(x_1)$ (hence $x_{2} \neq 0$). If $\ls x_1, z \rs = 0$ we take $x = x_{1} = x_{1}-\lambda x_{2}$ with $\lambda = 0$. Otherwise, with $\lambda = \frac{\ls x_2, z \rs }{\ls x_1, z \rs }$ we set $x  = \lambda x_1 -x_2$. In both cases we have $x \neq 0$ and, since $\Sigma$ is a cone, $x \in \Sigma -\Sigma $ and $\ls \lambda x_1 -x_2, z\rs =0$.

\end{proof}

\begin{proof}[Proof of \Cref{lem:ProjectionExists}]
Since $E \cap S(1)$ is compact, for any $z$ there exists $\tilde{x} \in E \cap S(1)$ such that 
\begin{equation}\label{eq:TmpProjOptim}
|\langle \tilde{x},z\rangle|^{2} = \max_{\tilde{y} \in E \cap S(1)} |\langle \tilde{y},z\rangle|^{2}.
\end{equation}
Since $E$ is a union of subspaces, it is homogeneous. Thus, as $\tilde{x} \in E$, we have $x:= \ls \tilde{x},z\rs \tilde{x} \in E$. If $y \in E \setminus \{0\}$, we have $\tilde{y} := y/\|y\|_{\sH} \in E \cap S(1)$, $\langle z,\tilde{y}\rangle \tilde{y}$ is the orthogonal projection of $z$ on $\tilde{y}$ and
\begin{equation}
\begin{split}
\|z-y\|_{\sH}^{2} &= \big\|z-\|y\|_{\sH} \cdot \tilde{y}\big\|_{\sH}^{2} \geq \|z-\langle z,\tilde{y}\rangle \tilde{y}\|_{\sH}^{2} = \|z\|_{\sH}^{2}-|\langle z,\tilde{y}\rangle|^{2} \\
&\stackrel{\eqref{eq:TmpProjOptim}}{\geq} \|z\|_{\sH}^{2}-|\langle z, \tilde{x}\rangle|^{2} \\
\end{split}
\end{equation}
 Since $\|z- x\|_{\sH}^{2} = \|z\|_{\sH}^{2}-2 \re \ls z, x \rs + \|x\|_\sH^2 = \|z\|_{\sH}^{2}-|\langle z, \tilde{x}\rangle|^{2}$, we conclude 
 
\begin{equation}
\begin{split}
\|z-y\|_{\sH}^{2}  &\geq \|z-x\|_{\sH}^{2}\\
\end{split}
\end{equation}

\noindent and $x \in P_{E}(z)$ by definition of $P_E$.

If $x' \in P_{E}(z)$, we have $\|z-x'\|_{\sH}^{2} = \|z-x\|_{\sH}^{2} = \min_{y \in E} \|z-y\|_{\sH}^{2}$ hence the notation $\|z-P_{E}(z)\|_{\sH}^{2}$ is unambiguous. Since $x' \in P_{E}(z)$, there is equality in the above equation with $y=x'$, hence $\|y\|_{\sH}=\langle z,\tilde{y}\rangle$ and $|\langle z,\tilde{y}\rangle|^{2}= |\langle z,\tilde{x}\rangle|^{2}$, therefore $\langle z,y\rangle = \langle z,\|y\|_{\sH}\tilde{y}\rangle = \|y\|_{\sH} \langle z,\tilde{y}\rangle = \|y\|_{\sH}^{2} = \langle z,\tilde{y}\rangle^{2} = \langle z,\tilde{x}\rangle^{2} = \|x\|_{\sH}^{2}$. This shows that the notations $\|P_{E}(z)\|_{\sH}^{2}$ and $\langle z,P_{E}(z)\rangle$ are unambiguous and that $\|P_{E}(z)\|_{\sH}^{2} = \langle z,P_{E}(z)\rangle$.

 We also have $\|z\|_{\sH}^{2} = \|x\|^{2}_{\sH}+\|z-x\|_{\sH}^{2} = \|x'\|^{2}$, and $\langle z,y\rangle = \|y\|_{\sH}$ hence the notations $\|z-P_{E}(z)\|_{\sH}^{2}$ and $\|P_{E}(z)\|_{\sH}^{2}$ are unambiguous.
\end{proof}

\begin{proof}[Proof of \Cref{cor:RCnecUoS}]
Since $\Sigma-\Sigma$ is a union of subspaces and $(\Sigma-\Sigma) \cap S(1)$ is compact, by \Cref{lem:ProjectionExists}, $\sup_{x\in (\Sigma-\Sigma)\cap S(1)}\frac{\ls x,z\rs^2}{\|z\|_\sH^2} =  \frac{\|P_{\Sigma-\Sigma}(z)\|_\sH^2}{\|z\|_\sH^2}$, hence we have
\begin{align*}
\left(B_{\Sigma}(R)+1\right)^{-1} & =  \left(\sup_{ z \in \sT_R(\Sigma) \setminus \{0\}} \frac{\|z- P_{\Sigma-\Sigma}(z)\|_\sH^2}{\|P_{\Sigma-\Sigma}(z)\|_\sH^2} +1\right)^{-1}
=  \inf_{ z \in \sT_R(\Sigma) \setminus \{0\}} \frac{\|P_{\Sigma-\Sigma}(z)\|_\sH^2}{\|z- P_{\Sigma-\Sigma}(z)\|_\sH^2 + \|P_{\Sigma-\Sigma}(z)\|_\sH^2 }\\
&=  \inf_{ z \in \sT_R(\Sigma) \setminus \{0\}} \frac{\|P_{\Sigma-\Sigma}(z)\|_\sH^2}{\|z\|_\sH^2 }.
\end{align*}

Since $\Sigma$ is a cone and $\Sigma \neq \tspan(x)$ for each $x \in \Sigma$, by \Cref{lem:expr_RC_sparse}, using~\eqref{eq:TranslateRIPtoRC} we have
\(
\gamma_{\Sigma}^{\mathtt{nec}}(R) = \frac{1}{1-(1+B_{\Sigma}(R))^{-1}} = 1+1/B_{\Sigma}(R)
\)
hence
\(
\delta_{\Sigma}^{\mathtt{nec}}(R) = \frac{\gamma_{\Sigma}^{\mathtt{nec}}(R)-1}{\gamma_{\Sigma}^{\mathtt{nec}}(R)+1}  = (2B_{\Sigma}(R)+1)^{-1}.
\)

We conclude using that  $b \mapsto  1/(1+2b)$ is decreasing.

\end{proof}

\subsubsection{Lemmas for the proof of \texorpdfstring{\Cref{th:RIP_nec_atom}}{Theorem \ref{th:RIP_nec_atom}}  (sparse recovery)}\label{sec:proofsparsity}

We begin by some technical lemmas. We recall that $T_{2} = T_{2}(z) \subseteq \{ 1, \ldots,n\}$ denotes a set indexing any $2k$ largest components (in magnitude) 
of vector $z$ , while $T=T(z)  \subseteq  \{ 1, \ldots,n\}$ will denote a set indexing $k$ largest components (in magnitude). Given an index set $\emptyset \neq H \subseteq  \{ 1, \ldots,n\}$, $Q_{H}$ is the ``cube'' of all vectors $v \in \mathbb{R}^{n}$ such that $\supp(v) = H$ and $|v_{i}|=1$ for every $i \in H$. The restriction of $v$ to $H$, $v_{H} \in \bR^{n}$, is such that 
$(v_{H})_{i}= v_{i}$, $i \in H$ and $\supp(v_{H}) \subseteq H$.

\begin{lemma} \label{lem:opt_support_wl1}
 Let $\Sigma = \Sigma_k$. Let $\|\cdot\|_w$ be a weighted $\ell^1$-norm ( for $w= (w_i)_{i=1}^n$ with $w_i>0$, $\|x\|_w = \sum w_i \|x\|_1$). Let $z \in \sT_{\|\cdot\|_w}(\Sigma)$. There is a support $H$ of size $\leq k$ such that  
 \begin{equation}
\|z_{H^{c}}\|_{w}-\|z_{H}\|_{w} = \inf_{x \in \Sigma} \left\{\|x+z\|_{w}-\|x\|_{w}\right\} \leq 0,
 \end{equation}
\ie the infimum is achieved at $x^* = -z_{H}$. 

Moreover, if $\|\cdot\|_w = \|\cdot\|_1$, $H = T(z)$.
\end{lemma} 
\begin{proof}
The result is trivial for $z = 0$, so we prove it for $z \in \sT_{\|\cdot\|_{w}}(\Sigma) \setminus \{0\}$.
  Consider $H \in \arg\min_{T: |T| \leq k} \left\{ \|z_{T^{c}}\|_{w}-\|z_{T}\|_{w}\right\}$.
  By definition of $\sT_{\|\cdot\|_w}(\Sigma)$, since $z \in \sT_{\|\cdot\|_w}(\Sigma) \setminus \{0\}$, there are $x' \in \Sigma$, $\lambda \in \bR \setminus\{0\}$  such that $\|x'+\lambda z\|_w \leq \|x'\|_w$. 
By homogeneity of $\Sigma$, $x := x'/\lambda \in \Sigma$ and $\|x+z\|_{w} \leq \|x\|_{w}$.
  This shows that  $\inf_{x \in \Sigma} \left\{\|x+ z\|_{w}-\|x\|_{w}\right\} \leq 0$ as claimed.
  For any such $x \in \Sigma$, consider $T = \supp(x)$. 
  
By the reverse triangle inequality $|x_{i}+ z_{i}|-|x_{i}| \geq -| z_{i}|$, we have
  \begin{equation}
   \|x+ z_T\|_w  - \|x\|_w = \sum_{i\in T} w_i ( |x_i+ z_i| -|x_i|) \geq -\sum_{i \in T} w_i | z_i| = - \| z_T\|_w
  \end{equation}
  Hence $\|x+ z\|_{w}-\|x\|_{w} = \|x+ z_{T}\|_{w}+\| z_{T^{c}}\|_{w}-\|x\|_{w} \geq \| z_{T^{c}}\|_{w}-\| z\|_{w} \geq \| z_{H^{c}}\|_{w}-\| z_{H}\|_{w}$. 
  
  If $\|\cdot\|_w = \|\cdot\|_1$, let $T = T(z)$ and remark that $\|z_{H^{c}}\|_{1}-\|z_{H}\|_{1} \geq \|z_{T^{c}}\|_{1}-\|z_{T}\|_{1}$

\end{proof}

The following Lemma permits to construct and to characterize elements of descent cones.

\begin{lemma}\label{lem:BuildDescentVector}
Assume that $R$ and $\Sigma$ are positively homogeneous.
For every $v_{0} \in \Sigma$ such that $R(v_{0})>0$ and any $v_{1} \in \sH$, we have that $z := v_{1}-\alpha v_{0} \in \sT_{R}(\Sigma)$ where $\alpha = \max(R(v_{1})/R(v_{0}),1)$. 
If, in addition, $\Sigma$ is homogeneous and $R$ is even,  we have conversely that any $z \in  \sT_{R}(\Sigma)$ can be written as $z = v_{1}-v_{0}$ where $v_{0} \in \Sigma$, $v_{1} \in \sH$, and $R(v_{1}) \leq R(v_{0})$.
\end{lemma}
\begin{proof}
Since $\Sigma$ is positively homogeneous, $x:= \alpha v_{0} \in \Sigma$, and $R(x+z) = R(\alpha v_{0}+z) = R(v_{1})$. If $R(v_{1}) > R(v_{0})$ then $\alpha > 1$ and $R(x+z) = R(v_{1}) = \alpha R(v_{0}) = R(\alpha v_{0}) = R(x)$. Otherwise, $\alpha = 1$ and $R(x+z) = R(v_{1}) \leq R(v_{0}) = R(x)$. In both cases we obtain that $z \in \sT_{R}(x) \subseteq \sT_{R}(\Sigma)$.

Regarding the second claim, when $z \in \sT_{R}(\Sigma)$, by definition there exists $x \in \Sigma$, $u \in \sH$ and $\gamma \in \bR$ such that $z = \gamma u$ where $R(x+u) \leq R(x)$. Denote $v_{0} := \gamma x$ and $v_{1} := v_{0}+z$. Since $\Sigma$ is homogeneous, we have $v_{0} \in \Sigma$. Since $R$ is even and positively homogeneous, $R(v_{1}) = R(\gamma x+\gamma u) = |\gamma| R(x+u) \leq |\gamma| R(x) = R(\gamma x) = R(v_{0})$. 

\end{proof}

The next lemma permits to compare $B_\Sigma^s(R)$ with $B_\Sigma^s(\|\cdot\|_1)$ (see definition in~\eqref{eq:DefBsSigma}) which was calculated in~\cite{Davies_2009} to characterize the necessary RIP condition for sparse recovery. 

\begin{lemma}\label{lem:charact_supBL2_atom}
Let $\Sigma = \Sigma_k$ be the set of $k$-sparse vectors in $\bR^{n}$ with $k<n/2$ and $1 \leq L \leq n-2k$. Assume that $R$ is positively homogeneous, subadditive, and nonzero.
 
 Consider 
\begin{eqnarray}
(H_0,v_0) &\in& \arg \max_{\stackrel{H \subseteq \{ 1, \ldots,n\}:\ |H| = k}{v \in Q_{H}}}
R(v)\\
(H_1,v_1) &\in& \arg \min_{\stackrel{H \subseteq \{ 1, \ldots,n\} \setminus{H_{0}}, |H|=k+L}{v \in Q_{H}}}
R(v).
\end{eqnarray}
\begin{enumerate}
\item We have $R(v_{0})>0$, and for any $H$ of size $k' \geq k$ and any $v \in Q_{H}$, we have
\begin{equation}\label{eq:ineqKsparseonevectorR}
R(v) \leq \frac{k'}{k}R(v_{0}).
\end{equation}
If $R = R^{\star} = \|\cdot\|_{1}$ then we have indeed equality $R^{\star}(v)= \frac{k'}{k}R^{\star}(v_{0})$.
\item We have
\begin{equation}\label{eq:ineq_BL1_atom}
B_\Sigma^{2k+L}(R) := \sup_{ z \in \sT_{R}(\Sigma)\setminus \{0\} : |\supp(z)| =2k+L} \frac{\|z_{T_2^c}\|_2^2}{\|z_{T_2}\|_2^2}  \geq \frac{  \frac{L}{k}}{ \max\left(\left(\frac{R(v_1)}{R(v_0)}\right)^{2},1\right) + 1 }
\geq \frac{ \frac{L}{k}}{ \left(\frac{L}{k}+1\right)^2 + 1 }.
\end{equation}

\end{enumerate}
\end{lemma}

\begin{proof}
As a preliminary observe that if $R^{\star} = \|\cdot\|_{1}$ then $R^{\star}(v) = |H|$ for any $H,v \in Q_{H}$, hence $H_{0},H_{1}$ can be any pair of  disjoint sets of respective sizes $k,k+L$, and $v_{i} \in Q_{H_{i}}$ can be arbitrary, for example $v_{i} = 1_{H_{i}}$.  This yields $R^{\star}(v_{0}) = k$,  $R^{\star}(v_{1}) = k+L$, hence  $R^{\star}(v_{1})= (1+L/k)R^{\star}(v_{0})$.

To prove the first claim, consider $\{G_{i}\}_{1 \leq i \leq {k' \choose k}}$ the collection of all subsets $G_{i} \subseteq H$ of size exactly $k$. Since $v \in Q_{H}$, we have $v_{G_{i}} \in Q_{G_{i}}$ for each $i$. Also, since $|G_{i}|=k$ for every $i$, by definition of $H_{0},v_{0}$ we obtain  
$\max_{i} R(v_{G_{i}}) \leq R(v_{0})$.
Notice that given a coordinate  $j \in H$, there are ${k'-1 \choose k-1}$ sets $G_i$ such that $j \in G_i$. With  $\lambda := \frac{1}{{k'-1 \choose k-1}}$ we get $v=\lambda \sum_{i} v_{G_{i}}$ hence by positive homogeneity and subadditivity of $R$ (which imply convexity)
 \begin{equation}\label{eq:ineq_th_atom1}
 R(v) = R(\lambda \sum_{i=1}^{{k' \choose k}} v_{G_{i}})
 \leq \sum_{i=1}^{{k' \choose k}} R(\lambda v_{G_{i}})
  = \lambda  \sum_{i}^{{k' \choose k}} R(v_{G_{i}})
  \leq \frac{{k' \choose k}}{{k'-1 \choose k-1}} R(v_0) = \frac{k'}{k} R(v_0).
\end{equation}
This establishes~\eqref{eq:ineqKsparseonevectorR}. With $R = R^{\star}$, we have $R^{\star}(v) = \|v\|_{1} = k'$ for $v \in Q_{H}$, hence $R^{\star}(v)  = (k'/k) R^{\star}(v_{0})$ as claimed.

For the sake of contradiction, assume that $R(v_{0}) \leq 0$. As we have just proved, this implies $R(v) \leq (n/k) R(v_{0}) \leq 0$ for every $v \in \{-1,+1\}^{n}= Q_{H}$ with $H =  \{ 1, \ldots,n\}$. By convexity of $R$ it follows that $R(v) \leq 0$ for each $v \in [-1,1]^{n}  = \tconv(Q_{H})$, and by positive homogeneity, 
\begin{equation} \label{eq:TmpNegReg0}
R(v) \leq 0,\ \forall v \in \sH.
\end{equation}
Positive homogeneity and subadditivity also imply 
\[
0 = 0 \cdot R(v_{0}) = R(0 \cdot v_{0}) = R(0) = R(-v+v) \leq R(-v) + R(v) \stackrel{\eqref{eq:TmpNegReg0}}{\leq} R(-v)
\]
for every $v \in \sH$, hence $R(v) = 0$ on $\sH$, which yields the desired contradiction since we assume that $R$ is nonzero.

Regarding the second claim, since $2k+L \leq n$ there is indeed some $H$ of size $k+L$ such that $H \cap H_{0} = \emptyset$, hence $H_{1}$ is well defined. By construction, $H_{1} \cap H_{0} = \emptyset$. 
Since $R(v_{0})>0$, $R$ is positively homogeneous and $\Sigma$ is homogeneous, by \Cref{lem:BuildDescentVector}, 
 $z = -\alpha  v_0 +  v_1 \in \sT_{R}(\Sigma)$ with $\alpha :=  \max(R(v_1)/R(v_0),1)$. 
Observe that $|\supp(z)| = |H_{0}|+|H_{1}| = 2k+L$. Since $\alpha \geq 1$ and all nonzero entries of $v_{0},v_{1}$ have magnitude one, a set of $2k$ largest components of $z$ is $T_{2} = H_{0} \cup T'_{1}$ with $T'_{1}$ any subset of $H_{1}$ with $k$ components, and we obtain~\eqref{eq:ineq_BL1_atom}.
once we observe that
\[
 \frac{\|z_{T_2^c}\|_2^2}{\|z_{T_2}\|_2^2} = \frac{  L}{ k\alpha^2 + k } = \frac{L/k}{\alpha^{2}+1}. 
\]

\end{proof}

\begin{lemma}\label{lem:FlatVectors}
Consider $c_{\infty},c_{1} >0$, an integer $n \geq 2$, and the optimization problem
\begin{equation}
\sup_{x \in \mathbb{R}_{+}^{n}: \|x\|_{\infty} \leq c_{\infty}; \|x\|_{1} \leq c_{1}} \|x\|_{2}^{2}.
\end{equation}
If $c_{1} \geq c_{\infty}$ then there exists $1 \leq L \leq n-1$ and $0 \leq \theta \leq 1$ such that
\[
x^{*} := c_{\infty} (\underbrace{1,\ldots,1}_{L \geq 1},\theta,\underbrace{0,\ldots,0}_{n-(L+1) \geq 0})
\]
is a maximizer. Otherwise, a maximizer is $x^{*} = (c_{1},0,\ldots,0)$.
\end{lemma}
\begin{proof}
Standard compactness arguments show the existence of a maximizer $x^{*}$. We  distinguish two cases: 
\begin{itemize}
\item[$\bullet$] If $\|x^{*}\|_{\infty} < c_{\infty}$ then $x^{*}$ is indeed a maximizer of the Euclidean norm under an $\ell^{1}$ constraint, hence $x^{*}$ is a Dirac: without loss of generality, $x^{*} = (c_{1},0,\ldots,0)$ so that $c_{1} = \|x^{*}\|_{\infty} < c_{\infty}$.
\item[$\bullet$] Otherwise $\|x^{*}\|_{\infty} = c_{\infty}$, in which case we show that 
all entries of $x^{*}$, except at most one, are either zero or equal to $c_{\infty}$. 
For the sake of contradiction, assume that $x^{*}$ contains two distinct entries with values $0<a  < b< c_{\infty}$, then for small enough $t>0$, replacing these entries with $0<a-t<b+t< c_{\infty}$ and keeping all other entries unchanged would lead to a vector $x$ satisfying $\|x\|_{\infty} = \|x^{*}\|_{\infty} = c_{\infty}$, $\|x\|_{1} = \|x^{*}\|_{1}$. However, since $\|x\|_{2}^{2}-\|x^{*}\|_{2}^{2} = (a-t)^{2}+(b+t)^{2}-(a^{2}+b^{2}) = 2t^{2}+2(b-a)t>0$. Since $x^{*}$ has optimal objective value, this yields the desired contradiction. Since the objective value and the constraints are invariant to index permutations, there is thus a maximizer with the claimed shape, and we have $c_{1} \geq \|x^{*}\|_{1} \geq \|x^{*}\|_{\infty} = c_{\infty}$. 
\end{itemize}
The two cases respectively correspond to $c_{1}<c_{\infty}$ or $c_{1} \geq c_{\infty}$, which are mutually exclusive, hence the conclusion.

\end{proof}

\begin{lemma}[{\cite{Davies_2009}}]\label{lem:charact_supBL2_l1}
Consider $\Sigma = \Sigma_k \subseteq \bR^{n}$.  We have
\begin{equation}
B_\Sigma(\|\cdot\|_1) = \max_{1 \leq L \leq n-2k}  \frac{ \frac{L}{k}}{ \left(\frac{L}{k}+1\right)^2 + 1 }. 
\end{equation}
\end{lemma}

\begin{proof}
With $B_{\Sigma}^{s}(R)$ defined in~\eqref{eq:DefBsSigma}, and recalling the expression~\eqref{eq:PropBSigmaRKSparse} of $B_{\Sigma}(R)$,
we have \[B_\Sigma(\|\cdot\|_1) = \max_{1 \leq L \leq n-2k}  B_\Sigma^{2k+L}(\|\cdot\|_1) \]

By \Cref{lem:charact_supBL2_atom}, $\frac{R^{\star}(v_{1})}{R^{\star}(v_{0})} = (L/k)+1 > 1$ and $B_\Sigma^{2k+L}(\|\cdot\|_1) \geq  \frac{ \frac{L}{k}}{ \left(\frac{R^{\star}(v_{1})}{R^{\star}(v_{0})}\right)^2 + 1 } =  \frac{ \frac{L}{k}}{ \left(\frac{L}{k}+1\right)^2 + 1 }$. This implies 

\begin{equation}
B_\Sigma(\|\cdot\|_1) \geq \max_{1 \leq L \leq n-2k}  \frac{ \frac{L}{k}}{ \left(\frac{L}{k}+1\right)^2 + 1 } ,
\end{equation}
and there only remains to show there is indeed equality. We isolate  this result from \cite{Davies_2009} for completeness. This will also help understand the case  of sparsity in levels in \Cref{sec:proofslevels}.

First, we show we can restrict the maximization used to express $B_{\Sigma}(\|\cdot\|_{1})$ (cf \eqref{eq:PropBSigmaRKSparse}) over vectors $z$ having constant amplitude $\alpha>0$ on $T(z)$.

Indeed, consider $z \neq 0$ such that $z \in \sT_{\Sigma_{k}}(\|\cdot\|_{1})$. By \Cref{lem:opt_support_wl1}, we have
$\|z_{T^c}\|_1 \leq \|z_T\|_1$ with $T = T(z)$ a set of $k$ indices of components of largest magnitude of $z$. 
Assume  that there are $i \neq j$ in $T$ such that $|z_{i}| \neq |z_{j}|$.
Let $y$ such that $y_l = z_l$ for $l \notin \{i,j\}$ and $y_i= y_j = (|z_i|+|z_j|)/2$. The set $T$ remains a support of the $k$ largest amplitudes in $y$, and $T_{2}=T_{2}(z)$ remains a support of the $2k$ largest amplitudes in $y$. Moreover,
we have $\|y_T\|_1 = \|z_T\|_1 \geq \|z_{T^c}\|_1 = \|y_{T^c}\|_1 = \|-y_T +y\|_1$ hence we have $y \in \sT_{\Sigma_{k}}(\|\cdot\|_{1})$.
Since $ \|y_{T_{2}}\|_2^2-\|z_{T_{2}}\|_2^2 =  \|y_{T}\|_2^2-\|z_{T}\|_2^2 = 2 [ (|z_i|+|z_j|)/2]^2 -|z_i|^2-|z_j|^2 =- (|z_i|-|z_j|)^2/2 < 0 $ and $\|y_{T_{2}^{c}}\|_{2}^{2}= \|z_{T_{2}^{c}}\|_{2}^{2}$ we have $\|y_{T_{2}^{c}}\|_{2}^{2}/\|y_{T_{2}}\|_{2}^{2} > \|z_{T_{2}^{c}}\|_{2}^{2}/\|z_{T_{2}}\|_{2}^{2}$.

Second, the same reasoning on $T' = T_{2} \setminus T$, shows that we can further restrict the maximization used to define $B_{\Sigma}(\|\cdot\|_{1})$ to vectors having constant amplitude $0\leq \beta \leq \alpha$ over $T'$. This leads to
\begin{equation}
B_\Sigma(\|\cdot\|_1) =\sup_{z\neq 0: \|z_{T^c}\|_1\leq \|z_T\|_1} \frac{\|z_{T_2^c}\|_2^2}{\|z_{T_2}\|_2^2} =\sup_{\alpha,\beta: \alpha \geq \beta>0} \sup_{x\in\bR^{n-2k} : \|x\|_\infty \leq \beta, \|x\|_1 \leq k(\alpha-\beta)}\frac{ \|x\|_2^2}{k(\alpha^2+\beta^2)}.
\end{equation}
Using \Cref{lem:FlatVectors}, the supremum with respect to $x$ is reached with vectors with the shape
\[
  (\underbrace{\beta,\ldots,\beta}_{L},\theta,\underbrace{0,\ldots,0}_{n-2k-(L+1) \geq 0}) 
\]
with $0\leq \theta \leq \beta$ and $0\leq L \leq n-2k-1$. We deduce 
\begin{equation}
\begin{split}
B_\Sigma(\|\cdot\|_1)&= \sup_{\alpha,\beta: \alpha \geq \beta>0} \sup_{\stackrel{L,\theta : 0\leq L \leq n-2k-1 ,0\leq\theta\leq\beta}{ \theta\leq k\alpha-(k+L)\beta}} \frac{ L\beta^2 +\theta^2}{k(\alpha^2+\beta^2)}\\
&= \max_{ 0\leq L \leq n-2k-1 }\sup_{\alpha,\beta: \alpha \geq \beta>0} \sup_{\stackrel{\theta :0\leq\theta\leq\beta} {\theta\leq k\alpha-(k+L)\beta}} \frac{ L\beta^2 +\theta^2}{k(\alpha^2+\beta^2)}\\
\end{split}
\end{equation}

When $0 \leq \beta \leq  k\alpha-(k+L)\beta$ we have
\begin{equation}
  \sup_{\theta :0\leq\theta\leq\beta, \theta\leq k\alpha-(k+L)\beta} \frac{ L\beta^2 +\theta^2}{k(\alpha^2+\beta^2)} = \frac{ (L+1)\beta^2 }{k(\alpha^2+\beta^2)}
\end{equation}
while when  $\beta \geq  k\alpha-(k+L)\beta \geq 0$ we have
\begin{equation}
  \sup_{\theta :0\leq\theta\leq\beta, \theta\leq k\alpha-(k+L)\beta} \frac{ L\beta^2 +\theta^2}{k(\alpha^2+\beta^2)} = \frac{ L\beta^2 + ( k\alpha-(k+L)\beta)^{2}}{k(\alpha^2+\beta^2)}.
\end{equation}

On the one hand, when $0 < \beta \leq \alpha$ satisfies $\beta \leq  k\alpha-(k+L)\beta$ we have $\alpha \geq (1+(L+1)/k) \beta$ hence
\begin{equation}
  \sup_{\theta :0\leq\theta\leq\beta, \theta\leq k\alpha-(k+L)\beta} \frac{ L\beta^2 +\theta^2}{k(\alpha^2+\beta^2)} = \frac{ (L+1)\beta^2 }{k(\alpha^2+\beta^2)} 
  = \frac{(L+1)/k}{(\alpha/\beta)^{2}+1}
  \leq 
  \frac{(L+1)/k}{[1+(L+1)/k]^{2}+1}.
\end{equation}
On the other hand, when $0 < \beta \leq \alpha$ satisfies $\beta \geq  k\alpha-(k+L)\beta \geq 0$ we have $(1+(L+1)/k))\beta \geq \alpha \geq (1+L/k)\beta$ and, denoting $g(t) := \frac{L/k+kt^{2}}{(1+L/k+t)^{2}+1}$ for $t \geq 0$, we get
\begin{equation}
  \sup_{\theta :0\leq\theta\leq\beta, \theta\leq k\alpha-(k+L)\beta} \frac{ L\beta^2 +\theta^2}{k(\alpha^2+\beta^2)} = \frac{ L\beta^2 + ( k\alpha-(k+L)\beta)^{2}}{k(\alpha^2+\beta^2)}
  =
  \frac{L/k+k [\alpha/\beta-(1+L/k)]^{2}}{(\alpha/\beta)^{2}+1}
  = g(\alpha/\beta-(1+L/k)).
\end{equation}
A simple function study shows that $g'(t)$ is positively proportional to a second degree polynomial $P(t)$ with positive leading coefficient and such that $P(0)<0$. It follows that there is $t_{0}>0$ such that $g'(t) \leq 0$ for $0 \leq t \leq t_{0}$ and $g'(t) \geq 0$ for $t \geq t_{0}$. Hence, $g$ is decreasing on $[0,t_{0}]$ and increasing on $[t_{0},+\infty)$, so that 
\[
g(\alpha/\beta-(1+L/k)) \leq \sup_{0 \leq t \leq 1/k} g(t) = \max\left(g(0), g(1/k)\right) 
= \max\left(
\frac{L/k}{(1+L/k)^{2}+1},\frac{(L+1)/k}{(1+(L+1)/k)^{2}+1}
\right).
\]
As all the above bounds also hold if $\beta=0$, we obtain the claimed result.

\end{proof}

\begin{remark}\label{rk:opt_B_kL}
 
The maximum value of $ \frac{ L/k}{ ((k+L)/k)^2+1}$  (with respect to $L$) is reached for $L/k$ maximizing $f(u) = u/((u+1)^2+1)$ (which is maximized at $\sqrt{2}$ over $\bR$). We verify that it matches the necessary RIP condition $\frac{1}{\sqrt{2}}$ from \cite{Davies_2009}, $f(\sqrt{2})= 2\sqrt{2}/(2+\sqrt{2})$ which gives $\gamma_\Sigma(\|\cdot\|_1) = (4+3\sqrt{2})/\sqrt{2} = \frac{\sqrt{2}+1}{\sqrt{2}-1}$.
\end{remark}

\subsubsection{Lemmas for the proof of \texorpdfstring{\Cref{th:RIP_nec_atom_rank}}{Theorem \ref{th:RIP_nec_atom_rank}} }\label{sec:proofLRnex}
 
Given a matrix $U$, we denote  $U_{k:l}$ the restriction of $U$ to its rows ${k,\ldots,l}$. We denote $O(n)$ the orthogonal group. Given a symmetric matrix $z$, we write $\eig(z)$ the vector of eigenvalues ordered decreasingly with respect to their absolute value.
Given a vector $x$ of size $n$,  we write $\diag(x)$ the diagonal matrix with diagonal equal to $x$. To match the notations for the case of sparsity, given a matrix $z=U^T\diag(w)U$, we write $z_{H}  =U^T\diag(w_H)U $ and $Q_{H}$ as in the previous section. We denote  $T=\{1,..,r\}$ and $T_2 =\{1,..,2r\}$.  We denote $\|\cdot\|_F$ the Frobenius norm.

Using the same demonstration as~\Cref{lem:opt_support_wl1} we characterize the descent cones of the nuclear norm.
\begin{lemma} \label{lem:opt_support_wnuclear}
 Let $\Sigma = \Sigma_r$. Let $\|\cdot\|_w$ be a weighted nuclear-norm. Let $z \in \sT_{\|\cdot\|_w}(\Sigma)$. There is a support $H$ of size $\leq r$ such that  
 \begin{equation}
\|z_{H^{c}}\|_{w}-\|z_{H}\|_{w} = \inf_{x \in \Sigma} \left\{\|x+z\|_{w}-\|x\|_{w}\right\} \leq 0,
 \end{equation}
\ie the infimum is achieved at $x^* = -z_{H}$. Moreover, if $\|\cdot\|_w= \|\cdot\|_*$, $H = T(z)$.
\end{lemma}

\begin{lemma}\label{lem:charact_supBL2_atom_rank}
Let $\Sigma = \Sigma_r$ be the set of $n \times n$ symmetric matrices  with rank at most $r$ with $r < n/2$, and $1 \leq L \leq n-2r$. 
Assume $R$ is positively homogeneous, subadditive and nonzero. Consider the supports $H_0 =\{1,2,..,r\}$  and $H_1 = \{r+1,\ldots,2r+L\}$. 
\begin{eqnarray}
(U_0,v_0) &\in& \arg \max_{U \in O(n), v \in Q_{H_{0}}} \|U^T\diag(v)U \|_\sA ,\\
(U_1,v_1) &\in& \arg \min_{U \in O(n),v \in Q_{H_{1}}:\; U_{0,1:r}U_{r+1:2r+L}^T  =0  } \|U^T\diag(v)U \|_\sA.
\end{eqnarray}
\begin{enumerate}
 \item We have $R(U_0^Tv_0U_0)>0$, and for any $H$ of size $r' \geq r$,  $V\in O(n)$ and $w \in Q_H$, we have
 \begin{equation}\label{eq:ineqrakRonevectorR}
   R(V^T\diag(w)V) \leq \frac{r'}{r} R(U_0^Tv_0U_0) .
 \end{equation}
If $R = R^\star = \|\cdot\|_*$ then we have indeed equality $R(V^T\diag(w)V) = \frac{r'}{r} R(U_0^Tv_0U_0)$.
 
 \item We have 
 \begin{equation}\label{eq:ineq_Bnuc_atom}
 \begin{split}
 B_\Sigma^{L+2r}(R) := \sup_{ z \in \sT_{\|\cdot\|_\sA}(\Sigma)\setminus \{0\} : |\supp(\eig(z))| =2r+L} \frac{\|z_{T_2^c}\|_F^2}{\|z_{T_2}\|_F^2}  &\geq \frac{  \frac{L}{r}}{ \left( \max \left(\frac{R(U_1^T\diag(v_1)U_1)}{R(U_0^T\diag(v_0)U_0)},1\right)\right)^2 + 1 } \\
 &\geq \frac{  \frac{L}{r}}{ \left( \frac{L}{r}+1\right)^2 + 1 }. \\
 \end{split}
 \end{equation}

\end{enumerate}
\end{lemma}

\begin{proof}
As a preliminary observe that if $R^{\star} = \|\cdot\|_{*}$ then $R^{\star}(V^TwV) = |H|$ for any $H,w \in Q_{H}, V \in O(n)$, hence  $w_{i} \in Q_{H_{i}}$ can be arbitrary, for example $w_{i} = 1_{H_{i}}$.  This yields $R^{\star}(U_0^T\diag(v_0)U_0) = r$,  $R^{\star}(U_1^T\diag(v_1)U_1) = r+L$, hence  $R^{\star}(U_1^T\diag(v_1)U_1)= (1+L/r)R^{\star}((U_0^T\diag(v_0)U_0)$.

To prove the first claim, consider $\{G_{i}\}_{1 \leq i \leq {r' \choose r}}$ the collection of all subsets $G_{i} \subseteq H$ of size exactly $r$. Since $w \in Q_{H}$, we have $w_{G_{i}} \in Q_{G_{i}}$ for each $i$. Also, since $|G_{i}|=r$ for every $i$, by definition of $H_{0},v_{0}$ and remarking that the maximization over  $O(n)$ permits to consider any permutation of the support, we obtain  
$\max_{i} R(V^T\diag(v_{G_{i}})V) \leq R(U_0^T\diag(v_{0})U_0)$.

Notice that given a coordinate  $j \in H$, there are ${r'-1 \choose r-1}$ sets $G_i$ such that $j \in G_i$. With  $\lambda := \frac{1}{{r'-1 \choose r-1}}$, we get $V^T\diag(w)V=V^T\lambda \sum_{i} \diag(w_{G_{i}})V$ hence by positive homogeneity and subadditivity of $R$ (which imply convexity)
 \begin{equation}
 \begin{split}
 R(V^TwV) = R(\lambda V^T\sum_{i=1}^{{r' \choose r}} \diag(w_{G_{i}})V)
 &\leq \sum_{i=1}^{{r' \choose r}} R(V^T\lambda \diag(w_{G_{i}})V)
  = \lambda  \sum_{i}^{{r' \choose r}} R(V^T\diag(w_{G_{i}})V)\\
 & \leq \frac{{r' \choose r}}{{r'-1 \choose r-1}} R(U_0^T\diag(v_0)U_0) = \frac{r'}{r} R(U_0^T\diag(v_0)U_0).
  \end{split}
\end{equation}
This establishes~\eqref{eq:ineqrakRonevectorR}. With $R = R^{\star}$, we have $R^{\star}(V^T\diag(w)V) = \|w\|_{1} = r'$ for $w \in Q_{H}$, hence $R^{\star}(V^T\diag(w)V)  = (r'/r) R^{\star}(U_0^T\diag(v_{0})U_0)$ as claimed.

For the sake of contradiction, assume that $R(U_0^T\diag(v_{0})U_0) \leq 0$. As we have just proved, this implies $R(V^T\diag(w)V) \leq (n/k) R(U_0^T\diag(v_{0})U_0) \leq 0$ for every $w \in \{-1,+1\}^{n}= Q_{H}$ with $H =  \{ 1, \ldots,n\}$ and $V \in O(n)$. By convexity of $R$ it follows that $R(V^T\diag(w)V) \leq 0$ for each $w \in [-1,1]^{n}  = \tconv(Q_{H})$, and by positive homogeneity, 
\begin{equation}
\label{eq:TmpNegReg}
R(V^T \diag(w)V) \leq 0,\quad \forall w \in \bR^n.
\end{equation}
Positive homogeneity and subadditivity also imply 
\begin{align*}
0 = 0 \cdot R(U_0^T\diag(v_{0})U_0) = R(0 \cdot U_0^T\diag(v_{0})U_0) = R(0) &= R(-V^T\diag(w)V+V^T\diag(w)V) \\
& \leq R(-V^T\diag(w)V) + R(V^T\diag(w)V)\\
&  \stackrel{\eqref{eq:TmpNegReg}}{\leq} R(-V^T\diag(w)V)
\end{align*}
 for every $V^T\diag(w)V \in \sH$, hence $R(V^T\diag(w)V) = 0$ on $\sH$, which yields the desired contradiction since we assume that $R$ is nonzero.

Regarding the second claim, since $2r+L \leq n$, by construction, $H_{1} \cap H_{0} = \emptyset$. 
Since $R(U_0^T\diag(v_{0})U_0)>0$, $R$ is positively homogeneous and the set $\Sigma$ is homogeneous, and we have by \Cref{lem:BuildDescentVector},
 $z = -\alpha  U_0^T\diag(v_{0})U_0 +  U_1^T\diag(v_1)U_1 \in \sT_{R}(\Sigma)$ with $\alpha :=  \max(R(U_1^T\diag(v_1)U_1)/R(U_0^T\diag(v_{0})U_0 ),1)$. 
Observe that $|\supp(\eig(z))| = |H_{0}|+|H_{1}| = 2r+L$. Since $\alpha \geq 1$ and all nonzero entries of $v_{0},v_{1}$ have magnitude one, a set of the $2r$ largest components of $eig(z)$ is $T_{2} = H_{0} \cup T'_{1}$ with $T'_{1}$ any subset of $H_{1}$ with $k$ components, and we obtain~\eqref{eq:ineq_Bnuc_atom}.
once we observe that
\begin{equation}
 \frac{\|z_{T_2^c}\|_2^2}{\|z_{T_2}\|_2^2} = \frac{  L}{ r\alpha^2 + r } = \frac{L/r}{\alpha^{2}+1}. 
\end{equation}

\end{proof}

\begin{lemma}\label{lem:charact_supBL2_nuc}
Let $\Sigma = \Sigma_r$.   Then 
\begin{equation}
B_\Sigma(\|\cdot\|_*) = \max_{0\leq L\leq n-2r}  \frac{ \frac{L}{r}}{ \left(\frac{L}{r}+1\right)^2 + 1 }.  \\
\end{equation}
\end{lemma}

\begin{proof}

 We have $z \in \sT_{\|\cdot\|_*}(\Sigma_r)$ is equivalent to $\|z_{T_2^c}\|_* + \|z_{T'}\|_* \leq \|z_{T}\|_*$ where $T' = \supp(z) \setminus (T_2^c \cup T)$ (\Cref{lem:opt_support_wnuclear}). Hence, 
 \begin{equation}
 \begin{split}
B_\Sigma^{L+2r}(\|\cdot\|_*) &= \sup_{z : \|z_{T_2^c}\|_* + \|z_{T'}\|_* \leq \|z_{T}\|_*} \frac{\|z_{T_2^c}\|_F^2}{\|z_{T_2}\|_F^2}  .\\
\end{split}
\end{equation}
Using the fact that $\|z\|_* = \|\eig(z)\|_1$ and   $\|z\|_F = \|\eig(z)\|_2$, we fall on the expression of $B_\Sigma^{L+2r}(\|\cdot\|_1)$ and get the result using \Cref{lem:charact_supBL2_l1}.

\end{proof}

\subsection{Proofs for \texorpdfstring{\Cref{sec:suff_RIP}}{Theorem \ref{sec:suff_RIP}}  }\label{sec:ProofSuffRIP}

\begin{proof}[Proof of \Cref{lem:charac_suff_RIP}]
 The constant $\delta_\Sigma^{\mathtt{suff}}(R)$ \cite{Traonmilin_2016}[Eq. (5)] has the following expression:
\begin{equation}\label{eq:ExplicitExpressionSuffRIP}
 \delta_\Sigma^{\mathtt{suff}}(R)  = \underset{z\in \sT_R(\Sigma) \setminus \{0\} }{\inf}  \underset{x\in \Sigma}{\sup} \frac{-\re \ls x,z \rs}{\|x\|_\sH \sqrt{ \|x+z\|_\Sigma^2 -  \|x\|_\sH^2 - 2\re \ls x,z \rs}} .
\end{equation}
Considering any nonzero $z \in \sH$, since $\Sigma$ is a union of subspaces and $\Sigma \cap S(1)$ is compact, by \Cref{lem:ProjectionExists} the set $P_{\Sigma}(z)$ is not empty and $\langle P_{\Sigma}(z),z\rangle = \|P_{\Sigma}(z)\|_{\sH}^{2}$ is unambiguous. Choosing an arbitrary $y \in P_{\Sigma}(z)$ and setting $x = -y$, we obtain

 \begin{equation*}
\begin{split}
  \underset{x\in \Sigma}{\sup} \frac{-\re \ls x,z \rs}{\|x\|_\sH \sqrt{ \|x+z\|_\Sigma^2 -  \|x\|_\sH^2 - 2\re \ls x,z \rs}}& \geq  \frac{\|P_{\Sigma}(z)\|_\sH^2}{\|P_{\Sigma}(z)\|_\sH \sqrt{ \|z-P_{\Sigma}(z)\|_\Sigma^2 -  \|P_{\Sigma}(z)\|_\sH^2 + 2\|P_{\Sigma}(z)\|_\sH^2}}\\
  & =   \frac{1}{ \sqrt{ \underset{z\in \sT_R(\Sigma) \setminus \{0\} }{\sup}  \frac{\|z-P_{\Sigma}(z) \|_\Sigma^2}{\|P_{\Sigma}(z)\|_\sH^2} + 1 }}.
 \end{split}
\end{equation*}
Considering the infimum over $z \in \sT_{R}(\Sigma) \setminus \{0\}$ yields the first claim. 
Let us now proceed to the second claim.

Given $z \in \sT_{R}(\Sigma) \setminus \{0\}$, consider an arbitrary $x \in \Sigma$, and $V \in \mathcal{V}$ such that $x \in V$. 
With Fact~\ref{fact:atom2}, for every $v \in \sH$, $\|v\|_\Sigma^{2}$ is the infimum of $\sum_i \lambda_i \|u_i\|_{\sH}^2$ over convex decompositions $v = \sum_i \lambda_i u_i$ over $\Sigma$, hence there exists $u_i\in \Sigma$, $\lambda_i\geq0$  such that $\sum_i \lambda_i =1$, $\sum_i \lambda_i u_i = x+z$ and
 \[
 \|x+z\|_\Sigma^2 = \sum_i \lambda_i \|u_i\|_{\sH}^2 .
 \]

Since $V \subset \Sigma$, $u_{i,V} := P_{V} u_{i} \in \Sigma$. By the additional assumption, since $u_{i} \in \Sigma$ we also have and $u_{i,V^{\perp}} := P_{V^{\perp}}u_{i} \in \Sigma$ for each $i$. Observe also that $P_{V^{\perp}}x = 0$. Hence, with the notations $z_{V}= P_{V}z$, $z_{V^{\perp}} = P_{V^{\perp}}z$, we have the convex decompositions 
\begin{align*}
z_{V^\perp}  =  P_{V^{\perp}}(x+z) &= \sum \lambda_{i} u_{i,V^\perp}\\
 x+z_{V} = P_{V}(x+z) &= \sum \lambda_{i} u_{i,V}. 
\end{align*}
Using Jensen's inequality for the convex functions $\|\cdot\|_{\Sigma}^{2}$ and $\|\cdot\|_{\sH}^{2}$ and  the identity $\|v\|_\Sigma^2 =\|v\|_\sH^2$ for $v \in \Sigma$ (Fact~\ref{fact:atom1}), we have
\begin{align*}
 \|z_{V^\perp}\|_\Sigma^2  +  \|x+z_V\|_\sH^2 
   \leq \sum_i \lambda_i \|u_{i,V^\perp}\|_\Sigma^2 + \sum_i\lambda_i \| u_{i,V}\|_\sH^2
   &= \sum_i \lambda_i \|u_{i,V^\perp}\|_\sH^2 + \sum_i\lambda_i \| u_{i,V}\|_\sH^2\\
   &= \sum_i \lambda_i \|u_{i}\|_\sH^2  = \|x+z\|_\Sigma^2.
\end{align*}
Since $P_{V}$ is the (linear and self-adjoint) orthogonal projection onto $V$, we have $\re \langle x,z_{V}\rangle =\re  \langle x,P_{V}z\rangle =\re  \langle P_{V}x,z\rangle =\re  \langle x,z\rangle$, and we obtain 
\begin{equation}
\begin{split}
  \|z_{V^\perp}\|_\Sigma^2  +   \|z_V\|_\sH^2  &\leq \|x+z\|_\Sigma^2 -  \|x+z_V\|_\sH^2 + \|z_V\|_\sH^2\\
  \|z_{V^\perp}\|_\Sigma^2  +   \|z_V\|_\sH^2  &\leq \|x+z\|_\Sigma^2 - \|x\|_\sH^2 -2\re \ls x, z_V \rs \\
   \|z_{V^\perp}\|_\Sigma^2  +   \|z_V\|_\sH^2  &\leq \|x+z\|_\Sigma^2 - \|x\|_\sH^2 -2\re \ls x, z\rs. \\
 \end{split}
\end{equation}
Using Cauchy-Schwarz inequality, we have $\left(\re(\langle x,z\rangle\right)^{2} = \left(\re(\langle x,z_{V}\rangle\right)^{2} \leq \|x\|_{\sH}^{2}\|z_{V}\|_{\sH}^{2}$. Denoting $V_0$ such that 
$P_{V_0}(z) \in P_\Sigma(z)$, we get

\begin{align*}
  \left(\re \ls x, z\rs\right)^2  \left(\|z_{V^\perp}\|_\Sigma^2  +   \|z_V\|_\sH^2 \right) &
  \leq \|x\|_\sH^2\|z_V\|_\sH^2\left(\|x+z\|_\Sigma^2 - \|x\|_\sH^2 -2\re \ls x, z\rs\right) \\
  \frac{\left(\re \ls x, z\rs\right)^2}{ \|x\|_\sH^2\left(\|x+z\|_\Sigma^2 - \|x\|_\sH^2 -2\re \ls x, z\rs\right) }  
  &\leq\frac{\|z_V\|_\sH^2}{\left(\|z_{V^\perp}\|_\Sigma^2  +   \|z_V\|_\sH^2 \right)} =  \frac{1}{\frac{\|z_{V^\perp}\|_\Sigma^2}{\|z_{V}\|_\sH^2}  +1} 
  \leq \frac{1}{\frac{\|z-P_{V}z\|_\Sigma^2}{\|P_{\Sigma}(z)\|_\sH^2}  +1} ,
 \end{align*}
where the last inequality (we could use here the weaker alternative assumption $P_\Sigma(z) \cap \arg \min_{x \in \Sigma} \|x-z\|_\Sigma/\|x\|_{\sH} \neq \emptyset$) uses that $z_{V^{\perp}} = z-P_{V}z$ and
$\|P_{V_{0}}z\|_{\sH} = \|P_{\Sigma}(z)\|_{\sH} \geq \|P_{V}(z)\|_{\sH} = \|z_{V}\|_{\sH}$.
To conclude, we use the additional hypothesis $P_\Sigma(z) \subseteq \arg\min_{x\in\Sigma} \|x-z\|_\Sigma$, which implies  $\|z-P_{\Sigma}(z)\|_{\Sigma} \leq \|z-P_{V}z\|_{\Sigma}$ since $P_{V}z \in \Sigma$

 \begin{align*}
  \underset{x\in \Sigma}{\sup} \frac{-\re \ls x,z \rs}{\|x\|_\sH \sqrt{ \|x+z\|_\Sigma^2 -  \|x\|_\sH^2 - 2\re \ls x,z \rs}} \leq  \frac{1}{ \sqrt{ \underset{z\in \sT_R(\Sigma) \setminus \{0\} }{\sup}  \frac{\|z - P_\Sigma(z)\|_\Sigma^2}{\|P_\Sigma(z)\|_\sH^2} + 1 }}. 
\end{align*}

\end{proof}

To replicate the proof used in the necessary case, we show a monotonicity property of $\|\cdot\|_\Sigma$.

\begin{lemma}\label{lem:sigma_norm_monotonic}
Consider a model set $\Sigma \subset \sH$, $\|\cdot\|_{\Sigma}$ the atomic  ``norm'' induced by $\Sigma$, and $D:\sH \to \sH$ a linear operator. 
If $D \Sigma \subseteq \Sigma$ and $\|D\|_{op} := \sup_{\|v\|_{\sH} \leq 1} \|Dv\|_{\sH} \leq 1$ then 
\begin{equation}
\|Dv\|_{\Sigma} \leq \|v\|_{\Sigma},\quad \forall v \in \sH.
\end{equation}
\end{lemma}
\begin{proof}
 Let $\lambda_i,u_i$ such that $u_i \in \Sigma$, $\sum_i \lambda_i =1$, $\sum_i \lambda_i u_i = v$.
 Denoting $u'_{i} = D u_{i}$ we have $u'_{i} \in \Sigma$ and $Dv = \sum \lambda_{i} u'_{i}$. 
 By Jensen's inequality and the fact that $\|u\|_{\Sigma} = \|u\|_{\sH}$ for any $u \in \Sigma$ (Fact~\ref{fact:atom1}), it follows that
 \begin{equation}
 \begin{split}
  \|Dv\|_\Sigma^2 &\leq 
  \sum \lambda_i \|u'_i\|_\Sigma^2 =  
  \sum \lambda_i \|u'_i\|_\sH^2 
  =
  \sum \lambda_i \|Du_i\|_\sH^2 
  \leq 
  \sum \lambda_i \|u_i\|_\sH^2.
  \end{split}
 \end{equation}
 With Fact~\ref{fact:atom2}, $\|v\|_\Sigma^2$ is the infimum of the right-hand side over all such decompositions $v=\sum \lambda_i u_i$. 

\end{proof} 

 \begin{corollary} \label{cor:sigma_norm}
With $\Sigma := \Sigma_{k}$ the set of $k$-sparse vectors in $\sH = \bR^{n}$, we have:
 \begin{enumerate}
 \item the norm $\|\cdot\|_{\Sigma}$ is invariant by permutation and coordinate sign changes;
 \item for any vectors $v,v' \in \sH$ such that $|v_{j}| \leq |v'_{j}|$ for all $j$ we have $\|v\|_{\Sigma} \leq \|v'\|_{\Sigma}$;
 \item consider any vector $z$, and $T_k$ a subset indexing $k$ components of the largest magnitude, \ie  $\min_{i \in T} |z_{i}| \geq \max_{j \notin T} |z_{j}|$, with $|T| = k$. Then
 \begin{eqnarray}
 \max_{|T| \leq k} \|z_T\|_\Sigma &=& \|z_{T_k}\|_\Sigma\\
 \min_{|T| \leq k} \|z-z_T\|_\Sigma &=& \|z-z_{T_k}\|_\Sigma.
 \end{eqnarray}
 \end{enumerate}
 
 \end{corollary}
 \begin{proof}
 We show the three properties separately. 
 
 \begin{itemize}
  \item \textbf{Property 1:} Let $\pi$ be a permutation of $(1,\ldots,n)$ and $\epsilon_{1},\ldots,\epsilon_{n} \in \{\pm1\}$. Define $D$ by $(Du)_{i} = \epsilon_{i} u_{\pi(i)}$. Observe that $D \Sigma_{k} \subseteq \Sigma_{k}$ and $\|D\|_{op} =1$. Conclude using \Cref{lem:sigma_norm_monotonic} that $\|Du\|_{\Sigma} \leq \|u\|_{\Sigma}$ for any $u \in \sH$. The same holds with $D' = D^{-1}$, hence $\|u\|_{\Sigma} = \|D^{-1}D u\|_{\Sigma} \leq \|Du\|_{\Sigma}$ for any $u$. This shows  $\|D\cdot\|_{\Sigma} = \|\cdot\|_{\Sigma}$.
 
 \item \textbf{Property 2:} Given the assumptions on $v,v'$, the linear operator defined by $(Du)_{i} = v_{i} u_{i}/v'_{i}$ if $v'_{i} \neq 0$ (and $(Du)_{i} = 0$ otherwise) satisfies $D \Sigma \subseteq \Sigma$ and $\|D\|_{op} \leq 1$ hence, using \Cref{lem:sigma_norm_monotonic} again, $\|v\|_{\Sigma} = \|Dv'\|_{\Sigma} \leq \|v'\|_{\Sigma}$.
 
\item \textbf{Property 3:} By the invariance by permutation and coordinate sign changes of $\|\cdot\|_{\Sigma}$, it is sufficient to prove the result when $z_{1} \geq \ldots \geq z_{n} \geq 0$ and $T_{k} = \{1,\ldots,k\}$. Given $T$ of size $k$, there is a permutation $\phi$ of $(1,\ldots,n)$ such that $T = \{\phi(1),\ldots,\phi(k)\}$ where $\phi(1) < \ldots < \phi(k)$. 
It follows that $z_{\phi(i)} \leq z_{i}$ for $1 \leq i \leq k$.  Hence by Property 2, we have $\|z_{T}\|_{\Sigma} = \|(z_{\phi(1)},\ldots,z_{\phi(k)},0,\ldots,0)\|_{\Sigma} \leq \|(z_{1},\ldots,z_{k},0,\ldots,0)\|_{\Sigma} =  \|z_{T_{k}}\|_{\Sigma}$. A similar argument using $T^{c}$ yields $\|z-z_{T}\|_{\Sigma} \geq \|z-z_{T_{k}}\|_{\Sigma}$.  
\end{itemize}

\end{proof}

 \begin{corollary} \label{cor:sigma_norm_rank}
With $\Sigma := \Sigma_{r}$ the set of matrices of rank lower than $r$ in $\sH$ the set of symmetric matrices in $\bR^{n \times n}$, we have:
 \begin{enumerate}
 \item for any matrices $V^T\diag(w)V,V^T\diag(w')V$ with $V \in O(n)$  such that $|w_{j}| \leq |w'_{j}|$ for all $j$ we have $\|V^T\diag(w)V\|_{\Sigma} \leq \|V^T\diag(w')V\|_{\Sigma}$;
   \item For any symmetric matrix  $z$, and $T_r$ a subset indexing $r$ components of largest magnitude of $\eig(z)$, \ie
   \[\min_{i \in T} |\eig(z)_{i}| \geq \max_{j \notin T} |\eig(z)_{j}|,\]
   with $|T| = r$. Then
 \begin{eqnarray}
 \max_{|T| \leq r} \|z_T\|_\Sigma &=& \|z_{T_r}\|_\Sigma\\
 \min_{|T| \leq r} \|z-z_T\|_\Sigma &=& \|z-z_{T_r}\|_\Sigma.
 \end{eqnarray}
 \end{enumerate}
 
 \end{corollary}
 \begin{proof}
 We show the two properties separately. 
 
 \begin{itemize}
 
 \item \textbf{Property 1:} Given the assumptions on $w,w'$, the linear operator defined by $Dz =V^TW Vz$ where $W$ is the diagonal matrix  such that $W_{ii} = w_i/w_i'$ if $w'_{i} \neq 0$ (and $W_{ii} = 0$ otherwise) satisfies $D \Sigma \subseteq \Sigma$ and $\|D\|_{op} \leq 1$. We have $D(V^T\diag(w')V)=V^TWw'V=V^TwV$. With \Cref{lem:sigma_norm_monotonic}, we get $\|V^T\diag(w)V\|_{\Sigma} = \|D(V^T\diag(w')V)\|_{\Sigma} \leq \|V^T\diag(w')V\|_{\Sigma}$.
 
\item \textbf{Property 2:}  This property is direct using the eigenvalue decomposition  \[z = U^T\diag(\eig(z))U^T = U^T\diag(\eig(z)_T + \eig(z)_{T^c})U^T \] and Property 1.

\end{itemize}

\end{proof} 
 We now prove \Cref{lem:assumption_suff}.
\begin{proof}[Proof of \Cref{lem:assumption_suff}]
 Consider first $\Sigma = \Sigma_k$. First, the properties of $\|\cdot\|_\Sigma$ established in \Cref{cor:sigma_norm} directly show that the minimum  of $\|x-z\|_\Sigma$  with respect to $x \in \Sigma$ is reached at any $x \in P_{\Sigma}(z)$. Then, we can write $\Sigma = \cup_{V \in \mathcal{V}} V$ where $V \in \mathcal{V}$ if, and only if there is an index set $T \subseteq \{ 1, \ldots,n\}$ such that $|T| \leq k $ and $V = \tspan(e_{i})_{i \in T}$.
 Given $V \in \mathcal{V}$ and $u \in \Sigma_k$, let us show that $P_{V^{\perp}} u \in \Sigma_{k}$.  Writing $V = \tspan (e_i)_{i \in T}$ where $|T| \leq k$, we have
$ P_V(u) = u_T$ and  $P_{V^\perp}(u) =   u_{T^c}$. As $\supp(u_{T^{c}}) \subseteq \supp(u)$ it follows that $\|u_{T^{c}}\|_{0} \leq k$, hence
$P_{V^\perp}(u) \in \Sigma_{k}$.

 In the case of low rank matrices $\Sigma = \Sigma_r$. We take $ \mathcal{V} = \{ \tspan (U_i)_{i \in I}, |I| \leq r, \|U_i\|_F=1, \rank(U_i) = 1, \ls U_i, U_j\rs =0 , i\neq j  \}$ . With \Cref{cor:sigma_norm_rank}, the minimum  of $\|x-z\|_\Sigma$  with respect to $x \in \Sigma$ is reached at any $x \in P_{\Sigma}(z)$.  Let $z\in \Sigma_r$ and $V \in \mathcal{V}$. 
 We have $P_V(z)=V_1^T S_1V_1$ has rank $r'$ lower than $r$. We can write $z= V_1^T S_1V_1 + V_2^T S_2V_2$ with $V_1V_2^T=0$.  Hence, $P_{V^\perp}(z)$ has rank at most $r-r' \leq r$ and $P_{V^\perp}(z)\in \Sigma_r$  otherwise $z$ would be of rank greater than $r$.

\end{proof}
 We need the following Lemma to control $\|\cdot\|_\Sigma$.
\begin{lemma}\label{lem:part_value_norm_sigma}

Let $\Sigma = \Sigma_k \subset \bR^n$. Then for any $v$
\begin{equation}
\|v\|_\Sigma^2 \geq \frac{\|v\|_1^2}{k}.
\end{equation}
Let $\Sigma = \Sigma_r$. Then for any $v$
\begin{equation}
\|v\|_\Sigma^2 \geq \frac{\|v\|_*^2}{r}.
\end{equation}
\end{lemma}
\begin{proof}

Case $\Sigma = \Sigma_k$ :
Let $\lambda_i\geq 0,u_i\in \Sigma$ such that $\|v\|_\Sigma^2 = \sum \lambda_i \|u_i\|_2^2$ and $v= \sum \lambda_i u_i$ from Fact~\ref{fact:atom2}.  We have, by convexity
 \begin{equation}
  \begin{split}
   \|v\|_1 &= \left\|\sum_i \lambda_i u_i \right\|_1 \leq  \sum_i\lambda_i \|u_i\|_1 . \\
 \end{split}
 \end{equation}
Using the fact that $\|x\|_1 \leq \sqrt{k}\|x\|_2 $ if $|\supp(x)|\leq k$ and the concavity of the square root,
 \begin{equation}
  \begin{split}      
       \|v\|_1 &\leq \sqrt{k}\sum_i\lambda_i \|u_i\|_2  
       \leq \sqrt{k}\sqrt{\sum_i\lambda_i \|u_i\|_2^2 }  = \sqrt{k}\|v\|_\Sigma .
  \end{split}
  \end{equation}

Case $\Sigma = \Sigma_r$ :
Let $\lambda_i\geq 0,u_i\in \Sigma$ such that $\|v\|_\Sigma^2 = \sum \lambda_i \|u_i\|_F^2$ and $v= \sum \lambda_i u_i$ from~Fact~\ref{fact:atom2}.
 We have, by convexity
 \begin{equation}
  \begin{split}
   \|v\|_1 &= \left\|\sum_i \lambda_i u_i \right\|_* 
   \leq  \sum_i\lambda_i \|u_i\|_* .
 \end{split}
 \end{equation}
Using the fact that $\|x\|_* \leq \sqrt{r}\|x\|_F $ if $\rank(x) \leq r$ and the concavity of the square root,
 \begin{equation}
  \begin{split}      
       \|v\|_* &\leq \sqrt{k}\sum_i\lambda_i \|u_i\|_F  
       \leq \sqrt{k}\sqrt{\sum_i\lambda_i \|u_i\|_F^2 }  = \sqrt{k}\|v\|_\Sigma.
  \end{split}
  \end{equation}

\end{proof}

\subsubsection{Sparsity}

We prove several intermediates lemmas to obtain $D_\Sigma(\|\cdot\|_1)$.
\begin{lemma}\label{lem:sup_DL_l1}
Consider $\Sigma = \Sigma_{k}$ the set of $k$-sparse vectors in $\sH = \bR^{n}$, and $0 \leq L \geq n-k$. We have
\begin{equation}
D_\Sigma^{k+L}(\|\cdot\|_1) :=  \sup_{ z \in \sT_{\|\cdot\|_1}(\Sigma) \setminus \{0\} : |\supp(z)| = k+L}   \frac{\|z_{T^c}\|_\Sigma^2}{\|z_{T}\|_2^2}   = \min \left(1,\frac{L}{k}\right) .\\
\end{equation}
\end{lemma}
\begin{proof}
It was already proven in \cite[Theorem 4.1]{Traonmilin_2016} that $\delta_\Sigma^{\mathtt{suff}}(\|\cdot\|_1)\geq \frac{1}{\sqrt{2}}$ hence by \Cref{lem:charac_suff_RIP}
\begin{equation}\label{eq:DL1}
 \sup_{ z \in \sT_{\|\cdot\|_1}(\Sigma) \setminus \{0\} } \frac{\|z_{T^c}\|_\Sigma^2}{\|z_{T}\|_2^2}  = D_\Sigma(\|\cdot\|_1) \leq 1. 
\end{equation}

Hence, $D_\Sigma^{k+L}(\|\cdot\|_1) \leq 1$

Consider $H_{0}$ of cardinality $k$, $H_{1}$ of cardinality $L$ such that $H_{0} \cap H_{1} = \emptyset$ (this is possible as $k+L \leq n$), and define $z = \alpha 1_{H_0}+1_{H_1}$ where $\alpha = \max(1,L/k)$. As $\alpha \geq 1$, a set of the $k$ largest components of $z$ is $T = H_0$. Moreover, $\|z_{H_{0}}\|_{1} = \alpha k = \max(k,L) \geq L = \|z_{H_{1}}\|_{1} = \|z_{H_{0}^{c}}\|_{1}$.

We distinguish two cases:
\begin{itemize}
 \item  \textbf{Case 1: $L \geq k$}, from  \Cref{lem:part_value_norm_sigma},
 $\|z_{T^c}\|_\Sigma^2 \geq     \frac{1}{k}\|z_{H_1}\|_1^2 = L^2/k$. Moreover, $\|z_{T}\|_2^2 =k\alpha^2 = L^2/k $, thus $\|z_{T^c}\|_\Sigma^2/\|z_{T}\|_2^2 \geq 1$. Combining with \eqref{eq:DL1} yields $D_\Sigma^{k+L}(\|\cdot\|_{1}) = 1  = \min(1,L/k)$.

  \item \textbf{Case 2: $ L<k$,} we have $z_{T^c} = z_{H_1} \in \Sigma_{k}$ hence $\|z_{T^c}\|_\Sigma^{2} = \|z_{T^c}\|_2^{2} = \|z_{H_{1}}\|_{2}^{2} = L$ and $\|z_{T^c}\|_{\Sigma}^{2}/\|z_{T}\|_{2}^{2} = L/k$. This shows that $D_{\Sigma}^{k+L}(\|\cdot\|_{1}) \geq L/k = \min(1,L/k)$. To conclude, we show that $D_\Sigma^{k+L}(\|\cdot\|_{1}) \leq L/k$. Consider {\em any} $z' \in \sT_{\|\cdot\|_1}(\Sigma)$ such that $|\supp(z')|=k+L$, with \Cref{lem:opt_support_wl1}, there is a support $H$ of size lower than $k$ such that, $\|z'_{H}\|_{1} \geq \|z'_{H^c}\|_{1}$, let $T$ a set of the $k$ largest components of $z'$. We have $\|z'_{T}\|_{1}-\|z'_{T^c}\|_{1} \geq  \|z'_{H}\|_{1} - \|z'_{H^c}\|_{1} $.
  As $\|z'\|_{0} \leq k+L$ and $L<k$, $z'_{T^{c}} \in \Sigma_{L} \subset \Sigma_{k}$ hence $\|z'_{T^c}\|_\Sigma = \|z'_{T^c}\|_2$. Moreover, $|z'_{i}| \geq \|z'_{T^{c}}\|_{\infty}$ for any $i \in T$,  hence $\|z'_{T}\|_{2}^{2} \geq k \|z'_{T^{c}}\|_{\infty}^{2}$. As a result
  \[
   \frac{\|z'_{T^c}\|_\Sigma^2}{\|z'_{T}\|_2^2}
   =
    \frac{\|z'_{T^c}\|_2^2}{\|z'_{T}\|_2^2}
    \leq
     \frac{L \|z'_{T^c}\|_\infty^2}{k\|z'_{T^{c}}\|_\infty^2} = L/k.
  \]

  \end{itemize}

\end{proof}

\begin{lemma}\label{lem:charact_supDL2_atom}
Let $\Sigma = \Sigma_k$ be the set of $k$-sparse vectors in $\bR^{n}$ with $k<n/2$
 and $1 \leq L \leq n-k$. Assume that $R$ is positively homogeneous, subadditive and nonzero.
 
 Consider 
\begin{eqnarray}
(H_0,v_0) &\in& \arg \max_{\stackrel{H \subseteq \{ 1, \ldots,n\}:\ |H| = k}{v \in Q_{H}}}
R(v)\\
(H_1,v_1) &\in& \arg \min_{\stackrel{H \subseteq \{ 1, \ldots,n\} \setminus{H_{0}}, |H|=L}{v \in Q_{H}}}
R(v).
\end{eqnarray}
 We have
\begin{equation}\label{eq:ineq_DL1_atom}
D_\Sigma^{k+L}(R) := \sup_{ z \in \sT_{R}(\Sigma)\setminus \{0\} : |\supp(z)| =k+L} \frac{\|z_{T^c}\|_\Sigma^2}{\|z_{T}\|_2^2}    
\geq \min \left(1, \frac{L}{k} \right).
\end{equation}

\end{lemma}

\begin{proof}
From \Cref{lem:charact_supBL2_atom}, $R^{\star}(v_{1})= \frac{L}{k}R^{\star}(v_{0})$.
Since $k+L \leq n$ there is indeed some $H$ of cardinality $L$ such that $H \cap H_{0} = \emptyset$, hence $H_{1}$ is well-defined. By construction, $H_{1} \cap H_{0} = \emptyset$. From \Cref{lem:charact_supBL2_atom}, we also have $R(v_0)>0$ and $R(v_1)/R(v_0)\leq L/k$.

Since $R(v_{0})>0$, $R$ is positively homogeneous and $\Sigma$ is homogeneous, by \Cref{lem:BuildDescentVector}, 
 $z = -\alpha  v_0 +  v_1 \in \sT_{R}(\Sigma)$ with $\alpha :=  \max(R(v_1)/R(v_0),1)$. 
Observe that $|\supp(z)| = |H_{0}|+|H_{1}| = k+L$. Since $\alpha \geq 1$ and all nonzero entries of $v_{0},v_{1}$ have magnitude one, a set of the $k$ largest components of $z$ is $T=  H_{0}$. We have 
\begin{equation}
 \frac{\|z_{T^c}\|_\Sigma^2}{\|z_{T}\|_2^2} = \frac{\|v_1\|_\Sigma^2}{k\alpha^2}.
\end{equation}
With \Cref{lem:part_value_norm_sigma}, $\|v_1\|_\Sigma^2\geq \frac{\|v_1\|_1^2}{k} \geq \frac{L^2}{k}$ if $L\geq k$ and $\|v_1\|_\Sigma^2 =\|v_1\|_2^2$ otherwise (Fact~\ref{fact:atom1}).
If $L\geq k$
\begin{equation}
  \frac{\|z_{T^c}\|_\Sigma^2}{\|z_{T}\|_2^2} \geq \frac{L^2}{k^2\alpha^2} \geq \frac{L^2}{k^2\max(L/k,1)^2} =1.
\end{equation}
If $L<k$, 
\begin{equation}
  \frac{\|z_{T^c}\|_\Sigma^2}{\|z_{T}\|_2^2}  = \frac{L}{k\alpha^2} \geq \frac{L}{k}
\end{equation}
which leads to the conclusion.

\end{proof}

\subsubsection{Low rank}

\begin{lemma}\label{lem:sup_DL_nuclear}
Consider $\Sigma = \Sigma_{r}$ the set of symmetric matrices of rank lower than $r$. For any $L \geq 0$ such that $r+L \leq n$ we have,
\begin{equation}
D_\Sigma^{r+L}(\|\cdot\|_*) :=  \sup_{ z \in \sT_{\|\cdot\|_*}(\Sigma) \setminus \{0\} : rank(z) = r+L}   \frac{\|z_{T^c}\|_\Sigma^2}{\|z_{T}\|_F^2}   = \min \left(1,\frac{L}{r}\right)\\
\end{equation}
where $z_T$ is $z$ restricted to its $r$ biggest eigenvalues, and $z_{T^c} = z-z_T$
\end{lemma}
\begin{proof}
It was already proven in \cite[Theorem 4.1]{Traonmilin_2016} that $\delta_\Sigma^{\mathtt{suff}}(\|\cdot\|_*)\geq \frac{1}{\sqrt{2}}$ hence by \Cref{lem:charac_suff_RIP}
\begin{equation}\label{eq:DL1_nuclear}
 \sup_{ z \in \sT_{\|\cdot\|_*}(\Sigma) \setminus \{0\} } \frac{\|z_{T^c}\|_\Sigma^2}{\|z_{T}\|_F^2}  = D_\Sigma(\|\cdot\|_*) \leq 1. 
\end{equation}

Consider $H_{0} =\{1,..r\}$ , $H_{1}=\{r+1,..,r+L\}$, let $U \in O(n)$ and define $z = U^T\diag(\alpha 1_{H_0}+1_{H_1})U$ where $\alpha = \max(1,L/r)$. As $\alpha \geq 1$, a set of the $r$ largest components of $\eig(z)$ is $T = H_0$. Moreover, $\|z_T\|_{*} = \alpha r = \max(r,L) \geq L = \|z-z_{T}\|_{*} = \|z_{T^{c}}\|_{*}$.

 If $L \geq r$, from  \Cref{lem:part_value_norm_sigma},
 $\|z_{T^c}\|_\Sigma^2 \geq     \frac{1}{r}(\|z_{T^c}\|_*)^2 = L^2/r$. Moreover, $\|z_{T}\|_F^2 =r\alpha^2 = L^2/r $, thus $\|z_{T^c}\|_\Sigma^2/\|z_{T}\|_F^2 \geq 1$. Combining with \eqref{eq:DL1_nuclear} yields $D_L(\|\cdot\|_*) = 1  = \min(1,L/r)$.

  If $ L<r$, we have $z_{T^c}  \in \Sigma_{r}$ hence $\|z_{T^c}\|_\Sigma^{2}  = L$ and $\|z_{T^c}\|_{\Sigma}^{2}/\|z_{T}\|_{2}^{2} = L/r$. This shows that $D_{L}(\|\cdot\|_*) \geq L/r = \min(1,L/r)$. To conclude, we show that $D_{L}(\|\cdot\|_*) \leq L/r$. Consider {\em any} $z' \in \sT_{\|\cdot\|_*}(\Sigma)$ such that $|\supp(z')|=r+L$, with \Cref{lem:opt_support_wnuclear}, there is a support $r'$ and $H = {1,..,r'}$ such that $\|z'_{H}\|_{*} \geq \|z'_{H^c}\|_{*}$, let $T$ a set of $r$ largest components of $z'$. We have $\|z'_{T}\|_{*}-\|z'_{T^c}\|_{*} \geq  \|z'_{H}\|_{*} - \|z'_{H^c}\|_{*} $.
  As $\|\eig(z')\|_{0} \leq r+L$ and $L<r$, $z'_{T^{c}} \in \Sigma_{L} \subset \Sigma_{r}$ hence $\|z'_{T^c}\|_\Sigma = \|z'_{T^c}\|_F$. Moreover, $|\eig(z')_{i}| \geq \|\eig(z'_{T^{c}})\|_{\infty}$ for any $i \in T$,  hence $\|z'_{T}\|_{F}^{2} \geq r \|\eig(z'_{T^{c}})\|_{\infty}^{2}$. As a result
  \[
   \frac{\|z'_{T^c}\|_\Sigma^2}{\|z'_{T}\|_F^2}
   =
    \frac{\|z'_{T^c}\|_F^2}{\|z'_{T}\|_F^2}
    \leq
     \frac{L \|\eig(z'_{T^c})\|_\infty^2}{r\|\eig(z'_{T^{c}})\|_\infty^2} = L/r.
  \]

\end{proof}

\begin{lemma}\label{lem:charact_supDL2_atom_nuclear}
Let $\Sigma = \Sigma_r$ be the set of $n \times n$ symmetric matrices  with rank at most $r$ with $r < n/2$, and $1 \leq L \leq n-r$. 
Assume $R$ is positively homogeneous, subadditive and nonzero. Consider the supports $H_0 =\{1,2,..,r\}$  and $H_1 = \{r+1,\ldots,r+L\}$. 
\begin{eqnarray}
(U_0,v_0) &\in& \arg \max_{U \in O(n), v \in Q_{H_{0}}} \|U^T\diag(v)U \|_\sA\\
(U_1,v_1) &\in& \arg \min_{U \in O(n),v \in Q_{H_{1}} :\; U_{0,1:r} U_{r+1:r+L}^T =0 } \|U^T\diag(v)U \|_\sA .
\end{eqnarray}

We have 
 \begin{equation}\label{eq:ineq_Dnuc_atom}
 \begin{split}
 D_\Sigma^{r+L}(R) := \sup_{ z \in \sT_{R}(\Sigma)\setminus \{0\} : |\supp(z)| =r+L} \frac{\|z_{T^c}\|_\Sigma^2}{\|z_{T}\|_F^2}  \geq  
 \min \left(1, \frac{L}{r} \right).
 \end{split}
 \end{equation}

\end{lemma}

\begin{proof}
From \Cref{lem:charact_supBL2_atom_rank}, $R^{\star}(U_1^T\diag(v_{1})U_1= \frac{L}{r}R^{\star}(U_0^T\diag(v_{0})U_0)$, $R(U_0^T\diag(v_{0})U_0)>0$ and \[R(U_1^T\diag(v_{1})U_1)/R(U_0^T\diag(v_{0})U_0)\leq L/r.\]

Since $R(v_{0})>0$, $R$ is positively homogeneous and $\Sigma$ is homogeneous, by \Cref{lem:BuildDescentVector}, 
 $z = -\alpha U_0^T\diag(v_{0})U_0 +  U_1^T\diag(v_{1})U_1\in \sT_{R}(\Sigma)$ with $\alpha :=  \max(R(U_1^T\diag(v_{1})U_1)/R(U_0^T\diag(v_{0})U_0),1)$. 
Observe that $|\supp(\eig(z))| = |H_{0}|+|H_{1}| = r+L$. Since $\alpha \geq 1$ and all nonzero entries of $v_{0},v_{1}$ have magnitude one, a set of the $r$ largest components of $z$ is $T=  H_{0}$. We have 
\begin{equation}
 \frac{\|z_{T^c}\|_\Sigma^2}{\|z_{T}\|_F^2} = \frac{\|U_1^T\diag(v_{1})U_1\|_\Sigma^2}{r\alpha^2}.
\end{equation}
With \Cref{lem:part_value_norm_sigma}, we have
\begin{equation}
  \begin{cases}
  \|U_1^T\diag(v_{1})U_1\|_\Sigma^2 \geq \frac{1}{r}\|U_1^T\diag(v_{1})U_1\|_*^2 = \frac{L^2}{r} & \text{if } L\geq r \\
  \|U_1^T\diag(v_{1})U_1\|_\Sigma^2 =\|U_1^T\diag(v_{1})U_1\|_F^2 & \text{otherwise (Fact~\ref{fact:atom1}).}
  \end{cases}
\end{equation}
If $L\geq r$
\begin{equation}
  \frac{\|z_{T^c}\|_\Sigma^2}{\|z_{T}\|_F^2} \geq \frac{L^2}{r^2\alpha^2} \geq \frac{L^2}{r^2\max(L/r,1)^2} =1.
\end{equation}
If $L<r$, 
\begin{equation}
  \frac{\|z_{T^c}\|_\Sigma^2}{\|z_{T}\|_F^2}  = \frac{L}{r\alpha^2} \geq \frac{L}{r}
\end{equation}
which leads to the conclusion.

\end{proof}

\subsection{Proofs for \texorpdfstring{\Cref{sec:sparsity_in_levels}}{Section \ref{sec:sparsity_in_levels}} }\label{sec:proofslevels}

We extend notations for classical sparsity to sparsity in levels ($\Sigma = \Sigma_{k_1,k_2}$). For $z= (z_1,z_2) \in \sH$, we
we define the following projections $P_1(z) :=z_1$ and $P_2(z):=z_2$ and denote $T=(S_1,S_2) =T(z)$ where for $i\in \{1,2\}$, $S_{i} \subseteq \{1,\ldots,n_i\}$ is a support containing $k_i$ largest coordinates (in absolute value) of $z_i$, i.e. $|S_{i}|=k_{i}$ and $\min_{j \in S_1} |z_{i,j}| \geq \max_{j\in S_i^c}|z_{i,j}|$. 
For every $U = (U_{1},U_{2})$ where $U_{i} \subseteq \{1,\ldots,n_{i}\}$ and $|U_{i}|=k_{i}$, 
 we also have $\|(z_{i})_{S_{i}}\|_{1} \geq \|(z_{i})_{U_{i}}\|_{1}$ hence $\|z_{T}\|_{w} \geq \|z_{U}\|_{w}$ and similarly $\|z_{T^{c}}\|_{w} \leq \|z_{U^{c}}\|_{w}$.

We define similarly  $T_2 =T_{2}(z) = (S'_{1},S'_{2})$ with $S'_{i}$ containing the $2k_{i}$ largest coordinates of $z_i$. 
We begin by simplifying  the condition $z \in \sT_{\|\cdot\|_{w}}(\Sigma)\setminus \{0\}$.

\begin{lemma} \label{lem:opt_support_wl1_level}
 Let $w=(w_1,w_2) \in \bRp^2$. Let $\|\cdot\|_w = w_1\|P_1(\cdot)\|_1 +w_2 \|P_2(\cdot)\|_1$ Let $z \in \sT_{\|\cdot\|_w}(\Sigma_{k_1,k_2})\setminus \{0\} $ then 
\begin{equation}
 \|z_{T^c}\|_w \leq \|z_{T}\|_w.
\end{equation}
 Reciprocally, 
 \begin{equation}
 \|z_{T^c}\|_w \leq \|z_{T}\|_w 
\end{equation}
implies $z \in \sT_{\|\cdot\|_w}(\Sigma_{k_1,k_2})$.
\end{lemma}
\begin{proof}
By definition, if $z \in \sT_{\|\cdot\|_w}(\Sigma_{k_1,k_2})\setminus \{0\} $ then there exists $x \in \Sigma_{k_1,k_2}$ and $\gamma \in \bR \setminus \{0\}  $ such that $z= \gamma y$ and  $\|x+ y\|_w \leq \|x\|_w$. With $U := \supp(x)$ we have $\|y_{U^{c}}\|_{w}+\|(x+y)_{U}\|_{w} = \|x+y\|_{w} \leq \|x\|_{w} = \|x_{U}\|_{w}$. By the triangle inequality this implies 
\begin{equation}
 \|y_{U^c}\|_w \leq \|x_U\|_w - \| (x+ y)_U\|_w \leq \| y_U\|_w.
\end{equation}
As $\gamma \neq 0$, we obtain  $ \| z_{U^c}\|_w \leq \| z_U\|_w$.\
We have 
\begin{equation}\label{eq:TmpPB0}
 \|z_{T}\|_w \geq \|z_{U}\|_w \geq  \|z_{U^c}\|_w \geq \|z_{T^c}\|_w.
\end{equation}

\end{proof}

To calculate $B_\Sigma(\|\cdot\|_w)$ (see definition in \Cref{cor:RCnecUoS}), we need a few lemmas.

\begin{lemma}\label{lem:intermediate_level_0} 
Consider $w_1,w_2,k_1,k_2 > 0$ and $\beta_1,\beta_2, \lambda  \geq 0$ and 
\begin{equation}
 V :=   \min_{\alpha_1,\alpha_2\geq 0} k_1 \alpha_1^2 +k_2\alpha_2^2 \;\; \;\; \mathrm{s.t.} \;\;\;\;\alpha_{1} \geq \beta_{1},\ \alpha_{2} \geq \beta_{2},\ k_1 w_1\alpha_1 +k_2w_2\alpha_2 \geq \lambda
\end{equation}
\begin{itemize}
\item If $\lambda < k_{1}w_{1}\beta_{1}+k_{2}w_{2}\beta_{2}$ then $V = k_{1}\beta_{1}^{2}+k_{2}\beta_{2}^{2}$. 
\item If $\lambda \geq k_{1}w_{1}\beta_{1}+k_{2}w_{2}\beta_{2}$  then the minimum is achieved at $\alpha_{1}^{*},\alpha_{2}^{*}$ such that $k_1 w_1\alpha_1^* +k_2w_2\alpha_2^* = \lambda$. Moreover,
\begin{itemize}
\item if $\lambda \geq (k_{1}w_{1}^{2}+k_{2}w_{2}^{2}) \max(\beta_{1}/w_{1},\beta_{2}/w_{2})$ then
\[V =  \min_{\alpha_1,\alpha_2\geq 0, w_1\alpha_1 +k_2w_2\alpha_2 = \lambda } k_1 \alpha_1^2 +k_2\alpha_2^2 =  \lambda^{2}/(k_{1}w_{1}^{2}+k_{2}w_{2}^{2});\]
\item otherwise 
\[
V = \min \left( k_1\beta_1^2 +\frac{(\lambda -k_1w_1\beta_1)^{2}}{k_2w_2^{2}}, k_2\beta_2^2 +\frac{(\lambda -k_2w_2\beta_2)^{2}}{k_1w_1^{2}} \right) > \lambda^{2}/(k_{1}w_{1}^{2}+k_{2}w_{2}^{2}).
\]
\end{itemize}
\end{itemize}
\end{lemma}
\begin{proof} Consider the change of variables $x = \sqrt{k_1} \alpha_1$, $y = \sqrt{k_{2}}\alpha_{2}$ and denote $x_{0} := \sqrt{k_{1}}\beta_{1}$, $y_{0} := \sqrt{k_{2}}\beta_{2}$, $a := \sqrt{k_{1}}w_{1}$, $b := \sqrt{k_{2}}w_{2}$. This leads to the equivalent problem
\[
  \min_{x,y \geq 0} x^{2}+y^{2} \;\; \mathrm{s.t.} \;\;x \geq x_{0}, y \geq y_{0}, ax+by \geq \lambda
\]
which involves a convex objective to be minimized over a polyhedral constraint set. If $ax_{0}+by_{0} > \lambda$, i.e., if $k_{1}w_{1}\beta_{1}+k_{2}w_{2}\beta_{2} > \lambda$, then this problem is equivalent to 
\[
\min_{x,y \geq 0} x^{2}+y^{2}  \;\; \mathrm{s.t.} \;\;x \geq x_{0}, y \geq y_{0}
\]
which is minimized at $(x_{0},y_{0})$, with value $x_{0}^{2}+y_{0}^{2} = k_{1}\beta_{1}^{2}+k_{2}\beta_{2}^{2}$. Otherwise, the candidate optima must satisfy the constraint $ax+by=\lambda$, hence $y=(\lambda-ax)/b$ and the problem is equivalent to 
\begin{equation}\label{eq:TmpIntermediate0-0}
\min_{x_{0} \leq x \leq (\lambda-by_{0})/a} x^{2}+(ax-\lambda)^{2}/b^{2}.
\end{equation}
The unconstrained minimum of~\eqref{eq:TmpIntermediate0-0} is at $x^{*}$ satisfying $2x^{*}+2a(ax^{*}-\lambda)/b^{2} = 0$, \ie, $x^{*} = \frac{a\lambda}{a^{2}+b^{2}}$, leading to $y^{*}=(\lambda-ax^{*})/b = \frac{b\lambda}{a^{2}+b^{2}}$ and to an optimal unconstrained problem value
\[
(x^{*})^{2}+(y^{*})^{2} = \lambda^{2}/(a^{2}+b^{2}) = \lambda^{2}/(k_{1}w_{1}^{2}+k_{2}w_{2}^{2}).
\]
This is also the value of the constrained minimum of~\eqref{eq:TmpIntermediate0-0}, provided that $x_{0} \leq x^{*} \leq (\lambda-by_{0})/a$, i.e., that $\lambda \geq (a^{2}+b^{2}) \max(x_{0}/a,y_{0}/b) = (k_{1}w_{1}^{2}+k_{2}w_{2}^{2}) \max(\beta_{1}/w_{1},\beta_{2}/w_{2})$.
Otherwise, the constrained minimum is either at $x = x_{0}$ and $y = (\lambda-ax_{0})/b$, so that $x^{2}+y^{2}=x_{0}^{2}+(\lambda-ax_{0})^{2}/b^{2}$; or at $y = y_{0}$ and $x = (\lambda-by_{0})/a$, so that $x^{2}+y^{2} = y_{0}^{2}+(\lambda-by_{0})^{2}/a^{2}$. The value at the optimum is then $\min(x_{0}^{2}+(\lambda-ax_{0})^{2}/b^{2},y_{0}^{2}+(\lambda-by_{0})^{2}/a^{2})$, which is necessarily larger than that of the unconstrained minimum.
Once translated in terms of the original variables, this yields the result.

\end{proof}

\begin{lemma}\label{lem:intermediate_level_2} Let $\rho \geq 0$, $k_1,k_2,L_1,L_2,w_1,w_2, \lambda  >0 $
 \begin{equation}\label{eq:SpLevCstraint}
  \max_{\beta_1\geq 0,\beta_2\geq 0} \frac{ L_1\beta_1^2 +L_2 \beta_2^2}{\rho +k_1 \beta_1^2+k_2\beta_2^2} \\ \; \mathrm{s.t.} \;\;w_1 (k_1+L_1)\beta_1  + w_2 (k_2 +L_2)\beta_2 = \lambda
\end{equation}
 is equal to 
 
 \begin{equation}
  \max_{i\in \{1,2\}} \frac{L_i\lambda^2} {\rho  w_i^2 (k_i +L_i)^2 +k_i \lambda ^2}.
 \end{equation}
Denoting $i^*$ the index maximizing this expression, the maximum is reached for $\beta_{i^*} = \frac{ \lambda}{ w_{i^{*}} (k_{i^{*}}+L_{i^{*}})}$ (and $\beta_j = 0$ for $j\neq i$).
\end{lemma}

\begin{proof}
Let $c\geq 0$. Observe that
 \begin{equation}\label{eq:SpLevCond1}
\frac{ L_1\beta_1^2 +L_2 \beta_2^2}{\rho +k_1 \beta_1^2+k_2\beta_2^2} \geq c
\end{equation}
is equivalent to 

 \begin{equation}\label{eq:SpLevCond2}
 (L_1-ck_1)\beta_1^2 +(L_2-ck_2 )\beta_2^2 \geq c \rho.
\end{equation}
With the change of variable $b_i =  w_i (k_i +L_i)\beta_i $ we have $b_1+b_2 = \lambda$ and \eqref{eq:SpLevCond2} reads
 \begin{equation}\label{eq:SpLevCond3}
 \frac{(L_1-ck_1)}{w_1^2 (k_1 +L_1)^2}b_1^2 + \frac{(L_2-ck_2)}{w_2^2 (k_2 +L_2)^2}(b_1-\lambda)^2 \geq c \rho .
\end{equation}
  The left side  is \emph{maximized} (with respect to $0 \leq b_{1} \leq \lambda$) for either $b_1=0$ or $b_1 =\lambda$. 
  The initial inequality~\eqref{eq:SpLevCond1} is thus feasible if, and only if, the maximum of the left-hand side of \eqref{eq:SpLevCond3} over these two values verifies the inequality
 \begin{equation}
 \max_{i \in \{1,2\}} \frac{(L_i-ck_i)}{w_i^2 (k_i +L_i)^2}\lambda^2 \geq c \rho  
\end{equation}
\ie if there is $i\in \{1,2\}$ such that $(L_i-ck_i)\lambda^2 \geq c \rho  w_i^2 (k_i +L_i)^2$. This is equivalent to
$L_i \lambda^2\geq c (\rho  w_i^2 (k_i +L_i)^2 +k_i \lambda ^2)$ and
 \begin{equation}
c \leq \frac{L_i\lambda^2} {\rho  w_i^2 (k_i +L_i)^2 +k_i \lambda ^2}.
\end{equation}

\end{proof}

\begin{lemma}\label{lem:max_theta}
Consider $w_1,w_2,\beta_1,\beta_2,c \geq 0$ and
  \begin{equation}
  V:=
  \sup_{0 \leq \theta_{i} \leq \beta_{i}, w_{1}\theta_{1}+w_{2}\theta_{2} \leq c} \theta_1^2 +\theta_2^2.
  \end{equation} 
Denoting $(\ell,r) \in \{(1,2),(2,1)\}$ such that $w_{\ell}\beta_{\ell} \leq w_{r}\beta_{r}$, we have
\begin{enumerate}
 \item\label{it:maxtheta_one} if $c < w_{\ell}\beta_{\ell}$ then $V = \max_{i \in \{1,2\}}(c/w_i)^2$;
  \item\label{it:maxtheta_two}  if $w_{\ell}\beta_{\ell} \leq c < w_{r}\beta_{r}$ then $V = \max((c/w_{r})^{2},\beta_{\ell}^{2}+[(c-w_{\ell}\beta_{\ell})/w_{r}]^{2}$;
  \item\label{it:maxtheta_three}  if $w_{r}\beta_{r} \leq c < w_{1}\beta_{1}+w_{2}\beta_{2}$ then $V = \max_{(i,j) \in \{(1,2),(2,1)\}} \beta_{i}^{2}+[(c-w_{i}\beta_{i})/w_{j}]^{2}$;
 \item\label{it:maxtheta_four}  if $c \geq w_{1}\beta_{1}+w_{2}\beta_{2}$ then $V = \beta_{1}^{2}+\beta_{2}^{2}$;
\end{enumerate}
\end{lemma}
\begin{proof}
The optimum $V$ is the maximization of a quadratic form within the intersection of a rectangle and a half-space delimited by an affine function. Using standard compactness arguments there exists at least a maximizer $(\theta_1^*,\theta_2^{*})$ of the considered expression.  If $\theta_{i}^{*}<\beta_{i}$ for some $i \in \{1,2\}$ then the constraint $c = w_{1}\theta_{1}^{*}+w_{2}\theta_{2}^{*}$ is satisfied (otherwise, we would have $0 \leq \theta_{i}^{*}<\beta_{i}$ and $w_{1}\theta_{1}+w_{2}\theta_{2} < c$, and we could exhibit other $\theta_i > \theta_{i}*^{}$ still satisfying the constraints and such that $\theta_1^2+\theta_2^2$ is increased), hence
$w_{1}\beta_{1}+w_{2}\beta_{2} > w_{1}\theta_{1}^{*}+w_{2}\theta_{2}^{*}=c$. \\
Vice-versa if $w_{1}\beta_{1}+w_{2}\beta_{2}>c$ then since $(\theta_{1}^{*},\theta_{2}^{*})$ satisfies all constraints we have $w_{1}\theta_{1}^{*}+w_{2}\theta_{2}^{*} \leq c < w_{1}\beta_{1}+w_{2}\beta_{2}$, hence there is at least one index $i \in \{1,2\}$ such that $\theta_{i}^{*}<\beta_{i}$. We can thus consider the following cases (depending on the shape of the domain): 
\begin{itemize}
 \item[$\bullet$]  if $w_{1}\beta_{1}+w_{2}\beta_{2}\leq c$ then for each $i\in \{1,2\}$,  $\theta_i^*=\beta_{i}$ hence $V = \beta_{1}^{2}+\beta_{2}^{2}$ as claimed;
 \item[$\bullet$] otherwise, i.e., if $w_{1}\beta_{1}+w_{2}\beta_{2} > c$, we have $w_{1}\theta_{1}^{*}+w_{2}\theta_{2}^{*} =  c$ and we distinguish three cases:
 \begin{itemize}
  \item[(a)] $\theta_{1}^{*} < \beta_{1}$, $\theta_{2}^{*} < \beta_{2}$: then, since $\theta_{2}^{*}=(c-w_{1}\theta_{1}^{*})/w_{2}$ where $\theta_{1}^{*}$ is a maximizer of $\theta_{1}^{2}+[(c-w_{1}\theta_{1})/w_{2}]^{2}$ under the constraint $0 \leq \theta_{1}$ and $c-w_{1}\theta_{1} \geq 0$, there is $(i,j) \in \{(1,2),(2,1) \}$ such that  $\theta_j^{*}=0$ and $\theta_i^{*}= c/w_i$. This is feasible provided that $c/w_{i} < \beta_{i}$.
 \item[(b)] $\theta_{1}^{*} = \beta_{1}$, $\theta_{2}^{*} < \beta_{2}$, hence $\theta_{2}^{*} = (c-w_{1}\beta_{1})/w_{2}$. This satisfies $0 \leq \theta_{2}^{*} < \beta_{2}$ if, and only if, $c \geq w_{1}\beta_{1}$.
 \item[(c)] $\theta_{1}^{*} < \beta_{1}$, $\theta_{2}^{*} = \beta_{2}$, hence $\theta_{1}^{*} = (c-w_{2}\beta_{2})/w_{1}$. This is feasible provided that $c \geq w_{2}\beta_{2}$.
\end{itemize}
We now discuss the possible cases depending on the value of $c$:
\begin{itemize}
\item $c < w_{\ell}\beta_{\ell}$: (a) with any $(i,j) \in \{(1,2),(2,1)\}$ is feasible; (b)-(c) are unfeasible, hence $V = \max_{i \in \{1,2\}} (c/w_{i})^{2}$. 
\item $c \geq w_{r}\beta_{r}$: (a) is unfeasible; (b)-(c) are both feasible, hence the claimed value of $V$ for this case. 
\item $w_{\ell}\beta_{\ell} \leq c < w_{r}\beta_{r}$: (a) is feasible with $(i,j)$ such that $c < w_{i}\beta_{i}$, \ie, with $(i,j) = (r,\ell)$, leading to a value $(\theta_{j}^{*})^{2}+(\theta_{i}^{*})^{2} = (c/w_{i})^{2} = (c/w_{r})^{2}$; (b) is feasible provided that $c \geq w_{1}\beta_{1}$, i.e., that $(r,\ell)=(2,1)$, leading to a value $(\theta_{1}^{*})^{2}+(\theta_{2}^{*})^{2} = \beta_{1}^{2}+[(c-w_{1}\beta_{1})/w_{2}]^{2}= \beta_{\ell}^{2}+[(c-w_{\ell}\beta_{\ell})/w_{r}]^{2}$; similarly, (c) is feasible provided that $(r,\ell)=(2,1)$, leading to a value $(\theta_{2}^{*})^{2}+(\theta_{1}^{*})^{2} = \beta_{2}^{2}+[(c-w_{2}\beta_{2})/w_{1}]^{2} = \beta_{\ell}^{2}+[(c-w_{\ell}\beta_{\ell})/w_{r}]^{2}$. Overall, this leads to $V = \max((c/w_{r})^{2},\beta_{\ell}^{2}+[(c-w_{\ell}\beta_{\ell})/w_{r}]^{2}$.
\end{itemize}

 \end{itemize}
\end{proof}

As in the case of the $\ell^{1}$ norm for sparsity and the nuclear norm for low-rank matrices, we compute $B_\Sigma(\|\cdot\|_w)$ (see definition in \Cref{cor:RCnecUoS}) via intermediate quantities $B^{L_1,L_2}(w)$ that we now introduce and control. We find an expression consistent with the $\ell^1$ case.

\begin{lemma} \label{lem:bound_gv}
Consider weights $w = (w_{1},w_{2})$ with $w_{i}>0$ and integers $k_{i} \geq 0$. Denote for any integers $L_{1},L_{2} \geq 0$
\begin{equation}\label{eq:DefBL1L2}
B^{L_1,L_2}(w) : =\sup_{\stackrel{\alpha_i \geq \beta_{i} \geq 0,\beta_{1}+\beta_{2}>0}{\sum_{i=1}^{2} (k_{i}w_{i}\alpha_{i}-w_i(k_i+L_i)\beta_{i})=0}}
\frac{\sum_{i=1}^{2}L_i \beta_i^2 }{\sum_{i=1}^{2}k_{i}(\alpha_{i}^{2}+\beta_{i}^{2})}.
\end{equation}
For $m \in \{ 1,2\}$, consider
\[
g_{m}(L_{1},L_{2},\alpha_{1},\alpha_{2},\beta_{1},\beta_{2}) := 
\frac{L_{1}\beta_{1}^{2}+L_{2}\beta_{2}^{2}+[(\sum_{i=1}^{2}(k_{i}w_{i}\alpha_{i}-(k_{i}+L_{i})w_{i}\beta_{i})/w_{m}]^{2}}
{\sum_{i=1}^{2} k_{i}(\alpha_{i}^{2}+\beta_{i}^{2})}.
\]
We have
 \begin{align}\label{eq:maximization_bound_gv}
 \sup_{  \stackrel{\alpha_i,\beta_i:0\leq \beta_{i} \leq \alpha_i; \beta_1+\beta_2>0}{
\sum_{i=1}^{2} (k_i+L_i)w_{i}\beta_{i} \leq \sum_{i=1}^{2} k_{i}w_{i}\alpha_{i}}}
g_m(L_1,L_2,\alpha_1,\alpha_2,\beta_1,\beta_2) \leq  B^{L_1,L_2}(w) .\\
\end{align}
\end{lemma}
\begin{proof}
 First we show that there exist $\alpha_i^* \in \bRp,\beta_i^* \in \bRp$ such that 
 \begin{align}
 g_m(L_1,L_2,\alpha_1^*,\alpha_2^*,\beta_1^*,\beta_2^*) 
 = 
 \sup_{  \stackrel{\alpha_i,\beta_i:0\leq \beta_{i} \leq \alpha_i; \beta_1+\beta_2>0}{
\sum_{i=1}^{2} (k_i+L_i)w_{i}\beta_{i} \leq \sum_{i=1}^{2} k_{i}w_{i}\alpha_{i}}}
g_m(L_1,L_2,\alpha_1,\alpha_2,\beta_1,\beta_2)
\end{align}
with $0\leq \beta_{i}^* \leq \alpha_i^*; \beta_1^*+\beta_2^* >0$, and $\sum_{i=1}^{2} (k_i+L_i)w_{i}\beta_{i}^* 
\leq \sum_{i=1}^{2} k_{i}w_{i}\alpha_{i}^*$.
Indeed, given any $\alpha_i,\beta_i$ satisfying these constraints, setting $\beta_j' = \beta_j /(\beta_1+\beta_2)$, $\alpha_j' = \alpha_j/(\beta_i+\beta_j)$, we have 
 \(
g_{m}(L_{1},L_{2},\alpha_{1}',\alpha_{2}',\beta_{1}',\beta_{2}') =g_{m}(L_{1},L_{2},\alpha_{1},\alpha_{2},\beta_{1},\beta_{2}) \)
hence the supremum is unchanged if we impose $\beta_1'+\beta_2' =1$ instead of $\beta_1+\beta_2>0$. Given any such pair $\beta'_{1},\beta'_{2}$, \Cref{lem:intermediate_level_0} yields the optimum over $\alpha_{i}$ satisfying the constraints, and as the resulting expression is continuous with respect to $\beta'_{j}$, the existence of a maximizer follows using a compactness argument.

We will soon prove that $\sum_i (k_i+L_i)w_{i}\beta_i^* = \sum_i k_iw_i\alpha_i^*$. If this equality is verified, since $0 \leq \beta_{i}^{*} \leq \alpha_{i}^{*}$, we obtain the desired result
  \begin{align}
   g_m(L_1,L_2,\alpha_1^*,\alpha_2^*,\beta_1^*,\beta_2^*) 
 &=\frac{
\sum_{i=1}^2 (L_{i}\beta_{i}^*)^{2}
}{\sum_{i=1}^{2}k_{i}((\alpha_{i}^*)^{2}+(\beta_{i}^*)^{2})}
\notag\\
 &\leq  \sup_{\stackrel{\alpha_i, \beta_i : 0\leq \beta_{i} \leq \alpha_i; \beta_1+\beta_2>0}{
 \sum_{i=1}^{2} (k_i+L_i)w_{i}\beta_{i} = \sum_{i=1}^{2} k_{i}w_{i}\alpha_{i}}}\frac{ \sum_{i=1}^2 L_{i}\beta_{i}^{2}}{\sum_{i=1}^{2}k_{i}((\alpha_{i})^{2}+(\beta_{i})^{2})}=B^{L_1,L_2}(w).
\end{align}

For the sake of contradiction, assume that $\sum_i (k_i+L_i)w_{i}\beta_i^*<\sum_i k_iw_i\alpha_i^*$, then with the shorthand $C := g_m(L_1,L_2,\alpha_1^*,\alpha_2^*,\beta_1^*,\beta_2^*)$, we have
\begin{equation}
[( \sum_i k_iw_i\alpha_i^*- \sum_i (k_i+L_i)w_{i}\beta_i^*)/w_m]^2+\sum_i (L_i -Ck_i) (\beta_i^*)^2   =  C\sum_i k_i(\alpha_i^*)^2 .
 \end{equation}
Since $g_{m}(L_{1},L_{2},\alpha_{1},\alpha_{2},\beta_{1},\beta_{2}) \leq C$ within the constraints  of~\eqref{eq:maximization_bound_gv}, 
  $(\beta_1^*,\beta_{2}^{*})$  maximize
 \[
 h(\beta_1,\beta_2) :=  [( \sum_i k_iw_i\alpha_i^*- \sum_i (k_i+L_i)w_{i}\beta_i)/w_v]^2+\sum_i (L_i -Ck_i) (\beta_i)^2
 \]
 among all $\beta_{1},\beta_{2}$ such that $0 \leq \beta_{i} \leq \alpha_{i}^{*}$, $\beta_{1}+\beta_{2}>0$ and $\sum_{i=1}^{2}(k_{i}+L_{i})w_{i}\beta_{i} \leq \sum_{i=1}^{2} k_{i}w_{i}\alpha_{i}^{*}$.

 Consider $j \in \{1,2\}$. 
  
   If $C> L_j/k_j$, then $h$ is decreasing with respect to $\beta_j$ on the considered range, hence $\beta_j^*=0$.   
Otherwise $C\leq L_j/k_j$, and since $h$ is a second degree polynomial in $\beta_{j}$ with positive leading coefficient, its maximum is at one of the extremities of the optimization interval, i.e.,  since we assumed $\sum_i (k_i+L_i)w_{i}\beta_i^*<\sum_i k_iw_i\alpha_i^*$, at least one of the constraints $ \beta_j^* =0$, $\beta_j^*=\alpha_j^*$ is reached. 

Since the optimum satisfies all constraints of~\eqref{eq:maximization_bound_gv}, we have 
$\beta_{1}^{*}+\beta_{2}^{*}>0$, hence in light of the above observations there is at least one index $j \in \{1,2\}$ such that $C \leq L_{j}/k_{j}$, and for which we have $\beta^{*}_{j} =\alpha_{j}^{*}>0$.

Since $\sum_{i=1}^{2}k_{i}w_{i}\beta_{i}^{*} \leq 
\sum_{i=1}^{2}(k_{i}+L_{i})w_{i}\beta_{i}^{*} < \sum_{i=1}^{2}k_{i}w_{i}\alpha_{i}^{*}$, both constraints $\beta_1^*=\alpha_1^*,\beta_{2}^{*}=\alpha_{2}^{*}$ cannot be reached at the same time hence there is $(i,j) \in \{(1,2),(2,1)\}$ such that  $\beta_i^*=0$, $\beta_j^*=\alpha_j^*$ and  
 \begin{align}
C =    g_m(L_1,L_2,\alpha_1^*,\alpha_2^*,\beta_1^*,\beta_2^*) 
&=\frac{L_{j}(\beta_{j}^*)^{2}+[(k_{i}w_{i}\alpha_{i}^{*}+k_{j}w_{j}\alpha_{j}^{*}-(k_{j}+L_{j})w_{j}\beta_{j}^{*})/w_{m}]^{2}
}
{k_{i}(\alpha_{i}^*)^{2}+k_{j}(\alpha_{j}^*)^{2}+ k_j(\beta_{j}^*)^{2}}\\
&=\frac{L_{j}(\alpha_{j}^*)^{2}+[(k_{i}w_{i}\alpha_{i}^{*}-L_{j}
w_{j}\alpha_{j}^{*})/w_{m}]^{2}
}
{k_{i}(\alpha_{i}^*)^{2}+2k_{j}(\alpha_{j}^*)^{2}}.
\end{align}
This can be rewritten $(L_{j}-2Ck_j)(\alpha_{j}^*)^{2}+[(k_{i}w_{i}\alpha_{i}^{*}-L_{j}
w_{j}\alpha_{j}^{*})/w_{m}]^{2}  = C k_i(\alpha_i^*)^2$. Observe that any $\alpha_{1},\alpha_{2},\beta_{1},\beta_{2}$ such that $\beta_{i}=0$, $\beta_{j} = \alpha_{j}>0$, $\alpha_{i}=\alpha_{i}^{*}$, and $L_{j}w_{j}\alpha_{j} \leq k_{i}w_{i}\alpha_{i}^{*}$ satisfy the constraints  of~\eqref{eq:maximization_bound_gv}, hence
$g_{m}(L_{1},L_{2},\alpha_{1},\alpha_{2},\beta_{1},\beta_{2}) \leq C$, or equivalently
\begin{equation}
\label{eq:TmpXXX}
(L_{j}-2Ck_j)(\alpha_j)^{2}+[(k_{i}w_{i}\alpha_{i}^{*}-L_{j}
w_{j}\alpha_{j})/w_{m}]^{2}  \leq C k_i(\alpha_i^*)^2 
\end{equation}
Thus, $\alpha_{j}^{*}$ maximizes the left hand side of~\eqref{eq:TmpXXX} under the constraint $0 \leq L_{j}w_{j}\alpha_{j} \leq k_{i}w_{i}\alpha_{i}^{*}$.
If $L-2Ck_{j} \leq 0$, then the left hand side of~\eqref{eq:TmpXXX} is decreasing with respect to $\alpha_j$ in the considered range, 
hence $\alpha_j^*=0$, which is not possible since $0< \beta_{1}+\beta_{2} = \beta_{j}^{*} = \alpha_{j}^{*}$.
Therefore we must have $L_{j}-2Ck_j > 0$, hence the left hand side of~\eqref{eq:TmpXXX} is a second degree polynomial in $\alpha_{j}$ with positive leading coefficient. Its maximum is achieved at one extremity of the interval constraint : the case $\alpha_j^*=0$ was already ruled out as impossible, hence $ L_{j} w_{j}\alpha_{j}^{*}= k_{i}w_{i}\alpha_{i}^{*}$. This implies $(k_{i}+L_{i}) w_{i}\beta_{i}^{*}+(k_{j}+L_{j}) w_{j}\beta_{j}^{*} = 
(k_{j}+L_{j}) w_{j}\alpha_{j}^{*} = k_{j}w_{j}\alpha_{j}^{*}+k_{i}w_{i}\alpha_{i}^{*}$, which yields the desired contradiction to the assumption that  $\sum_i (k_i+L_i)w_{i}\beta_i^*<\sum_i k_iw_i\alpha_i^*$.

\end{proof}

\begin{lemma} \label{lem:charact_supBL2_wl1_level-1}
Consider weights $w=(w_1,w_2)$ and integers $k_{i},n_{i}$ such that 
$1 \leq 2k_{i} < n_{i}$ and $\Sigma = \Sigma_{k_1,k_2} \subset \bR^{n_1} \times \bR^{n_2}$, $i\in \{1,2\}$. We have
\begin{equation}\label{eq:charact_supBL2_wl1_level-1}
B_\Sigma(\|\cdot\|_w)
=  \max_{
0\leq L_i \leq n-2k_i} B^{L_1,L_2}(w) 
\end{equation}
where $B^{L_{1},L_{2}}(w)$ is defined in \eqref{eq:DefBL1L2}.
\end{lemma}

\begin{proof}
We use the same proof method as in \Cref{lem:charact_supBL2_l1}. With the notations  $T=T(z),T_{2}=T_{2}(z)$ from the beginning of \Cref{sec:proofslevels}, denote  $T' = T_2 \setminus T$ so that $\|z_{T_{2}^{c}}\|_{w}+\|z_{T'}\|_{w} = \|z_{T^{c}}\|_{w}$. By \Cref{lem:opt_support_wl1_level}, we have 
\begin{equation}\label{eq:TmpOptLevels3}
 \begin{split}
B_\Sigma(\|\cdot\|_w) &= \sup_{z: z \neq 0, \|z_{T_2^c}\|_w + \|z_{T'}\|_w \leq \|z_{T}\|_w} \frac{\|z_{T_2^c}\|_2^2}{\|z_{T_2}\|_2^2}.
\end{split}
\end{equation}
We now show that this expression can be simplified by maximizing over vectors $z$ with a particular shape. 
Consider $z$ a vector satisfying the constraint in~\eqref{eq:TmpOptLevels3}. Replacing each entry $z_{i}$ of $z$ with its magnitude $|z_{i}|$ leaves the constraint (as well as the maximized quantity) unchanged, hence without loss of generality we can assume that $z$ has nonnegative entries $z_{i} \geq 0$. Similarly, we can assume without loss of generality that for each $i\in\{1,2\}$, the index set $S_{i} = \llbracket 1,k_{i}\rrbracket$ indexes $k_{i}$ largest entries of $P_{i}(z)$ and $S'_{i} = \llbracket 1,2k_{i}\rrbracket$ indexes the $2k_{i}$ largest entries.

Given some $j \in \{1,2\}$, consider two (equal or distinct) indices in $S_{j}$ and the vector $\tilde{z}$ obtained by keeping unchanged all entries of $z$, except those indexed by these indices which are replaced by their average. This has the following effect:
\begin{enumerate}
  \item Each $S_{i}$ (resp. $S'_{i}$), $i \in \{1,2\}$, is a set of the $k_{i}$ (resp. $2k_{i}$) largest coordinates of $P_{i}(\tilde{z})$, hence $T(\tilde{z}) = T = (S_{1},S_{2})$, $T_{2}(\tilde{z}) = T_{2} = (S'_{1},S'_{2})$, $T'(\tilde{z}) = T' = T_{2}\backslash T$, $\tilde{z}_{T_{2}^{c}}=z_{T_{2}^{c}}$, $\tilde{z}_{T'}=z_{T'}$, and the support of $P_{i}(\tilde{z})$, $i \in \{1,2\}$ is the same as that of $P_{i}(z)$.
  \item Denoting $a,b \geq 0$ the values of the two considered entries, since $(a+b)/2+(a+b)/2 = a+b$, we have $\|[P_{j}(\tilde{z})]_{S_{j}}\|_{1}= \|[P_{j}(z)]_{S_{j}}\|_{1}$, and we obtain that  $\|\tilde{z}_{T}\|_{w} = \|z_{T}\|_{w}$, hence $\tilde{z}$ still satisfies the optimization constraint;
  \item As $\|\tilde{z}_{T_{2}^{c}}\|_{2}=\|z_{T_{2}^{c}}\|_{2}$ and $\|\tilde{z}_{T_{2}}\|_{2}^{2} - \|z_{T_{2}}\|_{2}^{2} = 2[(a+b)/2]^{2}-a^{2}-b^{2}=-(a-b)^{2}/2 \leq 0$, hence $\|\tilde{z}_{T_{2}^{c}}\|_{2}^{2}/\|\tilde{z}_{T_{2}}\|_{2}^{2} \geq
\|\tilde{z}_{T_{2}^{c}}\|_{2}^{2}/\|\tilde{z}_{T_{2}}\|_{2}^{2}$ where the inequality is strict as soon as $a \neq b$.
\end{enumerate}
All the above imply that, without loss of generality, we can restrict the optimization to vectors $z$ such that, for $i \in \{1,2\}$, all entries of $z_{S_{i}}$  are equal. We denote $\alpha_{i} > 0$ their common value. A similar reasoning with $S'_{j}\backslash S_{j}$ instead of $S_{j}$ shows that we can also assume without loss of generality that all entries of $z_{S'_{i} \backslash S_{i}}$, $i \in \{1,2\}$, are equal. We denote $\beta_{i} \geq 0$ their common value.

The value of the smallest component of $[P_{i}(z)]_{S_{i}}$ is $\alpha_{i}$, 
while the smallest component of $[P_{i}(z)]_{S'_{i}}$ is $\min(\alpha_{i},\beta_{i})$. Denoting $x_{i} = P_{i}(z)_{(S'_{i})^{c}}$, we have $x_{i} \in \mathbb{R}_{+}^{n_{i}-2k_{i}}$ and the largest component of $[P_{i}(z)]_{(S'_{i})^{c}}$ is $\|x_{i}\|_{\infty}$. Hence, $S_{i}$ and $S'_{i}$ are respectively a set of the $k_{i}$ and $2k_{i}$ largest components of $P_{i}(z)$ if, and only if, $\|x_{i}\|_{\infty} \leq \beta_{i} \leq \alpha_{i}$. 

Finally, we observe that
$\|z_{T}\|_{w}-\|z_{T'}\|_{w}-\|z_{T_{2}^{c}}\|_{w}= w_{1}k_{1}\alpha_{1}+w_{2}k_{2}\alpha_{2}-w_{1}k_{1}\beta_{1}-w_{2}k_{2}\beta_{2}-w_{1}\|x_{1}\|_{1}-w_{2}\|x_{2}\|_{1}$, $\|z_{T_{2}^{c}}\|_{2}^{2} = \|x_{1}\|_{2}^{2}+\|x_{2}\|_{2}^{2}$ and $\|z_{T_{2}}\|_{2}^{2} = k_{1}\alpha_{1}^{2}+k_{2}\alpha_{2}^{2}+k_{1}\beta_{1}^{2}+k_{2}\beta_{2}^{2}$. 
This establishes
\begin{equation}\label{eq:subBL2_levels}
B_\Sigma(\|\cdot\|_w)
=
\sup_{\stackrel{\beta_i:\beta_{i} \geq 0}{\beta_{1}+\beta_{2}>0}}
\sup_{\alpha_i:\alpha_{i} \geq \beta_{i}}
\sup_{\stackrel{
\|x_{i}\|_{\infty} \leq \beta_{i}}{
\sum_{i=1}^{2}w_{i}\|x_{i}\|_{1} \leq \sum_{i=1}^{2} k_{i}w_{i}(\alpha_{i}-\beta_{i})}}
\frac{\sum_{i=1}^{2}\|x_{i}\|_{2}^{2}}{\sum_{i=1}^{2}k_{i}(\alpha_{i}^{2}+\beta_{i}^{2})} ,
\end{equation}
where the restriction $\beta_{1}+\beta_{2}>0$ simply follows from the fact that when $\beta_{1}+\beta_{2}=0$ we have $x_{1}=x_{2}=0$ which leads to a sub-optimal objective value. To show that the supremum in~\eqref{eq:subBL2_levels} is achieved, observe that both the constraints on $y:= (\alpha_{1},\alpha_{2},\beta_{1},\beta_{2},x_{1},x_{2})$ and the quantity $f(y)$ that is maximized are invariant by multiplication by a positive constant factor. Hence, the supremum is unchanged if we add a scaling constraint. e.g. by fixing $\|y\|_{\infty}$. This leads to the supremum of a continuous function over a compact set (the unit $\ell_{\infty}$ ball), hence there exists $\alpha_{i}^{*},\beta_{i}^{*},x_{i}^{*}$ reaching the supremum in~\eqref{eq:subBL2_levels}.

Thanks to \Cref{lem:FlatVectors}, given the constraints (depending on $\alpha_i$ and $\beta_i$), the  maximization w.r.t $x_i$ is reached with vectors with the shape
\[
  (\underbrace{\beta_i,\ldots,\beta_i}_{L_i},\theta_i,\underbrace{0,\ldots,0}_{n_{i}-2k_{i}-(L_i+1) \geq 0})
\]
with  $0\leq \theta_i \leq \beta_i$, $0\leq L_{i} \leq n_{i}-2k_{i}-1$, including potentially $L_i=0$ (case of vector $x_{i}$ with a single nonzero coordinate $\theta_i$).
We deduce
\begin{align}
B_\Sigma(\|\cdot\|_w)
&=
\sup_{\stackrel{\beta_i:\beta_{i} \geq 0}{\beta_{1}+\beta_{2}>0}}
\sup_{\alpha_i : \alpha_{i} \geq \beta_{i}}
\sup_{\stackrel{L_i,\theta_i:
0\leq L_i \leq n-2k_i-1, 0\leq \theta_i \leq \beta_{i}}{ 
\sum_{i=1}^{2}w_{i}\theta_i \leq \sum_{i=1}^{2} (k_{i}w_{i}\alpha_{i}-w_i(k_i+L_i)\beta_{i})}}
\frac{\sum_{i=1}^{2}L_i \beta_i^2 +\theta_i^2}{\sum_{i=1}^{2}k_{i}(\alpha_{i}^{2}+\beta_{i}^{2})} .
\end{align}
Hence, denoting
\begin{align}\label{eq:DefFL1L2}
f(L_{1},L_{2},\alpha_{1},\alpha_{2},\beta_{1},\beta_{2}) 
& :=
\sup_{\stackrel{\theta_i: 0\leq \theta_i \leq \beta_{i}}{ 
\sum_{i=1}^{2}w_{i}\theta_i \leq \sum_{i=1}^{2} (k_{i}w_{i}\alpha_{i}-w_i(k_i+L_i)\beta_{i})}}
\frac{\sum_{i=1}^{2}L_i \beta_i^2 +\theta_i^2}{\sum_{i=1}^{2}k_{i}(\alpha_{i}^{2}+\beta_{i}^{2})}, \\
\intertext{for parameters $\alpha_{i},\beta_{i},L_{i}$ such that \(
c:= \sum_{i=1}^{2} (k_{i}w_{i}\alpha_{i}-w_i(k_i+L_i)\beta_{i})\geq 0,
\) we have}\label{eq:DefFL1L20}
B_\Sigma(\|\cdot\|_w)
&=  \max_{0 \leq L_{i} \leq n_{i}-2k_{i} -1} \underbrace{
\sup_{\stackrel{\beta_{1},\beta_{2} \geq 0}{ \beta_{1}+\beta_{2}>0}}
\sup_{\stackrel{\alpha_i: \alpha_{i} \geq \beta_{i}}{\sum_{i=1}^{2} k_{i}w_{i}\alpha_{i} \geq \sum_{i=1}^{2}(k_i+L_i)w_{i}\beta_{i}}}
f(L_{1},L_{2},\alpha_{1},\alpha_{2},\beta_{1},\beta_{2})}_{f(L_{1},L_{2})} .
\end{align}
To continue, we bound $f(L_{1},L_{2})$ via characterizations of 
$f(L_{1},L_{2},\alpha_{1},\alpha_{2},\beta_{1},\beta_{2})$ in different parameter ranges. The supremum in~\eqref{eq:DefFL1L2} is covered by \Cref{lem:max_theta} hence we need to primarily distinguish cases depending on 
relative order of $c= \sum_{i=1}^{2} (k_{i}w_{i}\alpha_{i}-w_i(k_i+L_i)\beta_{i})\geq 0$, $w_{1}\beta_{1}+w_{2}\beta_{2}$,
$w_{1}\beta_{1}$, and $w_{2}\beta_{2}$. 
This suggests writing $f(L_{1},L_{2}) = \max_{u \in \{0,1\}} f_{u}(L_{1},L_{2})$ where
\begin{align}
f_{0}(L_{1},L_{2}) 
& := \sup_{  \stackrel{\beta_i,\alpha_i : 0\leq \beta_{i} \leq \alpha_i, \beta_1+\beta_2 >0}{\sum_{i=1}^{2} k_{i}w_{i}\alpha_{i} \geq \sum_{i=1}^{2} (k_i+L_i+1)w_{i}\beta_{i}}} 
f(L_{1},L_{2},\alpha_{1},\alpha_{2},\beta_{1},\beta_{2})\\
 f_1(L_{1},L_{2})  
& := \sup_{  \stackrel{\beta_i,\alpha_i :  0\leq \beta_{i} \leq \alpha_i, \beta_1+\beta_2 >0}{
\sum_{i=1}^{2} (k_i+L_i)w_{i}\beta_{i} \leq 
\sum_{i=1}^{2} k_{i}w_{i}\alpha_{i} < \sum_{i=1}^{2} (k_i+L_i+1)w_{i}\beta_{i} }} 
f(L_{1},L_{2},\alpha_{1},\alpha_{2},\beta_{1},\beta_{2}).
\end{align}
To express $f_{0}(L_{1},L_{2})$ and bound $f_{1}(L_{1},L_{2})$, we 
use the functions $g_{m}$, $m \in \{1,2\}$, from \Cref{lem:bound_gv}.\\
{\bf Expressing and bounding $f_{0}$:} if $\sum_{i=1}^{2} k_{i}w_{i}\alpha_{i} \geq \sum_{i=1}^{2} (k_i+L_i+1)w_{i}\beta_{i}$ then
$c \geq w_{1}\beta_{1}+w_{2}\beta_{2}$ hence \Cref{lem:max_theta}, case \ref{it:maxtheta_four} yields
\begin{align}\label{eq:subBL2_levels_bound1}
f(L_{1},L_{2},\alpha_{1},\alpha_{2},\beta_{1},\beta_{2}) 
& 
=
\frac{\sum_{i=1}^{2}(L_i+1) \beta_i^2}{\sum_{i=1}^{2}k_{i}(\alpha_{i}^{2}+\beta_{i}^{2})}\\
f_{0}(L_{1},L_{2}) &= 
\sup_{  \stackrel{0\leq \beta_{i} \leq \alpha_i, \beta_1+\beta_2 >0}{\sum_{i=1}^{2} k_{i}w_{i}\alpha_{i} \geq \sum_{i=1}^{2} (k_i+L_i+1)w_{i}\beta_{i}}} \frac{\sum_{i=1}^{2}(L_i+1) \beta_i^2}{\sum_{i=1}^{2}k_{i}(\alpha_{i}^{2}+\beta_{i}^{2})}\\
&\stackrel{\text{\Cref{lem:intermediate_level_0}}}{=} 
\sup_{  \stackrel{0\leq \beta_{i} \leq \alpha_i, \beta_1+\beta_2 >0}{\sum_{i=1}^{2} k_{i}w_{i}\alpha_{i} = \sum_{i=1}^{2} (k_i+L_i+1)w_{i}\beta_{i}}} \frac{\sum_{i=1}^{2}(L_i+1) \beta_i^2}{\sum_{i=1}^{2}k_{i}(\alpha_{i}^{2}+\beta_{i}^{2})}
= B^{L_1+1,L_{2}+1}(w) .
\end{align}
As a result
\begin{align}\label{eq:subBL2_levels_bound00}
f_{0}(L_{1},L_{2}) \leq \max_{0 \leq L_{i} \leq n_{i}-2k_{i}-1} B^{L_{1}+1,L_{2}+1}(w) \leq
\max_{0 \leq L'_{i} \leq n_{i}-2k_{i}} B^{L_{1}',L_{2}'}(w) 
\end{align}
{\bf Bounding $f_{1}$:}
we denote $(\ell,r) \in \{(1,2),(2,1)\}$ a pair such that $w_{\ell}\beta_{\ell} = \min_{i}w_{i}\beta_{i} \leq \max_{i}w_{i}\beta_{i} = w_{r}\beta_{r}$.
 
When $\sum_{i=1}^{2} (k_i+L_i)w_{i}\beta_{i} \leq \sum_{i=1}^{2} k_{i}w_{i}\alpha_{i} < \sum_{i=1}^{2} (k_i+L_i+1)w_{i}\beta_{i}$ we can distinguish three cases.
\begin{enumerate}
\item if $(k_\ell+L_\ell)w_{\ell}\beta_{\ell}+(k_r+L_r+1)w_{r}\beta_{r} \leq \sum_{i=1}^{2} k_{i}w_{i}\alpha_{i} < \sum_{i=1}^{2} (k_i+L_i+1)w_{i}\beta_{i}$  then 
$\max(w_{1}\beta_{1},w_{2}\beta_{2}) = w_{r}\beta_{r} \leq c <  w_{1}\beta_{1}+w_{2}\beta_{2}$ hence \Cref{lem:max_theta}, case~\ref{it:maxtheta_three} yields
\begin{align}\label{eq:subBL2_levels_bound2}
f(L_{1},L_{2},\alpha_{1},\alpha_{2},\beta_{1},\beta_{2}) 
&=
\max_{(u,v) \in \{(1,2),(2,1)\}}
\underbrace{\frac{
(L_{u}+1)\beta_{u}^{2}+L_{v}\beta_{v}^{2}+[(c- w_u\beta_{u})/w_{v}]^{2}
}{\sum_{i=1}^{2}k_{i}(\alpha_{i}^{2}+\beta_{i}^{2})}
}_{=g_{v}(L'_1,L'_2,\alpha_{1},\alpha_{2},\beta_{1},\beta_{2}),\ L'_{u}=L_{u}+1,L'_{v}=L_{v}}.
\end{align}

\item  if $ (k_\ell+L_\ell+1)w_{\ell}\beta_{\ell}+(k_r+L_r)w_{r}\beta_{r} \leq \sum_{i=1}^{2} k_{i}w_{i}\alpha_{i} < (k_\ell+L_\ell)w_{\ell}\beta_{\ell}+(k_r+L_r+1)w_{r}\beta_{r} $  then $\min(w_{1}\beta_{1},w_{2}\beta_{2}) = w_{\ell}\beta_{\ell} \leq c < w_{r}\beta_{r}  = \max(w_{1}\beta_{1},w_{2}\beta_{2}) $ hence \Cref{lem:max_theta}, case~\ref{it:maxtheta_two} yields
\begin{align}\label{eq:subBL2_levels_bound4}
f(L_{1},L_{2},\alpha_{1},\alpha_{2},\beta_{1},\beta_{2}) 
&=
\max\left(
\underbrace{\frac{
L_1\beta_1^{2}+L_{2}\beta_{2}^{2}+(c/w_{r})^{2}
}{\sum_{i=1}^{2}k_{i}(\alpha_{i}^{2}+\beta_{i}^{2})}}_{g_{r}(L_{1},L_{2},\alpha_{1},\alpha_{2},\beta_{1},\beta_{2})},
\underbrace{\frac{
(L_{\ell}+1)\beta_{\ell}^{2}+L_{r}\beta_{r}^{2}+[(c- w_\ell\beta_{\ell})/w_{r}]^{2}
}{\sum_{i=1}^{2}k_{i}(\alpha_{i}^{2}+\beta_{i}^{2})}}_{g_{r}(L'_{1},L'_{2},\alpha_{1},\alpha_{2},\beta_{1},\beta_{2}),\ L'_{\ell}=L_{\ell}+1,L'_{r}=L_{r}}
\right).
\end{align}

\item otherwise, $\sum_{i=1}^{2} (k_i+L_i)w_{i}\beta_{i} \leq \sum_{i=1}^{2} k_{i}w_{i}\alpha_{i} < (k_\ell+L_\ell+1)w_{\ell}\beta_{\ell}+(k_r+L_r)w_{r}\beta_{r}$, hence
$c < \min(w_{1}\beta_{1},w_{2}\beta_{2})$ and by \Cref{lem:max_theta}, case~\ref{it:maxtheta_one}

\begin{align}\label{eq:subBL2_levels_bound3}
f(L_{1},L_{2},\alpha_{1},\alpha_{2},\beta_{1},\beta_{2}) 
&=
\max\left(
\underbrace{
\frac{
L_{1}\beta_{1}^{2}+L_{2}\beta_{2}^{2}+(c/w_{1})^{2}
}{\sum_{i=1}^{2}k_{i}(\alpha_{i}^{2}+\beta_{i}^{2})}
}_{g_1(L_{1},L_{2},\alpha_{1},\alpha_{2},\beta_{1},\beta_{2})},
\underbrace{
\frac{
L_{1}\beta_{1}^{2}+L_{2}\beta_{2}^{2}+(c/w_{2})^{2}
}{\sum_{i=1}^{2}k_{i}(\alpha_{i}^{2}+\beta_{i}^{2})}
}_{g_2(L_{1},L_{2},\alpha_{1},\alpha_{2},\beta_{1},\beta_{2})}
\right) .
\end{align}
\end{enumerate}

Thus, in the range of $\alpha_{i},\beta_{i}$ involved in the definition of $f_{1}(L_{1},L_{2})$ as a supremum, there are integers $0 \leq L'_{i} \leq n_{i}-2k_{i}$ and $v \in \{1,2\}$ such that 
$f(L_{1},L_{2},\alpha_1,\alpha_{2},\beta_{1},\beta_{2})  = g_{v}(L'_{1},L'_{2},\alpha_1,\alpha_{2},\beta_{1},\beta_{2})$. 
We will shortly prove that given the relations between $L'_{i}$ and the considered range of $\alpha_{i},\beta_{i}$ we have
 
\begin{align}\label{eq:MainMissingStep}
\sum_{i=1}^{2} (k_i+L'_i)w_{i}\beta_{i} \leq  \sum_{i=1}^{2} k_{i}w_{i}\alpha_{i}.
\end{align}
hence using~\Cref{lem:bound_gv} we obtain $g_{v}(L'_1,L'_2,\alpha_1,\alpha_2,\beta_1,\beta_2)\leq  B^{L'_1,L'_2}(w)$.

This implies 
\[
f_{1}(L_{1},L_{2})  \leq \max_{0 \leq L'_{i} \leq n_{i}-2k_{i}} B^{L'_{1},L'_{2}}(w)
\]
and, combined with~\eqref{eq:DefFL1L20}-\eqref{eq:subBL2_levels_bound00}, yields the upper bound 

\begin{equation}\label{eq:BL1L2upper}
B_\Sigma(\|\cdot\|_w) =  \max_{0\leq L_i\leq n_i-2k_i} \max( f_0(L_1,L_2),f_1(L_1,L_2)) \leq \max_{0\leq L'_i\leq n_i-2k_i} B^{L'_1,L'_2}(w).
\end{equation}

{\bf Proof of~\eqref{eq:MainMissingStep}.} 
We treat separately the three cases respectively associated to~\eqref{eq:subBL2_levels_bound2},~\eqref{eq:subBL2_levels_bound4},~\eqref{eq:subBL2_levels_bound3}.
\begin{enumerate}
\item When $\sum_{i=1}^{2} (k_i+L_i)w_{i}\beta_{i} \leq  \sum_{i=1}^{2} k_{i}w_{i}\alpha_{i} < (k_\ell+L_\ell+1)w_{\ell}\beta_{\ell}+(k_r+L_r)w_{r}\beta_{r}$, by~\eqref{eq:subBL2_levels_bound3} there is $v \in \{1,2\}$ such that $f(L_{1},L_{2},\alpha_1,\alpha_{2},\beta_{1},\beta_{2})  = g_{v}(L'_{1},L'_{2},\alpha_1,\alpha_{2},\beta_{1},\beta_{2})$ with $(L'_{1},L'_{2})=(L_{1},L_{2})$.
We observe that
$\sum_{i=1}^{2} (k_i+L'_i)w_{i}\beta_{i} = \sum_{i=1}^{2} (k_i+L_i)w_{i}\beta_{i} \leq  \sum_{i=1}^{2} k_{i}w_{i}\alpha_{i}$. 
\item When $(k_\ell+L_\ell)w_{\ell}\beta_{\ell}+(k_r+L_r+1)w_{r}\beta_{r} \leq \sum_{i=1}^{2} k_{i}w_{i}\alpha_{i} < \sum_{i=1}^{2} (k_i+L_i+1)w_{i}\beta_{i}$, by~\eqref{eq:subBL2_levels_bound2}, we have
$f(L_{1},L_{2},\alpha_1,\alpha_{2},\beta_{1},\beta_{2})  = g_{v}(L'_{1},L'_{2},\alpha_1,\alpha_{2},\beta_{1},\beta_{2})$ where $(L_1',L_2',v) \in \{(L_1+1,L_2,2),(L_1,L_2+1,1)\}$. 
If $(L_\ell',L_r') =(L_\ell,L_r+1)$ then $\sum_{i=1}^{2} (k_i+L_i')w_{i}\beta_{i}=(k_\ell+L_\ell)w_{\ell}\beta_{\ell}+(k_r+L_r+1)w_{r}\beta_{r}$. Otherwise we have $(L_\ell',L_r') =(L_\ell+1,L_r)$, hence $ \sum_{i=1}^{2} (k_i+L_i')w_{i}\beta_{i} =(k_\ell+L_\ell+1)w_{\ell}\beta_{\ell}+(k_r+L_r)w_{r}\beta_{r} \leq (k_\ell+L_\ell)w_{\ell}\beta_{\ell}+(k_r+L_r+1)w_{r}\beta_{r} $ 
since $w_{\ell}\beta_{\ell} \leq w_{r}\beta_{r}$ by definition of $r,\ell$. 
In both cases we get
\(
\sum_{i=1}^{2} (k_i+L_i')w_{i}\beta_{i} \leq \sum_{i=1}^{2} k_{i}w_{i}\alpha_{i}.
\)

\item When
$(k_\ell+L_\ell+1)w_{\ell}\beta_{\ell}+(k_r+L_r)w_{r}\beta_{r} \leq \sum_{i=1}^{2} k_{i}w_{i}\alpha_{i} < (k_\ell+L_\ell)w_{\ell}\beta_{\ell}+(k_r+L_r+1)w_{r}\beta_{r}$,~\eqref{eq:subBL2_levels_bound4} yields $f(L_{1},L_{2},\alpha_1,\alpha_{2},\beta_{1},\beta_{2})  = g_{r}(L'_{1},L'_{2},\alpha_1,\alpha_{2},\beta_{1},\beta_{2})$
with $(L'_{\ell},L'_r) \in \{(L_{\ell},L_{r}), (L_{\ell}+1,L_{r})\}$, 
hence we have
$\sum_{i=1}^{2}(k_{i}+L'_{i})w_{i}\beta_{i} \leq
(k_\ell+L_\ell+1)w_{\ell}\beta_{\ell}+(k_r+L_r)w_{r}\beta_{r} \leq \sum_{i=1}^{2} k_{i}w_{i}\alpha_{i}$.

\end{enumerate}
As these three cases cover all possibilities, we deduce bound~\eqref{eq:MainMissingStep} as claimed.

To conclude, we obtain a lower bound on $B_\Sigma(\|\cdot\|_w)$. Consider any integers $0 \leq L_i \leq n_{i}-2k_{i}$ and any scalars $\alpha_i,\beta_i$ such that $0\leq \beta_{i} \leq \alpha_i$, $\beta_1+\beta_2>0$ and $\sum_{i=1}^{2} (k_i+L_i)w_{i}\beta_{i} = \sum_{i=1}^{2} k_{i}w_{i}\alpha_{i}$, and let $z = (z_{1},z_{2})$ where 
\[
  z_{i}= (\underbrace{\alpha_i,\ldots,\alpha_i}_{k_i}, \underbrace{\beta_i,\ldots,\beta_i}_{k_i+L_i},\underbrace{0,\ldots,0}_{n_{i}-(2k_i+L_{i})} ).\]
We have $\|z_T\|_w = k_{1}w_{1}\alpha_{1}+k_{2}w_{2}\alpha_{2} = (k_{1}+L_{1})w_{1}\beta_{1}+(k_{2}+L_{2})w_{2}\beta_{2} = \|z_{T^c}\|_w$ hence, by \Cref{lem:opt_support_wl1_level} and the definition of $B_{\Sigma}(\|\cdot\|_{w})$,
\begin{equation}
B_\Sigma(\|\cdot\|_w) 
\geq \frac{\|z_{T_2^c}\|_2^2}{\|z_{T_2}\|_2^2} 
= \frac{\sum_{i=1}^{2}L_i \beta_i^2 }{\sum_{i=1}^{2}k_{i}(\alpha_{i}^{2}+\beta_{i}^{2})}.
\end{equation}
Taking the supremum over $\alpha_{i},\beta_{i}$ under the considered constraints yields $B_\Sigma(\|\cdot\|_w) 
 \geq B^{L_1,L_2}(w)$.
We deduce 
\[
B_\Sigma(\|\cdot\|_w) \geq  \max_{0\leq L_i\leq n_i-2k_i} B^{L_1,L_2}(w).
\]

\end{proof}

We give a characterization/lower bound (depending on $w$) of the intermediate $B^{L_1,L_2}(w)$.

\begin{lemma} \label{lem:charact_supBL2_wl1_level}
Consider $w=(w_1,w_2)$, $0 \leq L_i \leq n_i -2k $, and $B^{L_1,L_2}(w)$  defined as in \Cref{lem:bound_gv}. 
We have
\begin{equation}\label{eq:charact_supBL2_wl1_level}
 \max_{(i,j) \in \{(1,2),(2,1)\}}\frac{L_i/k_i} {\frac{\nu_i}{1-\nu_i}(L_i/k_i)^2 +2}
 \leq B^{L_1,L_2}(w)
 \leq \max_{(i,j) \in \{(1,2),(2,1)\}}\frac{L_i/k_i} { \nu_i (L_i/k_i+1)^2 +1}
\end{equation}
with $\nu_i= \frac{1 }{1+k_j/(k_{i}\mu_i^2)}$ and $\mu_i = \frac{w_i}{w_j}$ for $(i,j) \in \{(1,2),(2,1)\}$. 
The rhs is an equality if $\nu_{i} \geq \frac{k_{i}}{k_{i}+L_{i}}$, $\forall i\in \{1,2 \}$.
\end{lemma}

\begin{proof}
For $L_1,L_2$ such that $L_1+L_2 >0$, we rewrite $B^{L_1,L_2}$ defined in~\eqref{eq:DefBL1L2} as
\begin{equation}\label{eq:TmpBoundLevels0}
 \begin{split}
B^{L_1,L_2}(w)
= &  
\sup_{\lambda>0}
\sup_{ \stackrel{\beta_i :\beta_1,\beta_2\geq 0}{\sum_{i=1}^2 w_i (k_i+L_i)\beta_i =\lambda}} \sup_{\stackrel{\alpha_i: \alpha_i \geq \beta_i}{\sum_{i=1}^2} k_iw_i\alpha_i = \lambda}\frac{ L_1\beta_1^2 +L_2 \beta_2^2}{ k_1 \alpha_1^2+k_2\alpha_2^2 +k_1 \beta_1^2+k_2\beta_2^2}.
\end{split}
\end{equation}
For fixed $\lambda>0$ and $\beta_1,\beta_{2}$ such that $\sum_{i=1}^2 w_{i}(k_{i}+L_{i})\beta_{i}=\lambda$, we have $\lambda>\sum_{i=1}^2 w_{i}k_{i}\beta_{i}$ hence, by \Cref{lem:intermediate_level_0}, 
\begin{equation}
 \begin{split}
B^{L_1,L_2}(w)
\leq &  
\sup_{\lambda>0}
\sup_{ \stackrel{\beta_1\geq 0,\beta_2\geq 0}{w_1 (k_1+L_1)\beta_1  + w_2 (k_2 +L_2)\beta_2 =\lambda}}\frac{ L_1\beta_1^2 +L_2 \beta_2^2}{ \frac{ \lambda^2 }{k_1w_1^2+k_2w_2^2}+k_1 \beta_1^2+k_2\beta_2^2}.
\end{split}
\end{equation}
 with equality if the maximizers $\hat{\lambda}, \hat{\beta}_i$ of the right side satisfy the constraints $\hat{\lambda} \geq (k_1w_1^2 + k_1w_2^2)\max(\hat{\beta}_1/w_1,\hat{\beta}_{2}/w_{2})$.

Consider $(i,j) \in \{(1,2),(2,1)\}$.  Since $\nu_i :=  \frac{1}{1+ k_jw_j^2/(k_iw_i^2)} = \frac{k_{i}w_{i}^{2}}{k_{i}w_{i}^{2}+k_{j}w_{j}^{2}}$, we obtain by \Cref{lem:intermediate_level_2}
\begin{equation}\label{eq:TmpBoundLevels}
 \begin{split}
B^{L_1,L_2}(w)
\leq &  
\sup_{\lambda>0}\max_{i \in \{1,2\}}\frac{L_i \lambda^2} { \frac{ \lambda^2 }{k_1w_1^2+k_2w_2^2} w_i^2 (k_i +L_i)^2 +k_i\lambda^2 }
=  \max_{(i,j) \in \{(1,2),(2,1)\}}\frac{L_i/k_i} {\nu_i  (L_i/k_i+1)^2 +1}.
\end{split}
\end{equation}
This establishes the upper bound in~\eqref{eq:charact_supBL2_wl1_level}. Denoting $(i^*,j^*)$ maximizing the right-hand-side expression above,  and using the optimal values from \Cref{lem:intermediate_level_2}, $\hat{\beta}_{i^*}= \frac{\hat{\lambda}}{w_{i^*}(k_{i^*}+L_{i^*})}$ (with $\hat{\beta}_{j^*} =0$ and an arbitrary $\hat{\lambda}>0$), we have 
$\max(\hat{\beta}_{1}/w_{1},\hat{\beta}_{2}/w_{2}) = \hat{\beta}_{i^{*}}/w_{i^{*}} = \frac{\hat{\lambda}}{w_{i^*}^{2}(k_{i^*}+L_{i^*})}$ hence equality holds in~\eqref{eq:TmpBoundLevels} if the following inequality is satisfied
\[
(k_1w_1^2 + k_1w_2^2)\frac{1}{w_{i^*}^2(k_{i^*}+L_{i^*})}  \leq 1,
\]
or equivalently if $\frac{k_{i^*}}{\nu_{i^{*}} (k_{i^*}+L_{i^*})} \leq 1$. This is guaranteed as soon as  $\nu_{\ell} \geq \frac{k_{\ell}}{k_{\ell}+L_{\ell}} $ for every $\ell \in \{1,2\}$. This establishes the equality case in the rhs of~\eqref{eq:charact_supBL2_wl1_level}.

We now treat the lower bound in~\eqref{eq:charact_supBL2_wl1_level}.
For fixed $\beta_i \geq 0$ and $\lambda>0$ such that  $(k_{1}+L_{1})w_{1}\beta_{1}+(k_{2}+L_{2})w_{2}\beta_{2} = \lambda$, we still have $\lambda > k_1w_1\beta_1+k_2w_2\beta_2$.  By \Cref{lem:intermediate_level_0}, letting
\begin{equation}
 \begin{split}
 V = \min_{_{\stackrel{\alpha_i: \alpha_i \geq \beta_i}{\sum_{i=1}^2} k_iw_i\alpha_i = \lambda}} k_1 \alpha_1^2 +k_2 \alpha_2^2  ,
\end{split}
\end{equation}
we either have 
\begin{itemize}
 \item $V= \min\left( k_1\beta_1^2 +k_2(\frac{\lambda -k_1w_1\beta_1}{k_2w_2})^2, k_2\beta_2^2 +k_1(\frac{\lambda -k_2w_2\beta_2}{k_1w_1})^2 \right)$;
 \item or
 \begin{align}
   V= \lambda^2/(k_1w_1^2+k_2w_2^2) &= \min_{_{\stackrel{\alpha_i: \alpha_i \geq 0}{\sum_{i=1}^2} k_i\alpha_i = \lambda}}  k_1 \alpha_1^2 +k_2 \alpha_2^2 \nonumber  \\
   &\leq \min\left( k_1\beta_1^2 +k_2(\frac{\lambda -k_1w_1\beta_1}{k_2w_2})^2, k_2\beta_2^2 +k_1(\frac{\lambda -k_2w_2\beta_2}{k_1w_1})^2 \right)
 \end{align}
where the last inequality was obtained by evaluating $k_1\alpha_1^2+k_2 \alpha_2^2$ at $\alpha_1 = \beta_1$ (resp. at $\alpha_2=\beta_2$) with $\alpha_{2}$ (resp. $\alpha_{1}$) tuned so that $k_{1}w_{1}\alpha_{1}+k_{2}w_{2}\alpha_{2}=\lambda$.
\end{itemize}
We deduce that  $V \leq  \min\left( k_1\beta_1^2 +k_2(\frac{\lambda -k_1w_1\beta_1}{k_2w_2})^2, k_2\beta_2^2 +k_1(\frac{\lambda -k_2w_2\beta_2}{k_1w_1})^2 \right)$ 
and it follows using~\eqref{eq:TmpBoundLevels0} that

\begin{equation}
 \begin{split}
B^{L_1,L_2}(w)
\geq &  
\sup_{\lambda>0}
\sup_{ \stackrel{\beta_i :\beta_1,\beta_2\geq 0}{\sum_{i=1}^2 w_i (k_i+L_i)\beta_i =\lambda}} \frac{ L_1\beta_1^2 +L_2 \beta_2^2}{ \min\left( k_1\beta_1^2 +k_2(\frac{\lambda -k_1w_1\beta_1}{k_2w_2})^2, k_2\beta_2^2 +k_1(\frac{\lambda -k_2w_2\beta_2}{k_1w_1})^2 \right)+k_1 \beta_1^2+k_2\beta_2^2} \\
= &
\sup_{ \stackrel{\beta_i :\beta_1,\beta_2\geq 0}{\sum_{i=1}^2 w_i (k_i+L_i)\beta_i =1}} \frac{ L_1\beta_1^2 +L_2 \beta_2^2}{ \min\left( k_1\beta_1^2 +k_2(\frac{1 -k_1w_1\beta_1}{k_2w_2})^2, k_2\beta_2^2 +k_1(\frac{1 -k_2w_2\beta_2}{k_1w_1})^2 \right)+k_1 \beta_1^2+k_2\beta_2^2} \\
= &
\sup_{ \stackrel{\beta_i :\beta_1,\beta_2\geq 0}{\sum_{i=1}^2 w_i (k_i+L_i)\beta_i =1}} \max \left( \frac{ L_1\beta_1^2 +L_2 \beta_2^2}{ 2k_1\beta_1^2 +k_2 \beta_2^2 +k_2(\frac{1 -k_1w_1\beta_1}{k_2w_2})^2}, \frac{ L_1\beta_1^2 +L_2 \beta_2^2}{ k_1\beta_1^2 +2k_2 \beta_2^2 +k_1(\frac{1 -k_2w_2\beta_2}{k_1w_1})^2}  \right) .\\
\end{split}
\end{equation}
For $(i,j) \in \{(1,2),(2,1) \}$, using the values $\tilde{\beta}_i = \frac{1}{w_{i}(k_{i}+L_{i})}, \tilde{\beta}_j = 0$, we have

\begin{equation}
 \begin{split}
B^{L_1,L_2}(w)
\geq &   \frac{ L_1\tilde{\beta}_i^2 }{ 2k_i\tilde{\beta}_i^2  +k_j(\frac{1 -k_iw_i\tilde{\beta}_i}{k_jw_j})^2}.
\end{split}
\end{equation}

Since
 $1-k_{i}w_{i}\tilde{\beta}_{i} = w_{i}(k_{i}+L_{i})\tilde{\beta}_{i}-k_{i}w_{i}\tilde{\beta}_{i} = w_{i}L_{i}\tilde{\beta}_{i}$, 
we have
\[
\frac{ L_i\tilde{\beta}_{i}^2}{  2k_i\tilde{\beta}_{i}^2 +k_j(\frac{1 -k_iw_i\tilde{\beta}_{i}}{k_jw_j})^2}
=
\frac{ L_i\tilde{\beta}_{i}^2}{  2k_i\tilde{\beta}_{i}^2  +k_j(\frac{w_iL_{i}}{k_jw_j})^2\beta_{i}^{2}}
=
\frac{ L_i}{  2k_i  +k_j(\frac{w_iL_{i}}{k_jw_j})^2}.
\]
Since $\nu_{i} = \frac{1}{1+ k_jw_j^2/(k_iw_i^2)}$, we have $(1-\nu_{i})/\nu_{i} = 1/\nu_{i}-1 = k_{j} w_{j}^{2}/k_{i}w_{i}^{2}$.
 We deduce
\begin{equation}
 \begin{split}
B^{L_1,L_2}(w)
 &\geq
 \frac{ L_i}{2k_i +\frac{ (w_iL_i)^2}{k_jw_j^2}} 
  =
 \frac{ L_i/k_i}{ \frac{k_iw_i^2}{k_jw_j^2}(L_i/k_i)^2 +2 } 
=
 \frac{ L_i/k_i}{ \frac{\nu_i}{1-\nu_i}(L_i/k_i)^2 +2 }.\\
\end{split}
\end{equation}

\end{proof} 

The following function study will be used to deal with the optimization of the $B^{L_1,L_2}(w)$.

\begin{lemma}\label{lem:TechnicalGH}
Consider $a$ such that $0 < a \leq 1$. The function
\begin{align}
g_{1}: u \geq 0  \mapsto g_1(u;a ):= \frac{u} {a ( u+1)^2 +1} &\qquad
\end{align}
is maximized at $u_1^* =\sqrt{1+1/a}$, increasing for $u \leq u_1^*$, decreasing for $u \geq u^{*}_{1}$ and
\begin{equation}
\begin{split}
g_1(u_1^* ; a)&= \frac{1}{2}(\sqrt{1+1/a}-1).
\end{split}
\end{equation}
\end{lemma}

\begin{proof}
Since $g_1'(u ; a) =  \frac{- a u^2 +a +1}{( a  ( u+1)^2 +1 )^2}
$, the equality $g_1'(u_1^*;a)=0$ implies $a (u_1^*)^2 = a +1$ and $u_1^* =\sqrt{1+1/a}$. 
Given the sign of $g'_{1}(u;a)$, $g_1( \cdot ;a)$ is increasing for $u \leq u_1^*$ and decreasing for $u\geq u_1^*$. As $a = [(u_{1}^{*})^{2}-1]^{-1}$ we get
\begin{equation}
\begin{split}
f_1(a):=g_1(u^*_{1},a) &= \frac{u_{1}^{*}((u_{1}^{*})^2-1)}{(u_{1}^{*}+1)^2+(u_{1}^{*})^2-1}
 =  \frac{u_{1}^{*}((u_{1}^{*})^2-1)}{2(u_{1}^{*})^2 +2u_{1}^{*} } = \frac{1}{2}(u_{1}^{*}-1) = \frac{1}{2}(\sqrt{1+1/a}-1)
\end{split}
\end{equation}
 which is decreasing with respect to $a$.
\end{proof}

\begin{lemma}\label{lem:prophlevels}
Consider $0\leq a \leq 1$, $g_1(u;a):= \frac{u} {  a ( u+1)^2 +1 }$  and $g_2(u;a):= \frac{u} { \frac{a}{1-a}  u^2 +2 }$ 
and define 
\begin{align}\label{eq:DefHFunTh7}
h_i(u,v;a) & := \max (g_i(u ; a),g_i(v;1-a)), i \in \{1,2\}.
\end{align}

\begin{enumerate}
\item \label{it:PropH1b} For $a = 1/2$ we have $h_{1}(u,v;1/2) \leq \frac{\sqrt{3}-1}{2}$
for every $u,v \geq 0$.
\item \label{it:PropH1} 
If $a \notin [\tilde{a},1-\tilde{a}]$, where $\tilde{a} := 2\sqrt{3}-3 \approx 0.46$ then $h_{2}(2,2,a) > \frac{\sqrt{3}-1}{2}$.
\end{enumerate}
Consider integers such that $k_i \geq 2$, $1 \leq 4k_{i} \leq n_{i}$ for $i \in \{1,2\}$ and 
\begin{equation}\label{eq:def_H1}
H_{1}(a) := \max_{0 \leq L_{i} \leq n_{i}-2k_{i}} h_{1}(L_{1}/k_{1},L_{2}/k_{2};a) .
\end{equation}
\begin{enumerate}[resume]
\item \label{it:PropH2m}  There exists $a^{*} \in  [\tilde{a},1-\tilde{a}]$ such that $H_{1}(a^{*}) = \min_{a \in  [\tilde{a},1-\tilde{a}]} H_{1}(a)$.

\item \label{it:PropH2} Consider $a \in [\tilde{a},1-\tilde{a}]$ then

\[
H_1(a) = \max\left(
\max_{L_1 \in \{ \lfloor k_1\sqrt{1+1/a}\rfloor; \lceil k_1\sqrt{1+1/a}\rceil\}} g_{1}(L_{1}/k_{1};a),
\max_{L_2 \in \{ \lfloor k_2\sqrt{1+1/(1-a)}\rfloor; \lceil k_2\sqrt{1+1/(1-a)}\rceil\}} g_{1}(L_{2}/k_{2};1-a)
 \right) .
 \]
where $\lfloor\cdot \rfloor$ and $\lceil \cdot\rceil$ denote the lower and upper integer part. Moreover, the  $L_i^*$ maximizing the above expression are also maximizing~\eqref{eq:def_H1} and are such that $L_1^*/k_1\geq 1/a -1$ and $L_2^*/k_2\geq 1/(1-a) -1$.
 \end{enumerate}
\end{lemma}
\begin{proof}

{\bf \Cref{it:PropH1b}.} By \Cref{lem:TechnicalGH},
with $a=1/2$ and any $u,v \geq 0$ we have 
$g_{1}(u;a) \leq g_{1}(u_{1}^{*};a) = \frac{1}{2}(\sqrt{1+1/a}-1) = (\sqrt{3}-1)/2$ and similarly $g_{1}(v;1-a) = (\sqrt{3}-1)/2$ hence $h_{1}(u,v;a) = (\sqrt{3}-1)/2$.

{\bf \Cref{it:PropH1}.} 
We prove the inequality for $a < \tilde{a}$. Since $h_{2}(2,2;1-a) = h_{2}(2,2;a)$ by definition of $h_{2}$, the same inequality holds if $a > 1-\tilde{a}$. For $a < \tilde{a}$ since $a  < 1/2$ we have $a/(1-a) < (1-a)/a$ hence using the definition of $g_{2}$ we have $g_{2}(2,a) > g_{2}(2,1-a)$. By monotonicity of $a \mapsto (1-a)/(1+a)$ we get
\[
h_{2}\left(2,2; a \right) = \max(g_2(2,a),g_2(2,1-a)) = g_2(2,a) = \frac{2}{4\frac{a}{1-a}+2} = \frac{1-a}{1+a} > \frac{1-\tilde{a}}{1+\tilde{a}} =  g_{2}(2,\tilde{a}).
\]
Finally, we compute
\[
g_{2}(2,\tilde{a}) = \frac{1-(2\sqrt{3}-3)}{1+2\sqrt{3}-3} = 
\frac{4-2\sqrt{3}}{2\sqrt{3}-2}
=\frac{(4-2\sqrt{3})(2\sqrt{3}+2)}{(2\sqrt{3})^{2}-2^{2}}
= \frac{8\sqrt{3}-12+8-4\sqrt{3}}{8}
= \frac{4\sqrt{3}-4}{8} = \frac{\sqrt{3}-1}{2}.
\]

{\bf \Cref{it:PropH2m}.}
 The function $H_{1}$ is defined as the maximum of a finite number of continuous functions of $a$. By continuity of the maximum, $H_1$ is continuous and its minimum on the compact set $ [\tilde{a},1-\tilde{a}]$ is reached.

{\bf \Cref{it:PropH2}.} 
Consider $a \in [\tilde{a},1- \tilde{a}] $. 
By~\Cref{lem:TechnicalGH}
the function $u \mapsto g_1(u ; a)$  is maximized at a 
$u_{1}^{*} = \sqrt{1+1/a}$. Similarly, $v \mapsto g_{1}(v;1-a)$ is maximized at $u_{2}^{*} = \sqrt{1+1/(1-a)}$.
 Since $1/3 \leq \tilde{a} \leq 1/2$, with $\nu_{1}=a$, $\nu_{2}=1-a$ we have $\nu_{i} \geq 1/3$ hence $u_{i}^{*} = \sqrt{1+1/\nu_{i}} \leq 2$. Moreover,
since we assume $n_i \geq 4k_i$, we have $ n_i-2k_i \geq 2k_i \geq k_i \sqrt{1 + 1/\nu_i} = k_{i}u_{i}^{*}$.
It follows that for $i \in \{1,2\}$ we have
\[
\max_{0\leq L_i\leq n_i-2k_i} g_1(L_i/k_i;\nu_{i}) 
= \max_{L_i \in\{ \lfloor k_iu_i^*\rfloor; \lceil k_iu_i^*\rceil\}}g_1(L_i/k_i;\nu_{i}).
\]
As a result, the maximizers $L_i^*$ of both sides are identical, and  we have $H_1(a) = \max( g_{1}(L_{1}^{*}/k_{1};a),g_{1}(L_{2}^{*}/k_{2};1-a))$
with
\begin{align*}
L_1^* & \in \arg\max_{L_{1} \in \{ \lfloor k_1u_1^*\rfloor; \lceil k_1u_1^*\rceil\}} g_1(L_1/k_1;a)\\
L_2^* & \in \arg\max_{L_{2} \in \{ \lfloor k_2u_2^*\rfloor; \lceil k_2u_2^*\rceil\}} g_1(L_2/k_2;1-a) .
\end{align*}

There remains to show that $L_{i}^{*}/k_{i} \geq 1/\nu_{1}-1$. For this, we first observe that since $k_{i} \geq 2$ we have
\[
L_i^*/k_i \geq \lfloor k_iu_i^*\rfloor/k_i \geq (k_{i}u_i^* -1)/k_i = u_{i}^{*} -1/k_i \geq u_{i}^{*}-1/2 = \sqrt{1+1/\nu_{i}} -1/2. 
\]
The derivative of $ x \to \sqrt{1+x} -1/2 -(x-1) =   \sqrt{1+x}  -x+1/2$ at any $x \geq 0$ is $1/(2\sqrt{1+x})-1 \leq -1/2$ hence this function is monotonically decreasing. Since $a \in [\tilde{a},1-\tilde{a}]$ and $\nu_{1}=a$, $\nu_{2}=1-a$ we have $\nu_{i} \geq \tilde{a}$ hence $1/\nu_{i} \leq 1/\tilde{a}$ for $i \in \{1,2\}$, hence
\[
 \sqrt{1+1/\nu_{i}} -1/2 -(1/\nu_{i} -1) \geq  \sqrt{1+1/\tilde{a}} -1/2 -(1/\tilde{a} -1) \approx 0.12 >0.   
 \]
 We deduce that $L_i^*/k_i \geq 1/\nu_{i}-1$ as claimed.

\end{proof}

We can conclude with the proof of \Cref{th:opt_in_levels}.

\begin{proof}[Proof of \Cref{th:opt_in_levels}]
The proof starts from the fact (\Cref{cor:RCnecUoS}) that 
\begin{equation}
\arg \max_{R\in\sC'} \delta_\Sigma^{\mathtt{nec}}(R)=   \arg\min_{R \in \sC'} B_{\Sigma}(R)\qquad 
\end{equation}
with $\sC' = \{R(\cdot) =  \|\cdot\|_w : w =(w_1,w_2), w_1 > 0,w_2 > 0  \}$.
Using \Cref{lem:charact_supBL2_wl1_level-1}, for each $w$ we have 
\begin{equation}\label{eq:proofthlevels1}
\begin{split}
 B_\Sigma(\|\cdot\|_w) &=  \max_{
0\leq L_i \leq n-2k_i, i \in \{1,2\}} B^{L_1,L_2}(w).
 \end{split}
\end{equation}

With the notations of~\Cref{lem:charact_supBL2_wl1_level} we have $\mu_{1}=w_{1}/w_{2}$ and $\mu_{2}=w_{2}/w_{1}$ hence $\mu_1 =1/\mu_2$,  and one can check that $\nu_{1}+\nu_{2}=1$ where $\nu_{1} = \nu_{1}(w) := (1+k_2w_{2}^{2}/(k_{1}w_1^2))^{-1}$. 
Hence,  by \Cref{lem:charact_supBL2_wl1_level} (taking $u =L_1/k_1, v = L_2/k_2, a = \nu_1$ and using~\eqref{eq:DefHFunTh7}) and with the notation of \Cref{lem:prophlevels}, for   all integers $0 \leq L_{i} \leq n_{i}-2k_{i}$ we have
\begin{equation}\label{eq:proofthlevels1bis}
h_{2}(L_{1}/k_{1},L_{2}/k_{2};\nu_{1}) \leq
B^{L_1,L_2}(w)
\leq
h_{1}(L_1/k_1,L_2/k_2;\nu_1)
\end{equation}
with equality in the right hand s if for each $i \in \{1,2\}$ we have $\nu_i \geq L_i/(k_i+L_i)$, i.e., $L_{i}/k_{i} \geq 1/\nu_{i}-1$. Using~\eqref{eq:proofthlevels1} we get
\begin{equation}\label{eq:proofthlevels2}
\max_{0\leq L_i \leq n-2k_i, i \in \{1,2\}} 
h_{2}(L_1/k_1,L_2/k_2;\nu_1)
\leq
B_{\Sigma}(\|\cdot\|_{w})
\leq 
 \max_{
0\leq L_i \leq n-2k_i, i \in \{1,2\}}h_{1}(L_1/k_1,L_2/k_2;\nu_1),
\end{equation}
and if the maximizers $L_{i}^{*}$ of the right-hand side of~\eqref{eq:proofthlevels2} satisfy $L_i^{*}/k_i \geq 1/\nu_i -1$ for each $i \in \{1,2\}$
then in fact 
\begin{equation}\label{eq:proofthlevels2bis}
B_{\Sigma}(\|\cdot\|_{w})
= H_1(\nu_1) :=
h_{1}(L_1^{*}/k_1,L_2^{*}/k_2;\nu_1)
\end{equation}
where $H_1$ is defined as the maximum of $h_{1}$ over the $L_i/k_i$ (\Cref{lem:prophlevels}). In particular for $\nu_i = 1/2$, this is verified if $L_i^* \geq k_i$.
Next we proceed in three steps. We set $\tilde{a} := 2\sqrt{3}-3 \approx 0.46$.

{\bf Step 1.} We show that if $w',w''$ are such that $\nu_{1}(w') \notin [\tilde{a},1-\tilde{a}]$ and $\nu_{1}(w'') = 1/2 \in [\tilde{a},1-\tilde{a}]$ then
\[
B_{\Sigma}(\|\cdot\|_{w'}) > \frac{\sqrt{3}-1}{2} \geq B_{\Sigma}(\|\cdot\|_{w''}).
\]
A first consequence is to establish~\eqref{eq:MainThmLevelsBound1}, using \Cref{cor:RCnecUoS} to convert the bound on $B_{\Sigma}(\|\cdot\|_{w})$, for $w \in \{w^{*},w_{0}\}$, to a bound on $\delta_{\Sigma}^{\mathtt{nec}}(\|\cdot\|_{w})$. Indeed, since $\nu_{1}(w'') = k_{1}(w''_{1})^{2}/(k_{1}(w''_{1})^{2}+k_{2}(w''_{2})^{2})$, the fact that $\nu_{1}(w'') = 1/2$ corresponds to $w'' \propto (1/\sqrt{k_{1}},1/\sqrt{k_{2}}) =: w_{0}(k_{1},k_{2})$, hence $B_{\Sigma}^{*} \leq B_{\Sigma}(\|\cdot\|_{w_{0}}) \leq (\sqrt{3}-1)/2$.

A second consequence is that
 the optimization of $w = (w_{1},w_{2})$ can be restricted to a range corresponding to $\nu_{1} = \nu_{1}(w) \in [\tilde{a},1-\tilde{a}]$. 

Indeed, on the one hand,  for $\nu_{1}(w'') = 
 1/2$,  by \Cref{lem:prophlevels}-\Cref{it:PropH2} we have ,
\[
H_1(1/2)  = \max\left(
\max_{L_1 \in \{ \lfloor k_1\sqrt{3}\rfloor; \lceil k_1\sqrt{3}\rceil\}} g_{1}(L_{1}/k_{1};1/2),
\max_{L_2 \in \{ \lfloor k_2\sqrt{3}\rfloor; \lceil k_2\sqrt{3}\rceil\}} g_{1}(L_{2}/k_{2};1/2)
 \right)
 \]
 where $g_1$ is defined in \Cref{lem:prophlevels}.

Hence, $L_i^* \geq k_i$ so that~\eqref{eq:proofthlevels2bis} holds, and we deduce that $B_{\Sigma}(\|\cdot\|_{w''}) =H_1(1/2)$.

From \Cref{lem:prophlevels}-\Cref{it:PropH1b}, we have
\[
H_1(1/2)
\leq \frac{\sqrt{3}-1}{2} ,
 \]
 and we obtain $B_{\Sigma}(\|\cdot\|_{w}) \leq (\sqrt{3}-1)/2$ as claimed.

 We can also establish the conclusion of the theorem \eqref{eq:MainThmLevelsBound2} by  considering $\sup_{k_1',k_2' \geq 1,n_1'\geq 4k_1',n_2'\geq 4k_2'}B_{\Sigma}(\|\cdot\|_{w''})$. Using the expression of $H(1/2)$ we have that
 \begin{equation}
 \begin{split}
  \sup_{k_1',k_2' \geq 1,n_1'\geq 4k_1',n_2'\geq 4k_2'}B_{\Sigma}(\|\cdot\|_{w''}) &= \sup_{k_1'\geq 1,n_1'\geq 4k_1'} \max_{L_1 \in \{ \lfloor k_1'\sqrt{3}\rfloor; \lceil k_1'\sqrt{3}\rceil\}} g_{1}(L_{1}/k_{1}';1/2) \\
  &=  \sup_{k\geq 1} \max_{L \in \{ \lfloor k\sqrt{3}\rfloor; \lceil k\sqrt{3}\rceil\}} g_{1}(L/k;1/2) \\
  \end{split}
 \end{equation}
Using \Cref{lem:TechnicalGH}, as $g_1$ is continuous and $\lfloor k\sqrt{3}\rfloor/k \to_{k\to \infty} \sqrt{3} =u_1^*$ the maximizer of $g_1(\cdot;1/2)$  (because $  (k\sqrt{3}-1)/k\leq \lfloor k\sqrt{3}\rfloor/k \leq \sqrt{3}$), we have $ \sup_{k} \max_{L \in \{ \lfloor k\sqrt{3}\rfloor; \lceil k\sqrt{3}\rceil\}} g_{1}(L/k;1/2) = g_{1}(u_1^*;1/2) = (\sqrt{3}-1)/2$. Again using  \Cref{cor:RCnecUoS} to link $\delta_\Sigma^{\mathtt{nec}}$ and $B_\Sigma$ yields \eqref{eq:MainThmLevelsBound2}.

On the other hand, if $\nu_{1} \notin [\tilde{a},1-\tilde{a}]$ then by \Cref{lem:prophlevels}-\Cref{it:PropH1} we have $h_{2}(2,2;\nu_1) > (\sqrt{3}-1)/2$. Since $n_{i} \geq 4k_{i}$,  the integers $L_{i} := 2k_{i}$, $i \in \{1,2\}$ satisfy $0 < L_{i} \leq n_{i}-2k_{i}$ hence, by the left-hand side in \eqref{eq:proofthlevels2},
\[
B_{\Sigma}(\|\cdot\|_{w'}) \geq h_{2}(L_{1}/k_{1},L_{2}/k_{2},\nu_{1}) = h_{2}(2,2;\nu_1) > \frac{\sqrt{3}-1}{2}.
\]  

{\bf Step 2.} We show that if $w$ satisfies  $\nu_1=\nu_{1}(w) \in [\tilde{a},1-\tilde{a}]$ then $B_{\Sigma}(\|\cdot\|_{w}) = H_1(\nu_1(w))$.

Since $k_{i} \geq 2$ and $n_{i} \geq 4k_{i}$, by \Cref{lem:prophlevels}-\Cref{it:PropH2}, we have the equality $H_1(\nu_{1})=h_1(L_1^*/k_1,L_2^*/k_2,\nu_{1})$ where 
$L_1^* \in \{ \lfloor k_1\sqrt{1+1/\nu_1}\rfloor; \lceil k_1\sqrt{1+1/\nu_1}\rceil\}$, $L_2^* \in \{ \lfloor k_2\sqrt{1+1/(1-\nu_1)}\rfloor; \lceil k_2\sqrt{1+1/(1-\nu_1)}\rceil\}$ and $L_i^*/k_i \geq 1/\nu_i-1$.  By~\eqref{eq:proofthlevels2}-
\eqref{eq:proofthlevels2bis} we deduce that the equality $B_{\Sigma}(\|\cdot\|_{w}) = H_1(\nu_1(w))$ holds.

{\bf Step 3.} By \Cref{lem:prophlevels}-\Cref{it:PropH2m}, there is $a^{*} \in [\tilde{a},1-\tilde{a}]$ such that $H_{1}(a^{*}) = \min_{\tilde{a} \leq a \leq 1-\tilde{a}} H_{1}(a)$. In light of Steps 1 and 2, the infimum over $w$ of $B_{\Sigma}(\|\cdot\|_{w})$ is thus achieved, and a weight vector $w^{*}$ satisfies
\begin{equation}
B_{\Sigma}(\|\cdot\|_{w^{*}}) = \min_{w} B_{\Sigma}(\|\cdot\|_{w}) ) = H_{1}(a^{*})
\end{equation}
if, and only if $H_{1}(\nu_{1}(w^{*})) = H_{1}(a^{*})$. 
 Since $\nu_{1}(w) = \left(1+\frac{k_{2}}{k_{1}}(w_{2}/w_{1})^{2}\right)^{-1}$, combining all the above yields
 \[
 \frac{w_{2}^{*}}{w_{1}^{*}} = \sqrt{\frac{k_{1}}{k_{2}}(1/\nu_{1}^{*}-1)}
 \]
 where $\nu_{1}^{*}$ is an optimum of
\begin{equation*}
\begin{split}
 B_{\Sigma}(\|\cdot\|_{w^{*}})
&=  \min_{\nu_1 \in  [\tilde{a},1-\tilde{a}], \nu_2 =1-\nu_1} \max_{i\in \{1,2\}}\max_{x_i \in \{ \lfloor k_i\sqrt{1+1/\nu_i}\rfloor; \lceil k_i\sqrt{1+1/\nu_i}\rceil\}} g_1(x_i/k_i ; \nu_i) . \\
\end{split}
\end{equation*}

\end{proof}

The following Lemma is needed for the proof of \Cref{th:RIP_suff_in_levels}.

\begin{lemma}\label{lem:Case3Knorm}
Consider integers $n \geq k \geq 1$, a nonzero vector $z \in \mathbb{R}^{n}$, $S$ a set of the $k$ largest entries of $z$.
 There exists $0 \leq r \leq k-1$,
$\beta \in \mathbb{R}^{n}$, $\gamma \geq 0$ 
such that
\begin{align}
\label{eq:BetaL0}
\|\beta\|_{0} &= k-r-1 \leq L :=  |\supp(z_{S^{c}})| \leq n-1\\
\label{eq:BetaLinf}
\|\beta\|_{\infty} & \leq \min_{l \in S} |z(l)|\\
\|z_{S^c}\|_1 &= \|\beta\|_1 +  \gamma 
 \label{eq:IneqThetaBetaCase3KNorm2}\\
\|z_{S^{c}}\|_{\Sigma_{k}}^{2}
&=
\|\beta\|_2^2+\frac{1}{r+1}
\gamma^{2}
\label{eq:IneqCase3KNorm}
\end{align}
Moreover if $k-r-1 \geq 1$ then
\begin{align}
\gamma < (r+1) \min_{l \in \supp(\beta)} |\beta(l)|\; ;
\label{eq:IneqThetaBetaCase3KNorm}
\end{align}

\end{lemma}

\begin{proof}
We use the fact that for any integer $k$ the norm $\|\cdot\|_{\Sigma_k}$ coincides with the so-called $k$-support norm \cite[Definition 2.1]{Argyriou_2012}, that is invariant by permutation of the coordinates and has the following expression for each $y \in \mathbb{R}^{n}$ sorted in descending order:
\begin{equation}\label{eq:ExplicitExpressionKNorm}
\|y\|_{\Sigma_k}^2 = \sum_{l=1}^{k-r-1} |y(l)|^2 + \frac{1}{r+1} \left(\sum_{l=k-r}^{n} |y(l)|\right)^2
\end{equation}
 where $r$ is the unique integer in $ \{0, \ldots,k-1\}$ such that
 \begin{equation}\label{eq:DefRiFromYi}
 |y(k-r-1)| > \frac{1}{r+1} \sum_{l=k-r}^{n} |y(l)| \geq |y(k-r)| ,
 \end{equation}
with the convention $|y(0)|=+\infty$, see \cite[Proposition 2.1]{Argyriou_2012}.

We apply the above characterization with $y \in \mathbb{R}^{n}$ the sorting of $z_{S^c}$ by descending order of absolute values.
 Notice that $k-r-1 \leq L := |\supp(z_{S^{c}})| = |\supp(y_{i})| \leq n-1$ (since $z \neq 0$, $|S| \geq 1$), which establishes the rhs inequality in~\eqref{eq:BetaL0}: otherwise we would have $y(k-r-1) = y(k-r) = 0$ which would contradict~\eqref{eq:DefRiFromYi}.

We define $\beta \in \mathbb{R}^{n}$ and $\gamma \geq 0$ as
\begin{align}
\beta(l) & := |y(l)|,\  1 \leq l \leq k-r-1;\quad \beta(l):=0,\ k-r \leq l \leq n;\label{eq:DefBetai}\\
\gamma & :=\sum_{l=k-r}^{n} |y(l))|
 \label{eq:DefThetai}
\end{align}
Since $\beta(l)$ is non-increasing with $l$, with $\beta(k-r-1) = |y(k-r-1)| > 0$ and $\beta(k-r)=0$ we have $\|\beta\|_{0}=k-r-1$ hence~\eqref{eq:BetaL0} holds. Moreover, by definition of $S$, we also have $\min_{l \in S}|z(l)| \geq \|z_{S^{c}}\|_{\infty} = \|y\|_{\infty} = \|\beta\|_{\infty}$ hence~\eqref{eq:BetaLinf} holds.
We also obviously have $\|z_{S^c}\|_1  = \|y\|_1= \|\beta\|_1 + \gamma
$, i.e. the required identity \eqref{eq:IneqThetaBetaCase3KNorm2}.

Since $y$ is a decreasing rearrangement of $z_{S^{c}}$ and $\|\cdot\|_{\Sigma_{k}}$ is invariant by permutation we have
\[
\|z_{{S}^{c}}\|_{\Sigma_{k}}^{2}
=
\|y\|_{\Sigma_{k}}^2
\stackrel{\eqref{eq:ExplicitExpressionKNorm}+\eqref{eq:DefBetai}+\eqref{eq:DefThetai}}{=}
\|\beta\|_2^2 +\frac{1}{r+1} \gamma^{2}.
\]
This establishes~\eqref{eq:IneqCase3KNorm}.
Finally  when $k-r-1 \geq 1$ we have $\min_{l \in \supp(\beta)}\beta(l) = \beta(k-r-1) = |y(k-r-1)|$ hence \eqref{eq:IneqThetaBetaCase3KNorm} is a direct consequence of \eqref{eq:DefRiFromYi}.

\end{proof}

We now give the proof of \Cref{th:RIP_suff_in_levels}.

\begin{proof}[Proof of \Cref{th:RIP_suff_in_levels}]
The assumptions of \Cref{lem:charac_suff_RIP} hold, so we can rely on expression~\eqref{eq:RIPsuffVsRIPsuff2}: to lower bound $ \delta_{\Sigma_{k_1,k_2}}^{\mathtt{suff}} ( \|\cdot\|_w)$ by $1/\sqrt{3}$, we thus upper bound
$\|z-P_\Sigma(z)\|_\Sigma^2$ by $2\|P_\Sigma(z)\|_{2}^2$ for  every $z\in \sT_\Sigma(\|\cdot\|_w)$. First we characterize $P_{\Sigma}(z)$ for any $z$.
With $T=T(z) = (S_{1},S_{2})$ defined as in the beginning of \Cref{sec:proofslevels} we have $P_{\Sigma}(z) = z_{T}$
because for $z = (z_{1},z_{2}) \in \sH$ we have
\[
\min_{y \in \Sigma} \|z -y\|_2^{2} =  \min_{y_{1} \in \Sigma_{k_{1}}, y_2 \in \Sigma_{k_{2}}}
\big(\|z_1 -y_1\|_{2}^{2} +  \|z_2 -y_2\|_2^{2}\big) = \|z_{1}-(z_{1})_{S_{1}}\|^{2}+\|z_{2}-(z_{2})_{S_{2}}\|_{2}^{2} = \|z-z_{T}\|_2^{2}.
\]

We will use that, by \Cref{lem:opt_support_wl1_level}, we have
\begin{equation}\label{eq:TmpIneqLastProof2}
z\in \sT_\Sigma(\|\cdot\|_w) \Leftrightarrow
\sum_{i=1}^2
\|(z_{i})_{S_{i}^{c}
}\|_{1}/\sqrt{k_i}
\leq \sum_{i=1}^2\|(z_{i})_{S_{i}
}\|_{1}/\sqrt{k_i} .
\end{equation}

Now, using the fact that
$\|(u_1,u_2)\|_\Sigma^2 = \|u_1\|_{\Sigma_{k_{1}}}^2 + \|u_2\|_{\Sigma_{k_{2}}}^2$ (from \cite[Lemma 4.2]{Traonmilin_2016})  we obtain
 \begin{equation}\label{eq:TmpIneqLastProof1}
\|z-P_\Sigma(z)\|_\Sigma^2
 = \|z_{T^c}\|_\Sigma^2
 = \sum_{i=1}^2
 \|(z_{i})_{S_i^c
 }\|_{\Sigma_{k_i}}^2
 \end{equation}

With Lemma~\ref{lem:Case3Knorm}, we obtain an explicit expression of the  ratio $\|z-P_\Sigma(z)\|_\Sigma^2 / \|P_\Sigma(z)\|_2^2$. For $i \in \{1,2 \}$, let $L_i = |\supp((z_{i})_{S_i^c})|$. There exists $0 \leq r_{i} \leq k_{i}-1$  and
$\beta_{i} \in \mathbb{R}^{n_{i}}$, $\gamma_{i} \geq 0$
such that
$\|\beta_{i}\|_{\infty} \leq \min_{l \in S_{i}}|z_{i}(l)|$,
$\|\beta_{i}\|_{0} = k_{i}-r_{i}-1 \leq L_{i}$ and
\begin{align}
\|(z_{i})_{S_{i}}\|_{\Sigma_{k_{i}}}^{2}
&=
\|\beta_i\|_2^2+\frac{1}{r_{i}+1}\gamma_{i}^{2}
\\
\|(z_{i})_{S_{i}^c}\|_1 &= \|\beta_{i}\|_1 + \gamma_{i}. \label{eq:IneqThetaBetaCase3KNormInit}
\end{align}
where,
if $k_{i}-r_{i}-1 \geq 1$, we further have
\begin{equation}\label{eq:IneqThetaBetaCase3KNormFirstExpression}
\gamma_{i} < (r_{i}+1) \min_{l \in \supp(\beta_{i})} |\beta_{i}(l)|.
\end{equation}

Consider $i \in \{1,2\}$. Depending on the value of $ k_{i}-r_i-1$, we have the following properties
\begin{itemize}
\item[$\bullet$] If $k_{i}-r_{i}-1=0$ then $\|\beta_i\|_0 =0$ and $\beta_{i}=0$. Hence by \eqref{eq:IneqThetaBetaCase3KNormInit} we have $\gamma = \|(z_{i})_{S_{i}^{c}}\|_{1}$ and
\begin{equation}
\|\beta_{i}\|_2^{2}+\frac{1}{r_{i}+1}\gamma_{i}^{2}
= \frac{\|(z_{i})_{S_{i}^{c}}\|_{1}^{2}}{k_{i}}. \label{eq:boundbetanum2}
\end{equation}
\item[$\bullet$] If $k_{i}-r_i-1 \geq 1$ then, since $k_{i}\|\beta_{i}\|_{\infty}^{2} \leq
k_{i} \min_{l \in S} |z_{i}(l)|^{2} \leq \|(z_{i})_{S_{i}}\|_{2}^{2}
$ and $\|\beta_{i}\|_{0} = k_{i}-r_{i}-1$, we have
\begin{eqnarray}
\|\beta_{i}\|_2^{2}+\frac{1}{r_{i}+1}
\gamma_{i}^{2}
 \stackrel{\eqref{eq:IneqThetaBetaCase3KNormFirstExpression}}{\leq} &
\|\beta_{i}\|_2^{2} + (r_{i}+1) \cdot \left( \min_{l \in \supp(\beta_{i})} |\beta_{i}(l)|\right)^{2}\notag\\
 \leq &
[(k_{i}-r_{i}-1)+(r_{i}+1)] \cdot \|\beta_{i}\|_\infty^{2}
=
k_{i}\|\beta_{i}\|_\infty^{2} \leq  \|(z_{i})_{S_{i}}\|_2^2. \label{eq:boundbetanum1}
\end{eqnarray}

\end{itemize}

Thanks to these properties, we distinguish two easy cases to bound $\|z-P_\Sigma(z)\|_\Sigma^2 / \|P_\Sigma(z)\|_2^2$.
\begin{enumerate}
 \item If $k_{i}-r_i-1 \geq 1$ for each $i\in \{1,2\}$ then
\begin{align*}
\|z-P_\Sigma(z)\|_\Sigma^2
&\stackrel{\eqref{eq:IneqCase3KNorm}}{=}
\sum_{i=1}^{2} \left(\|\beta_{i}\|_2^{2}+\frac{1}{r_{i}+1}\gamma_{i}^{2}
\right)
\stackrel{\eqref{eq:boundbetanum1} }{\leq}
\sum_{i=1}^{2} \|(z_{i})_{S_{i}}\|_2^2
=\|P_\Sigma(z)\|_2^{2}.
\end{align*}
\item If $k_{i}-r_{i}-1=0$ for each $i\in \{1,2\}$ then
\begin{align*}
\|z-P_\Sigma(z)\|_\Sigma^2
\stackrel{\eqref{eq:IneqCase3KNorm}}{=}
\sum_{i=1}^{2} \left(\|\beta_{i}\|_2^{2}+\frac{1}{r_{i}+1}\gamma_{i}^{2}
\right)
\stackrel{\eqref{eq:boundbetanum2}}{=}&
\sum_i\|z_{S_i^c}\|_1^2/k_i
\\
\stackrel{|a|^2+|b|^2\leq (|a|+|b|)^2}{\leq}&
(\sum_i\|z_{S_i^c}\|_1/\sqrt{k_i})^2
\\
\stackrel{\eqref{eq:TmpIneqLastProof2}}{\leq}&
(\sum_i\|z_{S_i}\|_1/\sqrt{k_i})^2
\\
\leq&
(\sum_i\|z_{S_i}\|_2)^2
\\
\stackrel{ (|a|+|b|)^2\leq 2(|a|^2+|b|^2)}{\leq}&
2
\sum_i\|z_{S_i}\|_2^2
=2\|P_\Sigma(z)\|_2^{2}.
\end{align*}
\end{enumerate}
In both cases we obtain $\|z-P_\Sigma(z)\|_\Sigma^2/\|P_\Sigma(z)\|_2^{2}\leq 2$.

When these easy cases do not hold we have e.g. $0 = k_{1}-r_1-1$ and $k_{2}-r_2-1 \geq 1$ (the same reasoning holds if $k_{2}-r_2-1 =0$ and $k_{1}-r_1-1 \geq 1$), $\|z_{S_{1}^c}\|_{\Sigma_{k_1}}^2 = \|z_{S_{1}^c}\|_{1}^2/k_1$ and $\|z_{S_{2}^c}\|_{\Sigma_{k_2}}^2 = \|\beta_2\|_{2}^2+\gamma^2 /(r+1) $ .

This gives
\begin{align*}
 \|z-P_\Sigma(z)\|_\Sigma^2
&=
  \|z_{S_{1}^c
  }\|_{1}^2/k_1 +\|\beta_2\|_{2}^2+\gamma^2 /(r+1) \\
  &
\stackrel{\eqref{eq:TmpIneqLastProof2}}{\leq} (\|z_{S_{1}
  }\|_{1}/\sqrt{k_1} + \|z_{S_{2}
  }\|_{1}/\sqrt{k_1}- \|z_{S_2^c}\|_1/\sqrt{k_2})^2+\|\beta_2\|_{2}^2+\gamma^2 /(r+1) \\
  &
\stackrel{\eqref{eq:IneqThetaBetaCase3KNormInit}}{=}(\|z_{S_{1}
  }\|_{1}/\sqrt{k_1} + \|z_{S_{2}
  }\|_{1}/\sqrt{k_1}- ( \|\beta_2\|_1 +\gamma)/\sqrt{k_2})^2+\|\beta_2\|_{2}^2+\gamma^2 /(r+1) .\\
   \end{align*}

We have $0 \leq \|z_{S_{1}
  }\|_{1}/\sqrt{k_1} + \|z_{S_{2}
  }\|_{1}/\sqrt{k_1}- ( \|\beta_2\|_1 +\gamma)/\sqrt{k_2})^2\leq \|z_{S_{1}
  }\|_{2} + \|z_{S_{2}
  }\|_{2}- ( \|\beta_2\|_1 +\gamma)/\sqrt{k_2})^2$ and

  \begin{align}
\|z-P_\Sigma(z)\|_\Sigma^2
  &\leq(\|z_{S_{1}
  }\|_{2} + \|z_{S_{2}
  }\|_{2}- ( \|\beta_2\|_1 +\gamma)/\sqrt{k_2})^2+\|\beta_2\|_{2}^2+\gamma^2 /(r+1) .
   \end{align}

   As this a quadratic function of $\gamma \geq 0$ with positive leading coefficient it is maximized, at either bound of the range of $\gamma$, i.e $\gamma=0$ or  $\gamma \to (r_2+1) \min_{l \in \supp(\beta_2)} |\beta_2(l)| =: (r_2+1) \tilde{\beta}_2$.

   For the case, $\gamma \to (r_2+1) \tilde{\beta}_2$,
\begin{align}
 \|z-P_\Sigma(z)\|_\Sigma^2
&\leq (\|z_{S_{1}
  }\|_{2} + \|z_{S_{2}
  }\|_{2}-  (\|\beta_2\|_1 + (r_2+1) \tilde{\beta}_2)/\sqrt{k_2})^2+\|\beta_2\|_{2}^2 +(r_2+1) \tilde{\beta}_2^2.
   \end{align}

 Let us call $f(\beta_2)$ the numerator. For  fixed $\tilde{\beta}_2,\|\beta_2\|_\infty$, consider  $\beta_2^*$ the maximizer of $f$ under the constraint $ \tilde{\beta}_2\leq |\beta_2(l)|\leq \|\beta_2\|_\infty$.   We remark that given $l \in \supp(\beta_2)$ , we have that $ f( \beta_2 )$ where we fixed $\beta_2(l') = \beta_2^*(l')$ for $l'  \in \supp(\beta_2) \setminus \{ l\}$ is a quadratic function of  $|\beta_2(l)|$ with positive leading coefficient. Under the constraint $ \tilde{\beta}_2\leq |\beta_2(l)|\leq \|\beta_2\|_\infty$, it is maximized at either of the two bounds on $|\beta_2(l)|$, we deduce that $\beta_2^*(l) = \tilde{\beta}_2$ or $\beta_2^*(l) = \|\beta_2\|_\infty$.

   This implies that there is (an integer -- but we relax this constraint ) $0\leq s \leq \|\beta_{2}\|_{0} = k_2-r_2-1$ such that
   \begin{align}
  \|z-P_\Sigma(z)\|_\Sigma^2&\leq (\|z_{S_{1}
  }\|_{2} + \|z_{S_{2}
  }\|_{2}-  (s\|\beta_2\|_\infty + (k_2-r_2-1-s) \tilde{\beta}_2
  + (r_2+1) \tilde{\beta}_2 )/\sqrt{k_2})^2\notag\\
  &+ s\|\beta_2\|_{\infty}^2 +(k_2-r_2-1-s) \tilde{\beta}_2^2 +(r_2+1) \tilde{\beta}_2 ^2.
   \end{align}
 Again the right side is  a quadratic function of $\tilde{\beta}_2$ under the constraint $ 0 <  \tilde{\beta}_2 \leq \|\beta\|_\infty$, and bounded at either $   \tilde{\beta}_2 \to 0$ or
$ \tilde{\beta}_2 = \|\beta_2\|_\infty$:

 \begin{align}
  &\|z-P_\Sigma(z)\|_\Sigma^2\notag\\
&\leq \max \left( (\|z_{S_{1}
  }\|_{2} + \|z_{S_{2}
  }\|_{2}-  s\|\beta_2\|_\infty /\sqrt{k_2})^2+ s\|\beta_2\|_{\infty}^2 , (\|z_{S_{1}
  }\|_{2} + \|z_{S_{2}
  }\|_{2}-  \sqrt{k_2}\|\beta_2\|_\infty ))^2+ k_2\|\beta_2\|_{\infty}^2  \right)\notag\\
&\leq \max_{0 \leq s' \leq k_{2}} \left( (\|z_{S_{1}
  }\|_{2} + \|z_{S_{2}
  }\|_{2}-  s'\|\beta_2\|_\infty /\sqrt{k_2})^2+ s'\|\beta_2\|_{\infty}^2\right).
   \end{align}
For each $0\leq s' \leq k_{2}$ the denominator in the last line is a quadratic function of
of   $\|\beta_2\|_\infty $
   with $0 < \|\beta_2\|_\infty \leq \min_{l \in S_2}|z(l)|$ (from Lemma~\ref{lem:Case3Knorm}) hence
 \begin{align}
  &\|z-P_\Sigma(z)\|_\Sigma^2\notag\\
&\leq
\max_{0 \leq s'\leq k_{2}}  \max \left( (\|z_{S_{1}
  }\|_{2} + \|z_{S_{2}
  }\|_{2})^2 , (\|z_{S_{1}
  }\|_{2} + \|z_{S_{2}
  }\|_{2}-   s'\min_{l \in S_2}|z(l)|]/\sqrt{k_2})^2+s' [\min_{l \in S_2}|z(l)| ]^2  \right)\\
  &
  =
  \max_{0 \leq s'\leq k_{2}}\left(
 (\|z_{S_{1}
  }\|_{2} + \|z_{S_{2}
  }\|_{2}-   s'\min_{l \in S_2}|z(l)|]/\sqrt{k_2})^2+s' [\min_{l \in S_2}|z(l)| ]^2 \right) .
\label{eq:boundzPz3rdcase}
  \end{align}

Similarly, for $\gamma = 0 $
\begin{align}
 \|z-P_\Sigma(z)\|_\Sigma^2
&\leq (\|z_{S_{1}
  }\|_{2} + \|z_{S_{2}
  }\|_{2}-  \|\beta_2\|_1 /\sqrt{k_2})^2+\|\beta_2\|_{2}^2
\end{align}
which implies using the same argument used to obtain ~\eqref{eq:boundzPz3rdcase}
 \begin{align}
 \|z-P_\Sigma(z)\|_\Sigma^2
&\leq \max_{0 \leq s'\leq k_{2}} 
(\|z_{S_{1}
  }\|_{2} + \|z_{S_{2}
  }\|_{2}-   s'\min_{l \in S}|z(l)|]/\sqrt{k_2})^2+s' [\min_{l \in S}|z(l)| ]^2 .
   \end{align}

  Still using the same argument about quadratic functions the right side of the maximum is bounded by either of the cases $s'=0$ or $s' = k_2 $. For $s' = 0 $

 \begin{align*}
 \|z-P_\Sigma(z)\|_\Sigma^2
&\leq
(\|z_{S_{1}
  }\|_{2} + \|z_{S_{2}
  }\|_{2})^2  
    \leq 2(\|z_{S_{1}
  }\|_{2}^2 + \|z_{S_{2}
  }\|_{2}^2)  = 2 \|P_\Sigma(z)\|_2^2.
   \end{align*}
For $s' = k_2 $
\begin{align*}
  \frac{\|z-P_\Sigma(z)\|_\Sigma^2}{\|P_\Sigma(z)\|_2^2}
&\leq V:= \frac{(\|z_{S_{1}
  }\|_{2} + \|z_{S_{2}
  }\|_{2}-   \sqrt{k_2}\min_{l \in S}|z(l)|])^2+ k_2 [\min_{l \in S}|z(l)| ]^2  }{ \sum_{i=1}^{2} \|z_{S_i}\|_2^2} .
  \end{align*}

  Denote $c= 1 -  \sqrt{k_2}\frac{[\min_{l \in S}|z(l)|]}{\|z_{S_2}\|_2} \in [0,1)$ and consider $r >0$, $\theta \in [0,\pi/2]$ such that $\|z_{S_{1}
  }\|_{2} = r \cos\theta$, $\|z_{S_{2}
  }\|_{2} = r \sin\theta$
  \begin{align*}
 V &= \frac{(r \cos \theta+cr \sin\theta)^2 +(1-c)^2r^2\sin^2\theta}{r^2 \cos^2 \theta+r^2 \sin^2\theta}\\
 &=\frac{ \cos^2 \theta+c^2 \sin^2\theta +(1-c)^2\sin^2\theta + 2c \sin  \theta \cos\theta}{ \cos^2 \theta+ \sin^2\theta}\\
 &=1+(2c^2-2c) \sin^2\theta  + 2c \sin  \theta \cos\theta.\\
 \end{align*}
This is a quadratic function of $c$ with positive leading coefficient hence it is maximized at $c =0$ or $c =1 $ . Hence
\begin{align*}
 V&\leq \max (1, 1+ 2 \sin  \theta \cos\theta)\\
 &= \max (1,  1+\sin(2\theta))\\
 &\leq 2.\\
 \end{align*}

\end{proof}

\begin{proof}[Proof of \Cref{th:sparse_LR}]

We follow the same proof structure (cf \Cref{sec:proofslevels})
as for Theorem~\ref{th:opt_in_levels}, with analog definitions of $P_1(z)=z_{1}$ and $P_2(z)=z_{2}$ for $z =(z_1,s_2) \in \sH$. We also denote $T=(S_1,S_2) = T(z)$ where $S_1$ denotes a support of the $k$ largest coordinates in absolute values and $S_2 =\{1,\ldots, r\}$ the set indexing the first $r$ (largest) eigenvalues collected in vector $\mathtt{eig}(u)$
(this index set was denoted $T$ in Appendix~\ref{sec:proofLRnex}). We modify accordingly the notation  $T_2 $ for the $2k$ (resp. $2r$) largest coordinates (resp eigenvalues).

Remark that Lemma~\ref{lem:opt_support_wl1_level} is still valid with $\|\cdot\|_w = w_1\|P_{1}(\cdot)\|_1 +w_2  \|P_{2}(\cdot)\|_*$.

This permits in turns to obtain  the expression (with $T' = T_2 \setminus T$)
\begin{equation}\label{eq:TmpOptSparseLR}
B_\Sigma(\|\cdot\|_w) = \sup_{z: z \neq 0, \|z_{T_2^c}\|_w + \|z_{T'}\|_w \leq \|z_{T}\|_w} \frac{\|z_{T_2^c}\|_\sH^2}{\|z_{T_2}\|_\sH^2}
\end{equation}
from the proof of Lemma~\ref{lem:charact_supBL2_wl1_level-1}.
Now remark that this is exactly the same expression as in the sparsity in levels case using the vector of ordered eigenvalues for the part in $\sH_p$: $\|(u_1,u_2)\|_\sH^2 = \|u_1\|_2^2 + \|\eig (u_2)\|_2^2$ and  $\|(u_1,u_2)\|_w = w_1\|u_1\|_1 + w_2\| \eig  (u_2)\|_1$.  This in turns show that (using the same proof as Lemma~\ref{lem:charact_supBL2_wl1_level-1})
\begin{equation}\label{eq:TmpOptSparseLR2}
B_\Sigma(\|\cdot\|_w) =  \max_{
0\leq L_1 \leq n-2k, 0\leq L_2 \leq p-2r} B^{L_1,L_2}(w)
\end{equation}
where $B^{L_{1},L_{2}}(w)$ is defined in \eqref{eq:DefBL1L2}.  This is the first  step of Theorem~\ref{th:opt_in_levels}.

The rest of the proof then exactly matches the next steps of the proof of  Theorem~\ref{th:opt_in_levels}.

\end{proof}

\end{document}